\documentclass[final]{article}

\usepackage{amsfonts,amssymb,amsmath,amsthm}

\usepackage{natbib}
\usepackage{setspace}
\usepackage{fancyhdr}
\usepackage[paperwidth=8.5in,left=1in,right=1in,top=1in,bottom=1in,paperheight=11in]{geometry}
\usepackage{titlesec}
\usepackage{color}
\usepackage{array}
\usepackage{booktabs}
\usepackage{threeparttable}
\usepackage{url}
\usepackage{array}
\usepackage{caption}
\usepackage[colorlinks=true, allcolors=black]{hyperref}

\usepackage{fourier}
\usepackage[T1]{fontenc}

\usepackage{subcaption}
\usepackage{graphicx}
\usepackage{tikz}
\usetikzlibrary{shapes.geometric, arrows.meta, positioning}
\usepackage{pgfplots}
\pgfplotsset{compat=1.18}
\usetikzlibrary{patterns}
\usepgfplotslibrary{groupplots,fillbetween}

\definecolor{darkblue}{rgb}{0,0,0.55}
\definecolor{darkred}{rgb}{0.5,0,0}
\newcommand\Bheadfont{\fontsize{14pt}{\baselineskip}\selectfont}

\definecolor{darkblue}{rgb}{0,0,0.55}
\definecolor{darkred}{rgb}{0.5,0,0}

\titleformat{\section}[hang] {\normalfont\sc\color{darkblue}\Bheadfont} {\thesection\hskip0.618em}{0em}{}
\titlespacing*{\section}
{0pt}{15pt plus 2pt minus 2pt}{9pt plus 2pt minus 2pt}
\titleformat{\subsection}[runin]
{\normalfont\sc\color{darkblue}} {\thesubsection\hskip0.618em}{0em}{}
\titlespacing*{\subsection}
{0pt}{13pt plus 2pt minus 2pt}{13pt plus 2pt minus 2pt}
\titleformat{\subsubsection}[runin]
{\normalfont\sc\color{darkblue}} {\thesubsubsection\hskip0.618em}{0em}{}
\titlespacing*{\subsubsection}
{0pt}{13pt plus 2pt minus 2pt}{13pt plus 2pt minus 2pt}

\theoremstyle{plain} \newtheorem{proposition}{Proposition}

\begin{document}

	\title{\Large{Beneath the VIX: Probability, Severity, and Uncertainty Shocks}\thanks{I thank Zheng Liu, Nick Bloom,
	Venky Venkateswaran and participants at the Stanford SITE session on the Macroeconomics of Uncertainty and Volatility for comments.}}
	\author{Todd B. Walker\thanks{Department of Economics, Indiana University, \texttt{walkertb@iu.edu}.}}
	\date{October 2026}
	\maketitle

	\vspace{0.7in}
\begin{abstract}
	Implied stock-market volatility is widely used to identify uncertainty
	shocks and discipline models of their macroeconomic effects.
	This paper gives distributional content to the option-price variation
	underlying these exercises.
	A normalized at-the-money forward put price, the Downside Protection
	Premium (DPP), closely tracks the VIX and preserves the principal
	macroeconomic responses in the specifications examined.
	Its price and local strike slope separate risk-neutral shortfall
	probability from conditional severity.
	Severity accounts for almost all monthly premium comovement and
	91.5 percent of the recession log premium.
	Conditional gains also rise substantially, while conditional losses
	remain larger in phase averages.
	Despite this severity dominance, quotations across strikes frequently
	exclude a common-scale explanation, conditional on the maintained
	forward and premium construction.
	Realized-volatility and future-volatility-news innovations both raise
	protection costs primarily through severity, yet produce different
	estimated employment responses in common impact-severity units;
	simultaneous significance depends on the inference formulation.
	Together, these findings supply joint empirical restrictions on the
	connection between uncertainty, option prices, and economic activity.
\end{abstract}

	\bgroup
	\let\footnoterule\relax 
	\thispagestyle{empty} 
	\egroup
	\vspace{0.25in} 
	\phantom{s} 

	\onehalfspacing

	\newpage

	\setcounter{page}{1}
	\thispagestyle{empty}

	\section{Introduction}
\label{sec:intro}

Volatility indices such as the VIX and VXO are widely used proxies for
market uncertainty in macroeconomics.
These measures provide observables from which uncertainty innovations are identified
and targets used to discipline quantitative macroeconomic models of uncertainty
transmission.
Among other results, a VXO or VIX shock has been shown to produce a sharp drop and rebound in
output and employment within about six months \citep{Bloom2009}; lower
output, consumption, investment, and hours together \citep{BasuBundick2017};
raise unemployment and lower inflation, much like a negative aggregate
demand shock \citep{LeducLiu2016}; raise unemployment by substantially more
in recessions than in expansions \citep{CaggianoCastelnuovoGroshenny2014};
and have larger contractionary effects when the zero lower bound binds
\citep{CaggianoCastelnuovoPellegrino2017}.
However, \citet{BergerDewBeckerGiglio2020SecondMomentNews} find that the contraction is associated
with realized-volatility shocks, while news about future volatility has
little effect on activity.
This paper similarly looks beneath the volatility index.
I construct the price of a one-month at-the-money forward put, normalized
by the index level, and show that this downside protection premium moves
almost perfectly with the VIX.
Its price equals the product of the risk-neutral probability of a
shortfall and the expected severity of that shortfall.
I show that severity accounts for nearly all of the premium's movements.
Moreover, I document that innovations to realized volatility and news
about future volatility both raise protection costs mainly through
severity, yet they have different estimated effects on employment.

I begin with the price of an at-the-money forward put divided by the current
index value, which I call the Downside Protection Premium (DPP).
This price measures the cost of protection against the index finishing below
its forward value.
Its monthly correlation with the VIX is 0.992, and substituting the DPP for
the VIX or VXO as a proxy for uncertainty preserves the principal macroeconomic responses
well-known settings \citet{Bloom2009,LeducLiu2016,BasuBundick2017}.
The DPP, therefore, captures the market-price variation
underlying familiar macroeconomic uncertainty exercises.

The DPP permits a direct interpretation of that variation.
After adjusting for carrying costs, the premium equals the risk-neutral
probability of finishing below the forward (i.e., probability of shortfall)
multiplied by the expected proportional shortfall conditional (i.e., conditional severity)
on doing so.
The DPP pins down only the product, but puts with nearby strikes
can be used to separate the two because raising the strike slightly increases the put's payoff
only when the index finishes below it, so the resulting change in price
reveals the probability of a shortfall.
Thus, simple put prices do more than track the VIX,
these prices also support a familiar insurance interpretation of its movements.
When the cost of protection rises, is a loss more likely, or would it be
more severe?

Thus, the paper is able to address three questions.
First, which component accounts for movements in the protection premium?
Second, can that composition be explained by a change in scale alone,
that is, by a distribution of index outcomes that simply widens or
narrows around the forward without changing its shape?
Such a pure widening is a natural benchmark because it keeps the
shortfall probability fixed and moves the premium entirely through
severity, while also fixing how the prices of puts at other strikes move
relative to the at-the-money price.
Third, do volatility innovations that raise the premium through a
similar mix of probability and severity have similar macroeconomic
consequences?
Together, these questions turn the decomposition into empirical
restrictions on the connection between uncertainty, option prices, and
economic activity.

My first finding is that conditional severity accounts for almost all
monthly premium comovement and most of the recession premium.
Using daily SPX option quotations from January 1990 through July 2026,
I estimate a monthly severity loading of 1.032 and a complementary
probability loading of $-0.032$.
The probability loading's 95 percent interval, $[-0.071,0.006]$, includes zero,
and probability changes are substantially less variable than severity changes.
Severity accounts for 91.5 percent of the mean log premium difference
between recessions and expansions, although the recession probability
contribution is positive.
These conclusions persist across time periods, recession comparisons,
and alternative component-recovery methods.
Quarterly-horizon options give the same broad pattern of severity dominance.

The interpretation is two-sided due to put-call parity.
Because the forward is the market's price-weighted average of where the
index will finish, expected shortfalls below the forward must be offset
by expected gains above it.
By put--call parity, the adjusted premium therefore also equals the
probability of finishing at or above the forward multiplied by the
expected gain conditional on doing so, which I call conditional gain
severity.
The monthly gain-severity loading is 0.977, and conditional gains rise
substantially during recessions.
Higher protection costs thus accompany larger conditional movements on
both sides of the forward, rather than an exclusively downside movement.
The probability determines the relative size of these movements.
The shortfall probability averages 0.423 in expansions and 0.444 in
recessions---shortfalls are priced as less likely than gains, so for
expected losses and gains to balance, a shortfall must be larger than a
gain when it occurs.
The conditional loss is roughly 36 percent larger
than the conditional gain in expansions and 25 percent larger in
recessions.
The priced distribution thus describes frequent modest gains and less
frequent, larger losses.
This asymmetry is information that the dominance of severity does not
convey.
A lognormal distribution, the benchmark behind Black--Scholes, also
moves its premium almost entirely through severity; doubling the premium
from two to four percent raises severity from 3.92 to 7.69 percent while
the probability barely moves, from 0.51 to 0.52.
But it implies the opposite asymmetry.
Its outcomes are skewed to the right, with a long upper tail, so the
median lies below the mean and more than half the outcomes fall short of
the forward.
A lognormal therefore prices shortfalls as more likely than gains but
smaller, whereas the data price them as less likely but larger.
Thus, matching the volatility index, or even the dominance of severity,
is not sufficient.
The data clearly want a model that prices shortfalls as less likely than gains but larger when they occur,
which rules out a lognormal distribution.

The second finding is that the dominance of severity does not mean that
option prices simply reflect a widening distribution.
If the distribution of index outcomes only widens
or narrows around the forward without changing its shape, the shortfall
probability stays fixed and the premium moves one-for-one with severity.
Such a widening has a testable implication.
Once the at-the-money premium reveals how wide the distribution is on a
given date, the prices of puts at every other strike follow mechanically.
Within each month, I therefore ask whether any single shape, widened or
narrowed day by day, can match the quoted prices across strikes to within
their bid--ask spreads.
In 352 of 405 months, no such shape exists.
A weaker version asks only whether the shape of losses, given that a
shortfall occurs, stays fixed, allowing the shortfall probability and the
distribution of gains to change freely from day to day.
Even this weaker version fails in 112 months, or 27.7 percent of the
sample.
These counts take the forward and premium construction as given; they
are not statistical rejection rates, and some of the discrepancies are
small.
Nonetheless it shows that the at-the-money price can move almost
entirely through severity while the shape of the priced distribution
changes in ways that a pure widening cannot explain.
Thus, for a quantitative model disciplined by implied volatility,
matching the volatility index is not sufficient; its equilibrium option
prices must also generate changes in the shape of the priced
distribution, whether through nonlinear payoffs, state-dependent
valuation, or the uncertainty shock itself.

The third finding connects the composition of protection prices to
differences in macroeconomic responses.
\citet{BergerDewBeckerGiglio2020SecondMomentNews} distinguish innovations
to realized volatility from news about future volatility and find
contrasting responses of economic activity.
I extend this comparison to the probability and severity components
of the protection premium.
Both innovations raise protection costs almost entirely through
conditional severity, yet manufacturing employment falls more after
the realized-volatility innovation.
To put these employment responses on a common scale, I normalize
each innovation to raise conditional severity by one percent on impact.
Under this normalization, the estimated manufacturing employment
response twelve months later is about 0.065 percentage points lower
after a realized-volatility innovation than after volatility news.
The employment contrast therefore survives equalizing the impact
severity response.
This comparison supplies a joint empirical target for models of
uncertainty transmission---both innovations raise the severity of
priced shortfalls, but the estimated employment responses differ.

\subsection{Relation to Existing Literature}

The paper primarily contributes to the macroeconomic literature that uses implied
volatility to identify uncertainty shocks and discipline models of their effects
\citep{Bloom2009,LeducLiu2016,BasuBundick2017}.
I show that replacing the VIX or VXO with the DPP in the empirical
specifications of these papers preserves their main conclusions about how
uncertainty shocks affect economic activity.
Unlike the VIX, the DPP can be decomposed which I exploit to show that most of the
option-price variation attributable to severity rather than probability.

Other papers decompose uncertainty into components with different economic
effects.
\citet{SegalShaliastovichYaron2015GoodBadUncertainty} find that good
uncertainty about macroeconomic growth predicts higher future activity and asset
valuations, whereas bad uncertainty predicts lower ones.
\citet{ForniGambettiSala2025} find that widening the left tail of the
GDP-growth forecast distribution is contractionary, whereas widening the right
tail is mildly expansionary.
My decomposition differs in both object and components, and provides
additional facts that structural models must explain.

The closest comparison is
\citet{BergerDewBeckerGiglio2020SecondMomentNews}.
Their distinction between realized-volatility innovations and news about
future volatility has different implications for economic activity.
I add the priced composition of the associated protection-price responses
and compare activity in common impact-severity and protection-premium units.
Both innovations produce severity-dominated responses, yet their estimated
employment responses differ in these units.
The additional evidence directs structural interpretation toward the origins,
subsequent paths, and other equilibrium effects of the innovations.
Related measurement work connects at-the-money option statistics to implied
volatility.
\citet{DewBeckerGiglio2023} use at-the-money implied volatility in their
analysis of cross-sectional uncertainty.
\citet{Dotsis2024} constructs an option-implied absolute-deviation index and
studies its relationship with the VIX and the predictive content of their
spread for equity returns.
The present paper connects the at-the-money protection price to macroeconomic
uncertainty identification and model evaluation.
It examines the composition of the price movements underlying standard
uncertainty exercises, their compatibility with common-scale restrictions,
and their association with economic activity.
The contribution consequently extends beyond the close correspondence between
an at-the-money statistic and the VIX.

The distributional analysis is also related to work on option-implied skewness.
\citet{DewBecker2024} finds that aggregate option-implied skewness is approximately
acyclical, while firm-level skewness is procyclical.
The local probability gap and conditional loss--gain ratio studied here are
distinct from a global standardized third moment.
On identical quarter-end observations, the recession change in the
premium-adjusted local probability gap is larger for 30-day than for
quarterly-horizon options.
This horizon dependence and the joint quotation restrictions address
properties of the priced distribution that global skewness alone does not determine.
The analysis does not reinterpret the local gap as a new estimator of global
skewness or claim that its behavior overturns the aggregate skewness finding.

Options have also been used to assess disaster mechanisms and the pricing
of adverse states. \citet{BackusChernovMartin2011} infer disaster implications
from index options, while \citet{SeoWachter2019} study option prices in a
model with stochastic disaster risk. \citet{BekaertEngstrom2017} model
non-Gaussian fundamentals and valuation, and \citet{ChabiYoLoudis2024}
decompose conditional risk premia across downside, upside, and central regions.
\citet{CulpGandhiNozawaVeronesi2021} relate option-implied spreads
to risk-neutral probabilities of finishing below the strike and
conditional losses given that event, and use these measures to
study option risk premia.
The present exercise concerns the composition of protection costs and their
connection to empirical uncertainty innovations; it does not separately
recover physical probabilities and risk premia.
To the best of my knowledge, this paper is the first to jointly document
three features of the option-price variation underlying macroeconomic
uncertainty exercises: conditional shortfall severity accounts for
nearly all monthly protection-premium comovement; quotations exclude
within-month common scale in most months; and realized-volatility
innovations and volatility news have different estimated employment
responses even when normalized to the same impact.

\section{The Downside Protection Premium}
\label{Sec:AtM}

Consider an at-the-money forward (AtM-forward) put option, defined by a
strike $K$ equal to the no-arbitrage forward price $F_{t,t+\tau}$,
where the forward definition accounts for the cost of carrying the
underlying to expiration.
Let $r_t$ and $d_t$ denote the continuously compounded risk-free rate and
dividend yield, held constant over the option horizon. Then
$F_{t,t+\tau}=S_t\exp[(r_t-d_t)\tau]$,
where $S_t$ is the current underlying price and $\tau$ is time to maturity in
years.
Let $P_t^{\mathrm{AtM}}(\tau)$ denote the price of this European put.
I define the Downside Protection Premium (DPP) as its price divided by
the current underlying price\footnote{The close empirical DPP--VIX correspondence established
	below also holds with simpler at-the-money constructions.
	Holding quote dates, expirations, and interpolation in time fixed,
	interpolating put prices to the current index level, $K=S_t$, yields
	a nearly identical monthly relationship with the VIX.
	Even using the quoted put strike closest to spot without interpolation
	across strikes works nearly as well.
	Thus, the close monthly relationship does not require the forward
	adjustment. It is retained throughout for the exact forward-centered
	probability--severity decomposition and loss--gain identities below.}
\begin{equation}
	\chi_t(\tau)\equiv\frac{P_t^{\mathrm{AtM}}(\tau)}{S_t}
	\label{Eqn:AtMPut}
\end{equation}
For much of the analysis, I focus on 30 calendar days,
$\tau=30/365$, and write $\chi_t^{30}\equiv\chi_t(30/365)$ to match the horizon of the VIX.

The DPP measures the cost of downside protection per dollar of current equity
exposure and has an intuitive insurance interpretation---it is the cost of one month
of protection against the index finishing below its forward price per dollar of current equity exposure.
An investor holding \$100 of the index at date $t$ pays $100\chi_t^{30}$ dollars for
the corresponding put. Equivalently, the premium is $100\chi_t^{30}$ percent of
current equity exposure.
Moreover, the ratio is dimensionless, allowing comparisons across time and assets.

\subsection{Constructing a Synthetic AtM-Forward Put Price}
\label{Sec:DPPConstruction}

I use daily end-of-day SPX option quotes from DeltaNeutral
(HistoricalOptionData.com) to construct a daily time series
of the Downside Protection Premium, $\chi_t^{30}$.
The data contain puts and calls across a range of strikes and expirations
from January 1990 through July 2026.
I impose neither a minimum trading volume nor a maximum relative bid--ask
spread.
I price options at the midpoint of the bid-ask spread
and retain quotes with a nonnegative bid, a strictly positive ask, and
an ask no lower than the bid.
For put-call parity and additional calculations below, I match
put and call contracts within the same recorded time frame
by strike, expiration, and identified contract family and settlement terms.
Note that SPX options have European exercise, avoiding the need for an early-exercise adjustment.
Appendix~\ref{app:data} provides additional details and documents a substantial
increase in volume around 2005, with put volume exceeding calls,
and with most put volume occurring near-the-money.

In order to estimate a synthetic put price at the forward strike
and a constant 30-day horizon, I use two linear interpolations
designed to be simple by construction.
The first targets the forward strike separately at each maturity bracketing
the 30-day horizon, using
$
	F_t(T;K)=K+e^{r_tT}\bigl[C_t(T,K)-P_t(T,K)\bigr]$.
For both short and long durations, I take the median implied forward from up
to three matched strikes closest to the contemporaneous index level and denote it by $F_t(T)$.
The current implementation uses a daily three-month Treasury bill rate,
converted to continuously compounded units with past rates carried forward for at most three
calendar days.
Using parity to infer the forward incorporates the pricing of dividends
without requiring a separate dividend series; see \cite{vanBinsbergen2012}.
Let $K^-<F_t(T)<K^+$ be the nearest eligible put strikes below and above
the targeted forward.
The interpolated AtM-forward put price for that expiration is therefore
\begin{equation}
	\widehat P_t^{\mathrm{AtM}}(T)
	=\left(\frac{K^+-F_t(T)}{K^+-K^-} \right) P_t(T,K^-)
	+\left(\frac{F_t(T)-K^-}{K^+-K^-}\right) P_t(T,K^+)
	\notag
\end{equation}
I limit the total width of the strike bracket to five percent of
the forward
$
	(K^+-K^-)/F_t(T)\leq 0.05;
$
this restriction drops 6 days.

The second interpolation pins down maturity by selecting eligible maturities
immediately above and below the target duration, $T_s<\tau<T_l$.
The synthetic price is the combination
\begin{equation}
	\widehat P_t(\tau)
	=\left(\frac{T_l-\tau}{T_l-T_s}\right) \widehat P_t^{\mathrm{AtM}}(T_s)
	+\left( \frac{\tau-T_s}{T_l-T_s}\right) \widehat P_t^{\mathrm{AtM}}(T_l)
	\label{Eqn:AtMConstruct}
\end{equation}
I impose no additional restriction on the distance between the bracketing maturities.
If the required strike or maturity brackets are unavailable, the daily
observation remains missing.
The procedure neither extrapolates nor fills missing dates.
The construction is intentionally simple with
the final price using a convex combination of four put prices closest to
the desired AtM strike and 30-day duration.

\begin{figure}[t]
	\centering
	\begin{tikzpicture}[x=1cm,y=1cm,>=Stealth,font=\small]
		\draw[->,thick] (0,-1.65)--(0,1.9)
		node[above,align=center] {Strike relative to forward};
		\draw[->,thick] (0,-1.65)--(9.0,-1.65)
		node[right] {$T$};
		\draw[dashed,gray] (0,0)--(8.8,0);
		\node[left] at (0,0) {$K-F_t(T)=0$};
		\draw[thick] (1.7,-1.1)--(1.7,1.1);
		\draw[thick] (7.4,-.85)--(7.4,1.3);
		\draw[densely dotted] (4.3,-1.65)--(4.3,1.5);
		\foreach \x/\y in {1.7/-1.1,1.7/1.1,7.4/-.85,7.4/1.3}
		{\fill (\x,\y) circle (2pt);}
		\node[right] at (1.7,1.1) {$P_t(T_s,K_s^+)$};
		\node[right] at (1.7,-1.1) {$P_t(T_s,K_s^-)$};
		\node[left] at (7.4,1.3) {$P_t(T_l,K_l^+)$};
		\node[left] at (7.4,-.85) {$P_t(T_l,K_l^-)$};
		\draw[thick,->] (1.85,0)--(4.1,0);
		\draw[thick,->] (7.25,0)--(4.5,0);
		\foreach \x in {1.7,7.4}
		{\draw[fill=white,thick] (\x,0) circle (2.5pt);}
		\fill (4.3,0) circle (2.5pt);
		\node[above] at (4.3,.1) {$\widehat P_t(\tau)$};
		\node[below] at (1.7,-1.65) {$T_s$};
		\node[below] at (4.3,-1.65) {$\tau$};
		\node[below] at (7.4,-1.65) {$T_l$};
	\end{tikzpicture}
	\caption{Constructing the synthetic AtM-forward put price. At each
		bracketing maturity, the two observed put prices (solid endpoints)
		are interpolated across strikes to obtain a price at the forward
		(open circles). Interpolation between these two prices gives the
		estimate at the target horizon (central solid circle). Vertical
		positions represent strike distances from the maturity-specific
		forward.}
	\label{fig:AtMConstructionGraphic}
\end{figure}
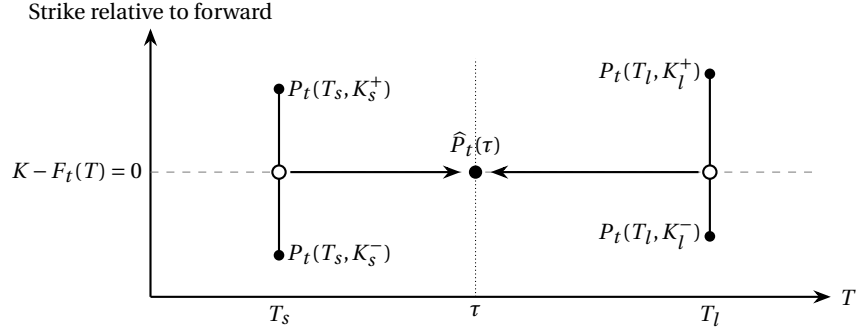

The estimated Downside Protection Premium, $\hat \chi_t^{30}$, is available on 8,825 trading days from January 2, 1990 through July 31, 2026, covering 95.80 percent of the
9,212 regular trading sessions in the sample.\footnote{Of the 387 sessions
	without a constructed premium, two have no source SPX quotations, 379 lack a
	usable maturity bracket for the 30-day construction, and six are due to the
	5 percent strike-bracket restriction.
	Of these 379 maturity failures, 359 occur between the
	14th and 21st of the month and coincide with dates when the next standard monthly
	expiration is more than 30 days away.
	The remaining 20 occur between September 18 and October 15, 1992.}
I aggregate the daily DPP to monthly frequency by taking the arithmetic
average of available daily observations within each month.
The resulting series contains 439 monthly observations with no missing
months from January 1990 through July 2026.
Panel A of Table \ref{Tab:DPPVIXCorrelations} reports monthly summary statistics.
The mean monthly DPP is 1.88 percent, corresponding to a premium of
approximately \$1.88 for 30 days of protection per \$100 of current equity
exposure.
The median monthly premium is substantially lower than the mean at 1.68 percent.
Hence, there is substantial right skewness (2.039).
The standard deviation of the monthly series is 0.77 percentage points,
compared with 0.81 percentage points for the daily series; monthly averaging
leaves the mean premium nearly unchanged.
The monthly series is persistent with a first-order autocorrelation
of 0.853.

\subsection{Connection to the VIX and Uncertainty Shocks}
\label{Sec:VIXConnection}

Panels~C and~D of Table \ref{Tab:DPPVIXCorrelations} establish the empirical connection motivating the DPP for the subsequent analysis of uncertainty shocks.
\emph{The DPP and the VIX are almost perfectly correlated at every sampling frequency and in both levels and changes.}
The Pearson correlation is 0.990 across 8,822 matched daily observations and 0.992 across 439 monthly averages; the Spearman correlations are nearly as high (0.984 and 0.988), which, together with the expansion-month correlation reported below, indicates that the association extends beyond crisis observations.
Correlations in monthly log changes are also quite high at 0.984, so the relationship is not an artifact of common trends or shared persistence in levels.
The comovement is stable over time and across business-cycle phases.
Recession months are more highly correlated 0.995 than expansions 0.988, so the two measures are more strongly
coupled when uncertainty is high.
Rolling 24-month correlations never fall below 0.950 in levels or 0.921 in log changes, with medians of 0.991 and 0.987.
The one visibly lower reading is the 12-month minimum of 0.716 in levels.

Panel~B regresses the DPP in basis points of exposure on the VIX, which yields a slope of 10.33 (HAC standard error 0.12) in levels and 10.13 (0.14) in first differences, with $R^2$ of 0.985 and 0.980.
Thus, a one-point increase in the VIX is associated with an increase in the cost of one month of protection of roughly 10 basis points of exposure, and the near-identical slopes support an approximately linear relationship in levels and changes.
The small negative intercept in levels ($-13.2$ basis points) and zero counterpart in changes ($-0.01$) are consistent with an approximately linear relationship.

\begin{table}[htbp]
	\centering
	\normalsize
	\renewcommand{\arraystretch}{0.95}
	\setlength{\tabcolsep}{4pt}
	\caption{The 30-Day DPP: Summary Statistics and Relationship with the VIX}
	\label{Tab:DPPVIXCorrelations}

	\begin{tabular}[t]{@{}lrrrrrr@{}}
		\toprule
		\multicolumn{7}{@{}l}{\textbf{A. Summary statistics}}\\
		\addlinespace[2pt]
		Mean & SD & Min & Med. & Max & Skew. & AR(1)\\
		\midrule
		1.88 & 0.77 & 0.81 & 1.68 & 6.35 & 2.04 & 0.85\\
		\addlinespace[2pt]
		\bottomrule
	\end{tabular}
	\hfil
	\begin{tabular}[t]{@{}lrrrr@{}}
		\toprule
		\multicolumn{5}{@{}l}{\textbf{B. Regression analysis}}\\
		\addlinespace[2pt]
		Series & Const. & Slope & $R^2$ & RMSE\\
		\midrule
		Levels  & $-13.2$ & 10.33 & 0.985 & 9.45\\
		& (2.21)  & (0.12) &      &      \\
		Changes & $-0.01$ & 10.13 & 0.980 & 5.89\\
		& (0.15)  & (0.14) &      &      \\
		\bottomrule
	\end{tabular}

	\vspace{8pt}

	\begin{tabular}[t]{@{}lrrr@{}}
		\toprule
		\multicolumn{4}{@{}l}{\textbf{C. Correlations with VIX}}\\
		\addlinespace[2pt]
		Comparison & Obs. & Pearson & Spear.\\
		\midrule
		Daily levels & 8,822 & 0.990 & 0.984\\
		Monthly levels & 439 & 0.992 & 0.988\\
		Monthly log chg. & 438 & 0.984 & 0.979\\
		\addlinespace[2pt]
		Recessions & 36 & 0.995 & 0.971\\
		Expansions & 403 & 0.988 & 0.986\\
		\bottomrule
	\end{tabular}
	\hfil
	\begin{tabular}[t]{@{}lrrrr@{}}
		\toprule
		\multicolumn{5}{@{}l}{\textbf{D. Rolling correlations}}\\
		\addlinespace[2pt]
		Series & Window & Min & Med. & Max\\
		\midrule
		Levels   & 12 & 0.716 & 0.991 & 1.000\\
		Log chg. & 12 & 0.856 & 0.987 & 0.999\\
		Levels   & 24 & 0.950 & 0.991 & 0.999\\
		Log chg. & 24 & 0.921 & 0.987 & 0.996\\
		\bottomrule
	\end{tabular}

	\begin{flushleft}\footnotesize
			\textit{Note:} Sample 1990-01-02 to 2026-07-31; 8,825 daily
			observations averaged into 439 calendar months. Panel~A reports monthly
			summary statistics for $100\chi_t^{30}$, the 30-day put premium as a percent of exposure;
			skewness is bias-adjusted and AR(1) is the adjacent-month
			autocorrelation. Panels~B--D align DPP and the VIX on 8,822 common
			daily dates before averaging (3 DPP dates lack VIX quotes). Panel~B
			estimates $10{,}000\chi_t^{30}=\alpha+\beta\,\mathrm{VIX}_t+u_t$ and its
			first difference. Intercepts and RMSEs are in basis points of exposure,
			and slopes are in basis points of exposure per VIX point; parentheses give Bartlett HAC standard errors
			(12 lags). Panel~C reports unitless correlations and splits months by NBER phase, with recessions running
			from the month after a peak through the trough. Panel~D uses
			overlapping windows of consecutive months.
	\end{flushleft}
\end{table}

Table \ref{Tab:DPPVIXCorrelations} shows that the DPP and VIX share most
of their observed time-series variation.
Appendix~\ref{App:MacroReplication} examines the implications of substituting
the DPP for the VIX/VXO in three influential papers,
\citet{Bloom2009,LeducLiu2016,BasuBundick2017}.
Swapping the VIX/VXO for the DPP preserves the principal macroeconomic responses to
uncertainty shocks in these papers, although event dates and some inference
classifications based on subjective metrics differ (see, Appendix A.2).
The aggregate time-series information in the DPP thus closely matches that of the VIX/VXO.

\subsection{Intuition}
\label{sec:dpp_vix_intuition}

Why is the DPP so strongly temporally correlated with the VIX?
The close correspondence has a natural historical antecedent.
The original VIX, developed by \citet{Whaley1993}, combined the implied
volatilities of eight near-the-money S\&P~100 options to represent a
30-day AtM option.
As noted above, the DPP construction uses four near-the-money put prices
and shares this emphasis on near-the-money options.\footnote{The original VIX,
	subsequently designated VXO, used American OEX options, defined AtM
	relative to the current index, and interpolated implied volatilities.
	The DPP uses European SPX options, defines AtM relative to the forward,
	and interpolates prices.}
However, the VIX was redesigned in 2003 to extend to a broader range of strikes
to capture more tail behavior \citep{Whaley2009}.
The question is how much this broader strike coverage changes the
fluctuations already captured by an AtM price.

Black--Scholes formulas offer intuition.
For an exact forward-AtM put, its implied volatility
$\sigma_{\mathrm{AtM},t}(\tau)$ satisfies
\begin{equation}
	\chi_t^\tau
	=
	e^{-d_t\tau}
	\left[
	2\Phi\!\left(
	\frac{\sigma_{\mathrm{AtM},t}(\tau)\sqrt{\tau}}{2}
	\right)-1
	\right]
	\approx
	\frac{e^{-d_t\tau}}{\sqrt{2\pi}}\,
	\sigma_{\mathrm{AtM},t}(\tau)\sqrt{\tau}
	\notag
\end{equation}
where $\Phi$ is the standard normal distribution function and $d_t$ is
the dividend yield. This is an implied-volatility representation of the
option price, not an assumption that the actual risk-neutral distribution is lognormal.
At $\tau=30/365$ and $d_t=0.02$, this gives
$100 \cdot \sigma_{\mathrm{AtM},t}\approx876\chi_t^{30}$.
Thus, the DPP is approximately proportional to AtM implied volatility,
and one can assess the empirical departure from proportionality through the relative
gap,
$
	R_t
	=
	\mathrm{VIX}_t / (876\chi_t^{30})$,
and
$
	\Delta\log\mathrm{VIX}_t
	=
	\Delta\log\chi_t^{30}
	+
	\Delta\log R_t
	\notag
$.
Using matched monthly levels, the median ratio is 1.187.
The standard deviation of $\Delta\log R_t$ is 0.0312, compared with
0.1614 for $\Delta\log\mathrm{VIX}_t$.
The relative gap thus changes much less than the VIX itself, consistent
with a dominant common component.
This evidence, along with Figure \ref{fig:atm_time_series} (which plots the approximation $876\chi_t^{30}$), supports approximate proportionality.\footnote{Related evidence appears
	in \citet{Dotsis2024}, whose option-implied absolute-deviation index has a monthly level correlation of 0.990 with VIX over 1996--2021.}

\begin{figure}[t]\centering
\includegraphics[width=\linewidth]{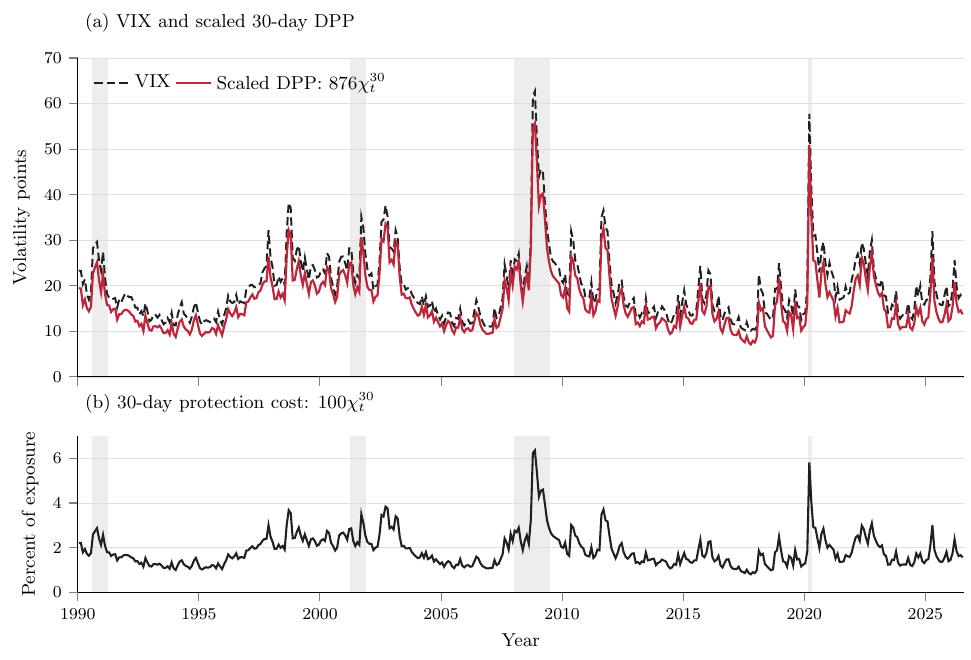}
\caption{Monthly 30-day downside protection premium and VIX. Panel (a) compares the VIX with the fixed benchmark $876\chi_t^{30}$ in volatility points. Panel (b) reports the same DPP in percent of current equity exposure, $100\chi_t^{30}$. Both series are arithmetic averages over identical daily dates. Shading marks NBER recession months, from the month after the peak through the trough. The scale 876 is fixed, not fitted to the VIX. The daily sample is 1990-01-02 through 2026-07-31, with 8,822 matched dates and 439 observed months. }\label{fig:atm_time_series}
\end{figure}

Beyond the Black-Scholes algebra, there is a common-scale interpretation that will
be revisited throughout the rest of the paper.
Let $X = S_{t+\tau}/F_{t,t+\tau}$ denote the forward-normalized terminal index value, so
that $E^\mathbb Q_t[X] = 1$ under the maintained carrying-cost assumptions.
The DPP and an idealized single-maturity log-contract volatility measure
are then determined by expectations of functions of the single variable $X$
\begin{equation}
	\chi_t^\tau
	=
	\frac{e^{-d_t\tau}}{2}\,
	E_t^{\mathbb Q}\!\left[\lvert X-1\rvert\right]
	\qquad
	v_t(\tau)^2
	=
	-\frac{2}{\tau}\,
	E_t^{\mathbb Q}\!\left[\log X\right]
	\notag
\end{equation}
where $v_t(\tau)$ denotes log-contract volatility,
distinct from the observed VIX with its discrete strikes and maturity interpolation.
The first identity expresses the DPP as one half of the discounted
risk-neutral mean absolute deviation of $X$.
Equivalently, by put--call parity, twice the DPP is the normalized
at-the-money-forward straddle price.
The second is the log-contract representation, where $v_t(\tau)^2$ equals
$2/\tau$ times the risk-neutral entropy
$\log E_t^{\mathbb Q}[X]-E_t^{\mathbb Q}[\log X]$.
The latter is recovered from a continuum of strikes at a single maturity
\citep{CarrWu2009Variance}.
To isolate the role of scale, hold carrying costs fixed and write
$X=1+s_tZ$,
where $Z$ has a fixed risk-neutral distribution with bounded support,
mean zero, and unit variance, and $s_t>0$ is small enough that $X>0$
almost surely.
The DPP is then \emph{exactly} linear in $s_t$
$
	\chi_t^\tau
	=
	0.5 \;  e^{-d_t\tau}
	s_t\,E^{\mathbb Q}\!\left[\lvert Z\rvert\right].
$
The log-contract volatility measure is linear only to first order.
Expanding the log contract as $s_t$ approaches $0$,
\begin{equation}
	v_t(\tau)
	=
	\frac{s_t}{\sqrt{\tau}}
	\left(
	1-\frac{\mu_3}{3}\,s_t+O\!\left(s_t^2\right)
	\right)
	\qquad
	\mu_3\equiv E^{\mathbb Q}\!\left[Z^3\right]
	\notag
\end{equation}
Several points are noteworthy.
First, a pure change in scale moves both measures proportionally to
first order, whatever the shape of $Z$.
Second, skewness enters the first correction to proportionality,
but this correction depends jointly on $\mu_3$ and $s_t$.
The ratio of DPP to log-contract volatility therefore varies with scale even when skewness is
fixed, and higher moments enter at subsequent orders.
This is the sense in which the two measures share a first-order loading
on scale while remaining distinct.\footnote{This mechanism is consistent
	with the prominent level factor in implied-volatility surface variation
	documented by \citet{ContDaFonseca2002}.}
Finally, this is a static comparison, at a single maturity, under the
maintained risk-neutral pricing assumptions.
While it does not specify dynamics for $s_t$, there is a conditional dynamic
implication.
Because the two measures are approximately linear in $s_t$, persistent
variation in dispersion can generate persistent and highly correlated
time series when the higher-order corrections are small.
A common scale component is therefore one explanation for the dominant
shared variation in the VIX and DPP.
Their high correlation does not establish that the shape of the priced
distribution is unchanged; distributional shape and valuation can also vary.
The next section uses the local strike slope of put prices to distinguish
the risk-neutral probability of a shortfall from its conditional severity.
Put--call parity also gives the premium an upside representation, so this
interpretation concerns conditional movements on both sides of the forward.
The price and probability contain the same local information as the
AtM price and implied-volatility slope, while the decomposition expresses
that information in probability and conditional-magnitude units.

\section{Decomposing the DPP}
\label{Sec:DecomposeDPP}

The preceding section establishes that the DPP closely tracks the VIX
and preserves the principal macroeconomic responses in the specifications examined.
I now recover its probability and severity components and establish two results:
conditional severity accounts for nearly all monthly premium comovement
and most of the recession premium, but a \emph{common} expansion of the priced
distribution does not generally account for quotations across strikes.

\subsection{The Probability--Severity Decomposition}

Let $\tilde R_t^\tau=S_{t+\tau}/F_{t,t+\tau}$ denote the terminal index
value relative to its forward, and let $D_t^\tau=e^{-d_t\tau}$ denote
the carry factor.
Under the maintained pricing assumptions,
$\mathbb E_t^{\mathbb Q}[\tilde R_t^\tau]=1$, and the AtM-forward put price over spot satisfies
\begin{equation}
	\chi_t^\tau
	=D_t^\tau\,
	\underbrace{\mathbb Q_t(\tilde R_t^\tau<1)}_{q_t^\tau}
	\underbrace{\mathbb E_t^{\mathbb Q}[1-\tilde R_t^\tau
		\mid\tilde R_t^\tau<1]}_{m_t^\tau}.
	\label{eq:chi_summary_factor}
\end{equation}
Here $q_t^\tau$ is shortfall probability and $m_t^\tau$ is conditional
shortfall severity.
A shortfall means finishing below the forward, rather than crossing a
fixed disaster threshold.
Both components are risk-neutral quantities.
Removing carry gives the adjusted premium
$a_t^\tau\equiv\chi_t^\tau/D_t^\tau=q_t^\tau m_t^\tau$.
Its price alone identifies the product; a local option-price slope will
provide the additional information needed to recover the components.

The product identity gives a direct interpretation of premium changes.
Along a differentiable distributional path with a changing premium,
\begin{equation}
	\eta_q^\tau\equiv\frac{d\log q^\tau}{d\log a^\tau}\qquad
	\eta_m^\tau\equiv\frac{d\log m^\tau}{d\log a^\tau}\qquad
	\eta_q^\tau+\eta_m^\tau=1
	\label{eq:dpp_decomposition_elasticities}
\end{equation}
A unit severity elasticity means that conditional losses rise
proportionately with the premium while shortfall probability is unchanged.
The forward-mean restriction also equates expected losses and gains
\begin{equation}
	a_t^\tau=q_t^\tau m_t^\tau=(1-q_t^\tau)m_t^{+,\tau}\qquad
	\frac{m_t^\tau}{m_t^{+,\tau}}=\frac{1-q_t^\tau}{q_t^\tau}
	\label{eq:dpp_upside_identity}
\end{equation}
where $m_t^{+,\tau}$ is the expected proportional gain conditional on
finishing at or above the forward.
Thus $q_t^\tau<1/2$ implies larger conditional priced losses than gains.
This asymmetry expresses the same probability information in
conditional-loss units; it is distinct from third-moment skewness.

\subsection{Severity Dominance in Scale Benchmarks}

A distribution that expands around the forward without changing shape
provides the simplest explanation of severity dominance.

\begin{proposition}[Scale and conditional severity]
	\label{prop:dpp_scale_severity}
	Suppose $\tilde R_t^\tau=1+s_tZ$, where $s_t>0$, $Z$ has a fixed
	risk-neutral distribution with mean zero and unit variance, and
	$1+s_tZ>0$ almost surely at each date considered. Then
	\begin{equation}
		q_t^\tau=\Pr(Z<0),\qquad
		m_t^\tau=s_t\mathbb E[-Z\mid Z<0],\qquad
		a_t^\tau=s_t\mathbb E[(-Z)^+].
		\label{eq:dpp_scale_components}
	\end{equation}
	Shortfall probability is constant, and all proportional premium
	variation is attributable to conditional severity.
\end{proposition}

\begin{proof}
	Since $s_t>0$, the events $\tilde R_t^\tau<1$ and $Z<0$ coincide,
	and $1-\tilde R_t^\tau=-s_tZ$.
	Taking the probability, conditional expectation, and expected put
	payoff gives the three equalities.
\end{proof}

The distribution of $Z$ may be asymmetric, so the constant probability
need not equal one-half.
The restriction is fixed shape as scale changes.
Severity can also dominate when probability changes with scale.
In the mean-one lognormal benchmark,
$\tilde R_t^\tau=\exp(-v_t^2/2+v_t\varepsilon)$ with
$\varepsilon\sim N(0,1)$, probability and severity satisfy
\begin{equation}
	q_t^\tau=\frac{1+a_t^\tau}{2}\qquad
	m_t^\tau=\frac{2a_t^\tau}{1+a_t^\tau}\qquad
	\eta_q^\tau=\frac{a_t^\tau}{1+a_t^\tau}\qquad
	\eta_m^\tau=\frac{1}{1+a_t^\tau}
	\label{eq:dpp_lognormal_components}
\end{equation}
At premiums of four percent or less, the severity elasticity exceeds
96 percent.
Doubling the premium from two to four percent raises probability from
51 to 52 percent and severity from 3.92 to 7.69 percent.
Severity accounts for 97.2 percent of that log premium increase.

These benchmarks explain why severity dominance alone cannot distinguish
dispersion from distributional shape.
They also give different probability restrictions---constant $q$ under
common scale and $q>1/2$ under lognormality.
A fixed nonlinear payoff map can produce still another pattern.
The fixed-curvature example in Appendix~\ref{App:DPPFixedCurvature}
shows that negative curvature can place $q$ below one-half and make it
decline with scale locally.
The empirical exercise therefore recovers probability without imposing
these families, then tests common scale using the wider quotation cross section.

\subsection{Recovering the Components}
\label{Sec:IdentifyingDPPComponents}

Recovering the decomposition requires one value of the risk-neutral
distribution function, rather than its entire density.
The first strike derivative of the put-price curve supplies that value
\citep{BreedenLitzenberger1978}.

\begin{proposition}[Identifying probability and severity]
	\label{prop:dpp_identification}
	Fix a maturity $\tau$ and forward $F_{t,t+\tau}$.
	Under the maintained pricing assumptions, suppose there is no atom
	at the forward and the AtM-forward put has a positive price.
	Its price and local strike slope identify
	\begin{equation}
		q_t^\tau=e^{r_t\tau}
		\left.\frac{\partial P_t(K,\tau)}{\partial K}\right|_{K=F_{t,t+\tau}}
		\qquad m_t^\tau=\frac{a_t^\tau}{q_t^\tau}
		\label{eq:dpp_slope_identification}
	\end{equation}
\end{proposition}

Differentiating the discounted expected put payoff gives the discounted
probability of finishing below the strike; the product identity then
gives severity.
Appendix~\ref{App:DPPIdentification} provides the proof and the
convexity bounds for discretely quoted strikes.
The price--probability pair also has a familiar smile interpretation:
at a fixed maturity, the gap $q_t^\tau-(1+a_t^\tau)/2$ equals the local
implied-volatility slope times a known positive factor.
The decomposition expresses this local option information in
probability and conditional-loss units
(Appendix~\ref{App:DPPLocalSlope}).

For the baseline recovery, I retain the original quotation selection,
parity-implied forwards, and selected maturities.
At each maturity $u$, let $K_t^-<F_{t,t+u}<K_t^+$ be the nearest
quoted put strikes bracketing the forward.
Using bid--ask midpoints,
\begin{equation}
	\widehat q_t^u=e^{r_tu}
	\frac{P_t(K_t^+,u)-P_t(K_t^-,u)}{K_t^+-K_t^-}\qquad
	\widehat m_t^u=\frac{a_t^u}{\widehat q_t^u} \notag
\end{equation}
The put spread averages the cumulative probability across the bracket
and approximates its value at the forward.
It imposes neither lognormality nor fixed shape.
Because the same bracket constructs the AtM premium, the midpoint
interpolation is a recovery choice, rather than exact identification
from finitely many quotations.

Carry-adjusted maturity weights preserve the original daily DPP.
Monthly probability averages daily probabilities with carry weights;
monthly severity additionally weights by shortfall probability.
With these weights, the monthly adjusted premium satisfies
\begin{equation}
	A_M\equiv\frac{\bar\chi_M}{\bar D_M}=q_Mm_M.
	\label{eq:dpp_monthly_product}
\end{equation}
These components summarize daily forward-looking distributions.
They are not components of a realized calendar-month return, nor does
maturity interpolation separately identify an exactly traded 30-day law.
Appendix~\ref{App:DPPAggregation} gives the weighting formulas and
sample checks.
The baseline contains 8,791 daily observations and 439 months from
January 1990 through July 2026.

\subsection{Severity Dominance in the Data}
\label{Sec:DPPComponentsData}

\paragraph{Monthly premium movements.}
The identity $\Delta\log A_M=\Delta\log q_M+\Delta\log m_M$
gives an exact covariance decomposition
\begin{equation}
	\beta_q\equiv
	\frac{\operatorname{Cov}(\Delta\log q_M,\Delta\log A_M)}
	{\operatorname{Var}(\Delta\log A_M)}\qquad
	\beta_m=1-\beta_q
	\label{eq:dpp_empirical_elasticities}
\end{equation}
These empirical elasticities are OLS slopes with an intercept on
identical adjacent monthly observations.
They measure comovement, rather than causal effects or local
derivatives along a specified distributional family.
Panel A of Table~\ref{Tab:DPPDecomposition} reports a severity
elasticity of 1.032, with a 95 percent interval of $[0.994,1.071]$.
The probability elasticity is $-0.032$, with an interval of
$[-0.071,0.006]$.
Thus conditional severity moves approximately one-for-one with the
premium, while probability contributes little to its monthly movements.
Probability also fluctuates substantially less: its log-change standard
deviation is 0.044, compared with 0.183 for severity.
Figure~\ref{fig:dpp_log_decomposition} displays this contrast on a
common proportional scale.

\begin{table}[tbp]
	\centering\small
	\setlength{\tabcolsep}{4pt}
	\renewcommand{\arraystretch}{1.12}
	\caption{Monthly Comovement and the Recession Premium}
	\label{Tab:DPPDecomposition}
	\label{Tab:DPPMonthlyDecomposition}
	\label{Tab:DPPRecessionDecomposition}
	\begin{tabular*}{\textwidth}{@{\extracolsep{\fill}}lrrrr@{}}
		\toprule
		\multicolumn{5}{@{}l}{\textit{Panel A: Monthly comovement}}\\
		\addlinespace[0.3em]
		Component & Elasticity & HAC s.e. & 95\% interval
		& $\operatorname{sd}(\Delta\log z_M)$\\
		\midrule
		Adjusted premium, $A_M$ & --- & --- & --- & $0.172$\\
		Probability, $q_M$ & $-0.032$ & $0.020$ & $[-0.071,0.006]$ & $0.044$\\
		Severity, $m_M$ & $1.032$ & $0.020$ & $[0.994,1.071]$ & $0.183$\\
		\bottomrule
	\end{tabular*}
	\vspace{0.8em}
	\begin{tabular*}{\textwidth}{@{\extracolsep{\fill}}lrrrrr@{}}
		\toprule
		\multicolumn{6}{@{}l}{\textit{Panel B: The recession premium}}\\
		\addlinespace[0.3em]
		Component & \shortstack[r]{Expansion\\mean (\%)}
		& \shortstack[r]{Recession\\mean (\%)} & $\mathcal G(z)$
		& HAC s.e. & \shortstack[r]{Contribution (\%)\\{}[95\% interval]}\\
		\midrule
		Adjusted premium, $A_M$ & $1.77$ & $3.17$ & $0.576$ & $0.109$ & $100.0$\\
		Probability, $q_M$ & $42.31$ & $44.41$ & $0.049$ & $0.014$
		& \shortstack[r]{$8.5$\\$[3.7,13.3]$}\\
		Severity, $m_M$ & $4.18$ & $7.16$ & $0.527$ & $0.105$
		& \shortstack[r]{$91.5$\\$[86.7,96.3]$}\\
		\bottomrule
	\end{tabular*}
	\par\smallskip
	\begin{minipage}{\textwidth}\footnotesize
		\textit{Notes:} January 1990--July 2026: 439 months, 438 adjacent
		changes, 36 recession months, and 403 expansion months.
		Recessions run from the month after an NBER peak through the trough.
		Phase means are arithmetic means; $\mathcal G(z)$ is the difference
		in mean log levels. Contributions equal $100\mathcal G(z)/\mathcal G(A)$;
		products of phase means need not equal the phase mean premium.
		Bartlett HAC uses six monthly lags and the correction $n/(n-2)$.
		Share intervals use the delta method and joint gap covariance.
		Intervals condition on recovered components.
	\end{minipage}
\end{table}

\begin{figure}[t!]
	\centering
	\input{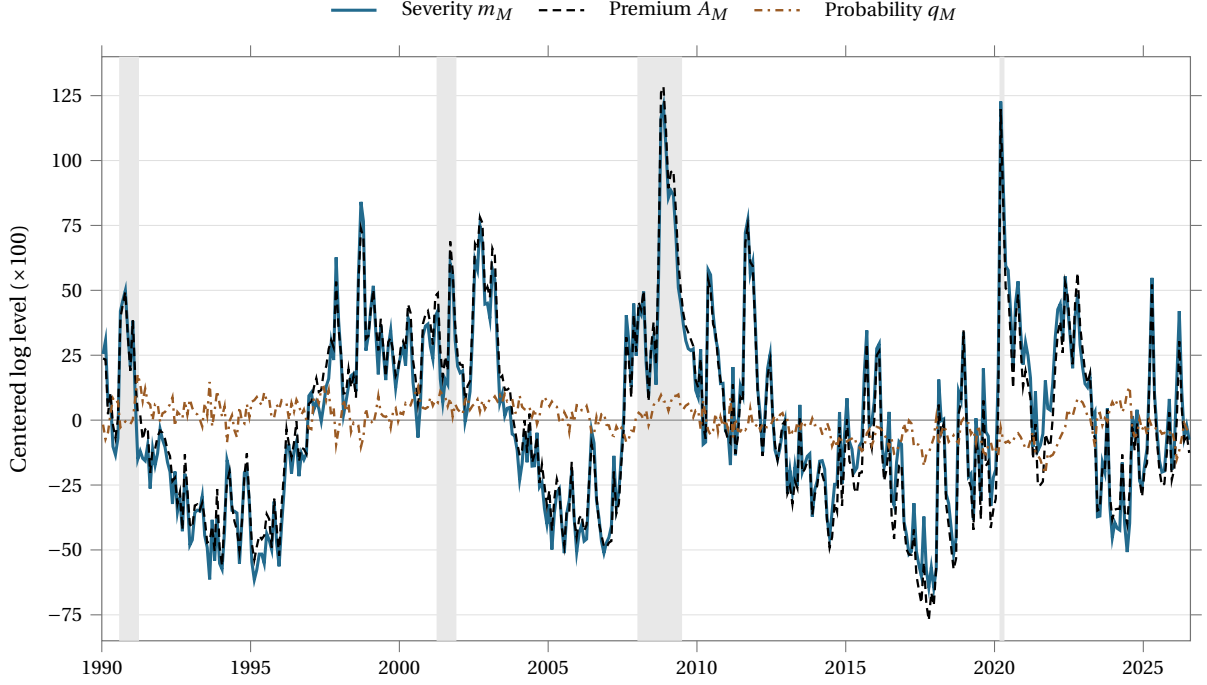}
	\caption{The DPP decomposition over time}
	\label{fig:dpp_log_decomposition}
	\begin{minipage}{\textwidth}
		\footnotesize
		\textit{Notes:}
		Each series is $100(\log z_M-\overline{\log z})$ for
		$z\in\{q,m\}$, where the bar denotes the full-sample mean.
		The sum of the two components equals the centered log adjusted
		premium at every month; the premium is not separately plotted.
		Gray shading denotes recession months, from the month after
		the business-cycle peak through the trough.
		The sample runs from January 1990 through July 2026.
		The lines show recovered component estimates; uncertainty
		about component recovery is not represented by bands.
	\end{minipage}
\end{figure}

\paragraph{The recession premium.}
Let $\mathcal G(z)$ denote the recession mean log level minus its
expansion counterpart.
On identical observations,
\begin{equation}
	\mathcal G(A)=\mathcal G(q)+\mathcal G(m)
	\label{eq:dpp_recession_decomposition}
\end{equation}
Panel B shows that severity accounts for 91.5 percent of this log
premium gap, with a 95 percent interval of $[86.7,96.3]$ percent.
Mean severity rises from 4.18 percent in expansions to 7.16 percent in
recessions, while mean probability rises from 42.31 to 44.41 percent.
Probability makes a smaller but positive contribution: its log gap is
0.049, with an interval of $[0.021,0.077]$.
The result is severity dominance, rather than an exactly constant
shortfall probability.
The phase probabilities below one-half also differ from the lognormal
benchmark, while the positive recession gap conflicts with a single
common shape spanning the full sample.

\paragraph{Recovery and stability.}
Table~\ref{Tab:DPPRecoveryCommon} compares five recoveries on common
daily dates: put spreads, call spreads, a smooth curve, and entropy
reconstructions with two reference distributions.
Severity elasticities range from 1.016 to 1.055, and recession severity
shares from 90.8 to 93.5 percent.
The smooth and entropy methods jointly recover the premium and
probability, so this agreement does not depend on fixing the original
premium or on the baseline's local interpolation.

\begin{table}[tbp]
	\centering\small
	\caption{Component Recovery on a Common Daily Sample}
	\label{Tab:DPPRecoveryCommon}
	\begin{tabular}{lcc}
		\toprule
		Method & Severity elasticity & Recession severity share (\%)\\
		\midrule
		Baseline put spread & \shortstack{$1.039$\\$[0.997,1.081]$}
		& \shortstack{$91.2$\\$[86.1,96.3]$}\\
		Call spread & \shortstack{$1.055$\\$[1.017,1.092]$}
		& \shortstack{$92.0$\\$[87.2,96.9]$}\\
		Smooth curve & \shortstack{$1.042$\\$[1.009,1.076]$}
		& \shortstack{$91.9$\\$[87.0,96.9]$}\\
		Entropy, lognormal reference & \shortstack{$1.016$\\$[0.972,1.060]$}
		& \shortstack{$90.8$\\$[86.0,95.6]$}\\
		Entropy, uniform reference & \shortstack{$1.027$\\$[0.985,1.069]$}
		& \shortstack{$93.5$\\$[89.1,97.9]$}\\
		\bottomrule
	\end{tabular}
	\par\smallskip
	\begin{minipage}{\textwidth}\footnotesize
		\textit{Notes:} All methods use the same 7,078 daily dates,
		408 months, and 398 adjacent monthly changes. Brackets contain
		95 percent intervals under the same six-lag HAC procedure.
		Smooth and entropy methods use their own jointly fitted premium.
		Inference conditions on recovered components.
	\end{minipage}
\end{table}

The conclusion also survives wider strike brackets, quotation-quality
filters, minimum monthly coverage, rolling windows, alternative
recession comparisons, and quarterly-horizon sampling.
For example, omitting individual recessions yields severity shares of
90.0--93.3 percent, while quarterly severity loadings range from
0.951 to 1.023 across the reported recovery and sampling conventions.
Appendices~\ref{App:DPPTimeStability} and~\ref{App:DPPRecoveryRobustness}
report these checks, coverage, and conditional inference.
The separate bid--ask quotation envelopes are too wide to establish
dominance without recovery assumptions; the robust result is the
agreement across the implemented recoveries, not identification of the
decomposition from quotation intervals alone.

\begin{figure}[t!]
	\centering
	\input{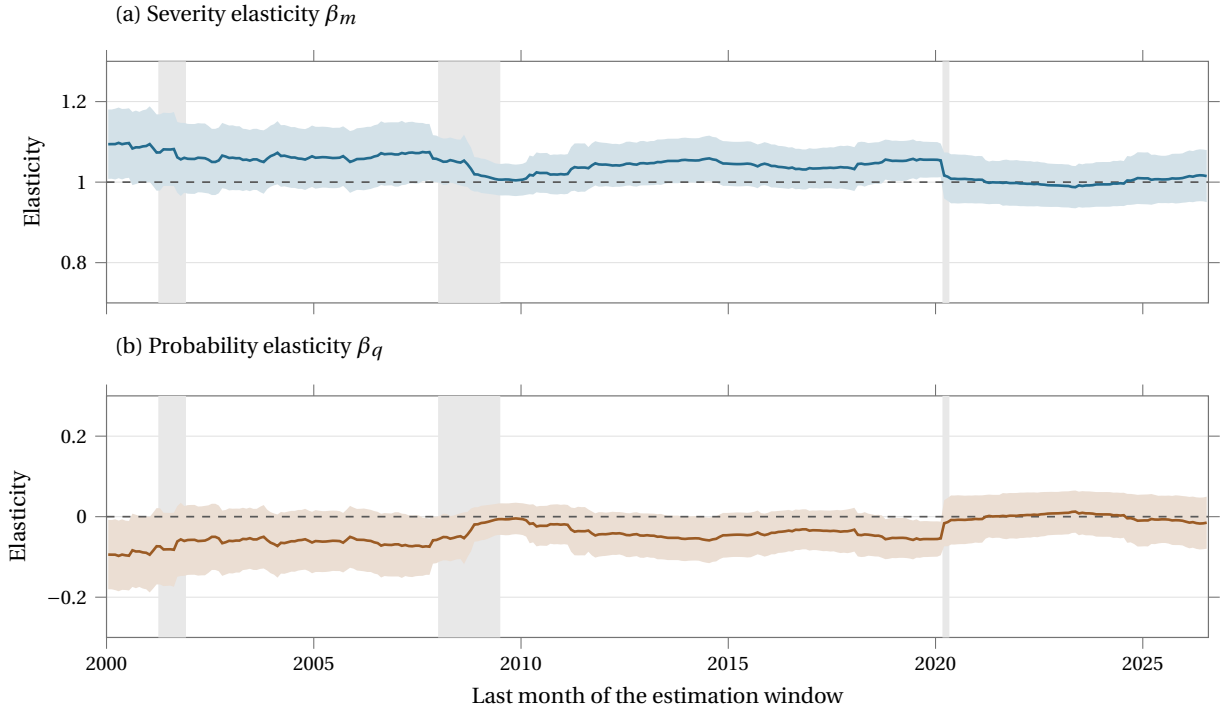}
	\caption{Probability and severity elasticities in rolling ten-year windows}
	\label{fig:dpp_rolling_elasticities}
	\begin{minipage}{\textwidth}
		\footnotesize
		\textit{Notes:}
		Each estimate uses 120 adjacent monthly log changes and an
		intercept; the horizontal axis gives the final month of the
		estimation window.
		Bands are pointwise 95 percent confidence intervals using
		Newey--West covariance with six monthly lags and the correction
		$n/(n-2)$.
		The two elasticities sum to one in every window.
		The panels use equal vertical ranges, centered on the fixed-shape
		benchmarks of one for severity and zero for probability,
		shown by dashed lines.
		Gray shading denotes recession months at the window endpoint.
		Inference conditions on the recovered components; the bands
		describe uncertainty in the regression elasticities, not in
		monthly component levels.
	\end{minipage}
\end{figure}

The evidence establishes a stable empirical composition of priced
downside risk: conditional severity accounts for nearly all monthly
premium comovement and roughly nine-tenths of the recession log premium.
Figure~\ref{fig:dpp_rolling_elasticities} shows that the severity elasticity
remains close to one across rolling ten-year windows, ranging from 0.987
to 1.097.
The complementary probability elasticity remains close to zero,
indicating that severity dominance is not confined to a single period.
Because the DPP closely tracks the VIX, this decomposition gives
economic content to the variation used in volatility-based macroeconomic studies.
The remaining question is whether that composition can be explained
by a common change in scale.

\subsection{Beyond a Common Change in Scale}
\label{Sec:CommonScaleInterpretation}

Severity dominance resembles the common-scale benchmark at the forward.
That benchmark also restricts option prices away from the forward:
the premium change determines how far the whole distribution expands.
I test this additional restriction against quotations across dates and
strikes, then relax it to distinguish probability changes from changes
in the shape of conditional losses.

Write $X_t=\tilde R_t^\tau$ and
$\mathcal P_t(k)=\mathbb E_t^{\mathbb Q}[(k-X_t)^+]$, where $k=K/F$.
Under Proposition~\ref{prop:dpp_scale_severity},
\begin{equation}
	\frac{\mathcal P_t(k)}{a_t}
	=C\!\left(\frac{k-1}{a_t}\right) \qquad C(0)=1
	\label{eq:dpp_scale_put_curve}
\end{equation}
with one convex function $C$ common to all dates in a comparison window.
I allow a different candidate shape in each calendar month and ask
whether any common shape within that month can fit the selected
quotations inside their original bid--ask bands.
The exercise retains the original forwards, carrying costs, component
AtM midpoint anchors, and maturity weights.
Each component law retains its quotation and moment restrictions;
common scale applies to the carry-adjusted daily mixture.
Thus the test conditions on the premium construction rather than
optimizing its anchors within quotation intervals.

Table~\ref{Tab:DPPShapeCompatibility} reports the 6,966-date sample
covering 405 months.
Full common scale is excluded in 352 months, or 86.9 percent.
Exclusions are certified in an enlarged distribution class that
contains the strict positive-terminal-value benchmark.
An exclusion therefore rules out that benchmark, whereas a verified
fit in the outer class need not establish an attainable strictly
positive distribution.
These are joint quotation-compatibility results, rather than sampling
rejection frequencies.

\begin{table}[tbp]
	\centering\small
	\caption{Compatibility of Scale Restrictions with Option Quotations}
	\label{Tab:DPPShapeCompatibility}
	\begin{tabular}{lrrrr}
		\toprule
		Restriction within each month & Months & Excluded & Outer fit & Unresolved\\
		\midrule
		Full distribution: common scale & 405 & 352 & 53 & 0\\
		Conditional downside shape: constant $q$ & 405 & 345 & 55 & 5\\
		Conditional downside shape: varying $q$ & 405 & 112 & 138 & 155\\
		\bottomrule
	\end{tabular}
	\par\smallskip
	\begin{minipage}{\textwidth}\footnotesize
		\textit{Notes:} The sample covers 405 months between January 1990
		and July 2026. All rows retain original quotation bands, fixed
		component anchors, and maturity construction. The last two rows
		allow the upside distribution to vary across dates. Exclusions
		are certified in an outer class; outer fits are verified numerical
		witnesses in that class. Unresolved outcomes establish neither
		exclusion nor a joint fit and remain in the denominators.
		Verified fits are carried into nested larger classes.
	\end{minipage}
\end{table}

To isolate conditional downside shape, let
\begin{equation}
	(1-X_t)\mid\{X_t<1\}=m_tU,\qquad
	\mathbb E^{\mathbb Q}[U]=1,\qquad m_t=\frac{a_t}{q_t},
	\label{eq:dpp_conditional_scale}
\end{equation}
where $U$ has a common conditional distribution within the month,
while the upside distribution can vary freely across dates.
For $k\leq1$, this implies
\begin{equation}
	\frac{\mathcal P_t(k)}{a_t}
	=H\!\left(\frac{(1-k)q_t}{a_t}\right),\qquad
	H(z)=\mathbb E^{\mathbb Q}[(U-z)^+].
	\label{eq:dpp_conditional_put_curve}
\end{equation}
I compare an unknown common $q$ with an unknown separate $q_t$ at
each date, retaining all component quotation and moment restrictions.

Most full-scale exclusions survive freeing the upside shape with
345 of the 352 excluded months remain excluded under constant-$q$
conditional scale.
Allowing probability to vary restores outer compatibility in 80 of
those months.
Nevertheless, 112 of all 405 months---27.7 percent---remain excluded
even when each date has its own probability.
In these months, probability adjustment alone cannot reconcile a
common conditional downside shape with the maintained quotations.
The varying-$q$ calculation also yields 138 outer fits and 155
unresolved months, so 27.7 percent is a lower bound on the frequency
of required conditional-shape departures under these restrictions.

For the full-scale exclusions, the median certified lower bound on
the additional price-band allowance is 1.39 percent of the daily
adjusted premium.
This is a minimum price discrepancy beyond the original bid--ask bands,
not an estimate of a macroeconomic effect or a sufficient allowance
to restore compatibility.
Appendix~\ref{App:DPPScaleRestrictions} provides the construction,
numerical criteria, price-distance diagnostics, and window comparisons.

The section therefore supplies two complementary empirical restrictions.
At the forward, movements in the price of protection are overwhelmingly
associated with conditional shortfall severity.
Across strikes, those movements frequently require more than a common
rescaling, including changes in conditional downside shape in a
substantial subset of months.
Together, these results sharpen what a model matching stock-market
volatility must explain: the composition of expected shortfall and its
relation to the wider option-price cross section.
The next section examines their macroeconomic interpretation.

	\section{Interpreting Uncertainty Shocks}
\label{Sec:InterpretingUncertainty}

Section~\ref{Sec:DecomposeDPP} establishes that conditional severity
accounts for most movements in the price of downside protection.
Does this common composition also imply a common macroeconomic response?
I address this question using the realized-volatility and
future-volatility-news innovations of
\citet{BergerDewBeckerGiglio2020SecondMomentNews}.
Both raise the DPP primarily through conditional severity, yet their
employment responses remain different at the point estimates after
equalizing the impact increase in severity.
I then explain how these observations restrict a structural model's
mapping from economic risk to option prices and macroeconomic activity.

\subsection{Severity Dominance and Macroeconomic Responses}
\label{sec:SevDominanceMacro}
\label{sec:bdg}

\citet{BergerDewBeckerGiglio2020SecondMomentNews} distinguish an
innovation in realized stock-market volatility from news about future
volatility orthogonal to the contemporaneous realized-volatility
innovation.
They find a larger employment decline following the realized-volatility
innovation, despite comparable cumulative effects on expected future
volatility.
This distinction connects uncertainty to the macroeconomic literature
on news: information about future conditions can affect current
decisions before those conditions materialize
\citep{BeaudryPortier2006StockPrices,BarskySims2011NewsShocks,
	BeaudryPortier2014NewsDriven}.
The DPP decomposition adds two questions: whether the innovations
reprice different shortfall components, and whether their employment
contrast survives expressing them in common impact-severity units.

\paragraph{Innovation construction and option responses.}
I reproduce the authors' innovation construction using their
February 1983--December 2014 sample.
Following their replication code, the realized-volatility innovation
is scaled to match the news innovation's cumulative effect on expected
future volatility over months 1--23.
On the overlapping 1990--2014 option sample, I estimate local projections
of $\log z_{t+h}-\log z_{t-1}$ on both innovations and lagged controls,
for $z\in\{A,q,m\}$.
The monthly option measures summarize the final three trading sessions,
rather than the full-month averages in Section~\ref{Sec:DecomposeDPP},
and preserve $A=qm$.
Identical observations and regressors therefore imply
\begin{equation}
	\operatorname{IRF}_{\log A}(h)
	=\operatorname{IRF}_{\log q}(h)
	+\operatorname{IRF}_{\log m}(h)
	\label{eq:interpretation_component_responses}
\end{equation}
The identity holds in every joint bootstrap draw.
Macroeconomic projections use the same observations, innovations, and
controls, and inference re-estimates the first stage and volatility
normalization in each draw.
Appendix~\ref{app:bdg_recovery} gives the estimation and sampling details.

Table~\ref{tab:bdg_components} shows that both impact premium responses
are dominated by severity.
The probability-to-premium response ratios are $-0.031$ for realized
volatility and $-0.082$ for news.
The positive severity responses rule out an impact premium increase
operating entirely through shortfall probability under this
identification and recovery.
The imprecise ratio difference does not establish equal compositions,
and the probability intervals do not establish an exactly zero response.
These innovation-specific ratios are distinct from the aggregate
monthly covariance loading in Section~\ref{Sec:DecomposeDPP}.

\begin{table}[!htbp]
	\centering
	\caption{Impact responses of the protection premium and its components}
	\label{tab:bdg_components}
	\small
	\begin{tabular*}{\textwidth}{@{\extracolsep{\fill}}lccc@{}}
		\toprule
		\multicolumn{4}{@{}l}{
			\textit{A. Cumulative log responses, multiplied by 100}}\\[3pt]
		Innovation & Premium $A$ & Probability $q$ & Severity $m$ \\
		\midrule
		Realized volatility & $8.76$ & $-0.27$ & $9.03$ \\
		& $[3.87,\,17.42]$ & $[-1.54,\,0.52]$
		& $[3.93,\,18.56]$ \\[4pt]
		Future-volatility news & $11.44$ & $-0.94$ & $12.38$ \\
		& $[10.26,\,12.14]$ & $[-2.03,\,0.22]$
		& $[10.84,\,13.33]$ \\
		\midrule
		\multicolumn{4}{@{}l}{
			\textit{B. Probability-to-premium impact response ratios}}\\[3pt]
		& Realized volatility & News & RV minus news \\
		\midrule
		Ratio or difference & $-0.031$ & $-0.082$ & $0.051$ \\
		& $[-0.120,\,0.060]$ & $[-0.187,\,0.019]$
		& $[-0.093,\,0.202]$ \\
		\bottomrule
	\end{tabular*}
	\par\smallskip
	\begin{minipage}{\textwidth}\footnotesize
		\emph{Notes:} There are 291 impact origins in 1990--2014.
		Brackets report 95\% pointwise percentile intervals from 999 joint
		circular-block bootstrap draws with 24-month blocks.
		The first-stage VAR, innovations, and realized-volatility
		normalization are re-estimated in each draw.
		Ratios and their difference are formed within each draw.
		Intervals condition on the recovered daily option components;
		recovery methods are compared separately in
		Appendix~\ref{app:bdg_recovery}.
	\end{minipage}
\end{table}

\paragraph{The employment contrast.}
Table~\ref{tab:bdg_macro} reports the twelve-month activity responses
under the equal-future-volatility normalization.
Manufacturing employment falls by 0.477 after realized volatility,
while its estimated response to news is 0.151, in units of 100 times
cumulative log changes.
The direct contrast is $-0.628$, with a 95 percent simultaneous interval
of $[-1.159,-0.097]$ covering both activity contrast paths over horizons
0--24.
Industrial production has a negative pointwise contrast, but its
simultaneous interval includes zero.
Thus employment supplies the stronger evidence of a difference,
reproducing the qualitative pattern of
\citet{BergerDewBeckerGiglio2020SecondMomentNews} on the option sample.
The individual news interval includes modest contractions, so the
comparison does not establish that news has no contractionary effect.

\begin{table}[!htbp]
	\centering
	\caption{Twelve-month macroeconomic responses}
	\label{tab:bdg_macro}
	\small
	\setlength{\tabcolsep}{4pt}
	\begin{tabular*}{\textwidth}{@{\extracolsep{\fill}}lcccc@{}}
		\toprule
		& & & \multicolumn{2}{c}{RV minus news} \\
		\cmidrule(l){4-5}
		& RV & News & Pointwise & Simultaneous \\
		\midrule
		\multicolumn{5}{@{}l}{
			\textit{Manufacturing employment}}\\[2pt]
		Estimate & $-0.477$ & $0.151$ & $-0.628$ & $-0.628$ \\
		95\% interval & $[-0.758,-0.117]$ & $[-0.056,0.362]$
		& $[-0.929,-0.224]$ & $[-1.159,-0.097]$ \\[5pt]
		\multicolumn{5}{@{}l}{
			\textit{Industrial production}}\\[2pt]
		Estimate & $-0.762$ & $0.421$ & $-1.183$ & $-1.183$ \\
		95\% interval & $[-1.428,-0.081]$ & $[-0.121,0.998]$
		& $[-1.966,-0.242]$ & $[-2.539,0.173]$ \\
		\bottomrule
	\end{tabular*}
	\par\smallskip
	\begin{minipage}{\textwidth}\footnotesize
		\emph{Notes:} Entries are 100 times cumulative log responses
		at $h=12$ (275 origins).
		Relative to the 291 impact origins, twelve origins in 2014
		are lost because $t+12$ exceeds December 2014, and four earlier
		origins are lost because option components are missing at $t+12$.
		Every column uses normalization re-estimated in each bootstrap draw.
		Individual-innovation intervals are pointwise.
		The final column uses a simultaneous band for both RV-minus-news
		contrast paths over horizons 0--24.
		Contrasts are formed within joint draws, rather than by
		subtracting interval endpoints.
	\end{minipage}
\end{table}

\paragraph{Comparing responses at equal impact severity.}
Equal cumulative effects on future volatility need not produce equal
impact severity responses.
Figure~\ref{fig:bdg_equal_severity} therefore scales each innovation to
raise $\log m$ by 0.01 on impact.
Let $e_j(h)$ denote the cumulative employment log response and
$b_j=\operatorname{IRF}_{\log m,j}(0)$ impact severity, for
$j\in\{R,N\}$.
A common proportional mapping from impact severity to employment requires
\begin{equation}
	\frac{e_R(h)}{b_R}=\frac{e_N(h)}{b_N}
	\quad\Longleftrightarrow\quad
	e_R(h)b_N-e_N(h)b_R=0
	\label{eq:bdg_division_free_restriction}
\end{equation}
for nonzero impact severities.
At twelve months, the normalized employment contrast is $-0.065$
in units of 100 times log changes per 0.01 impact log-severity increase.
Its pointwise interval is $[-0.098,-0.020]$, and its simultaneous
interval is $[-0.124,-0.006]$ over both activity contrasts at horizons
0--24.

The cross-product in \eqref{eq:bdg_division_free_restriction} checks
the same economic restriction without a random denominator.
Using the fixed point-estimate product to express each draw in comparable
units gives a pointwise interval of $[-0.104,-0.023]$ and a simultaneous
interval of $[-0.1305,0.0005]$.
Both formulations therefore provide pointwise evidence \emph{against} the common proportional mapping,
but their simultaneous conclusions differ: the ratio interval excludes zero,
while the cross-product interval narrowly includes it.
The twelve-month comparison was among the reporting horizons selected
before computing the macroeconomic estimates, after examining the
component paths, rather than a sole confirmatory test.
Appendix~\ref{app:bdg_recovery} reports the joint inference and reporting
protocol.

\begin{figure}[!htbp]
	\centering
	\includegraphics[width=\textwidth]
	{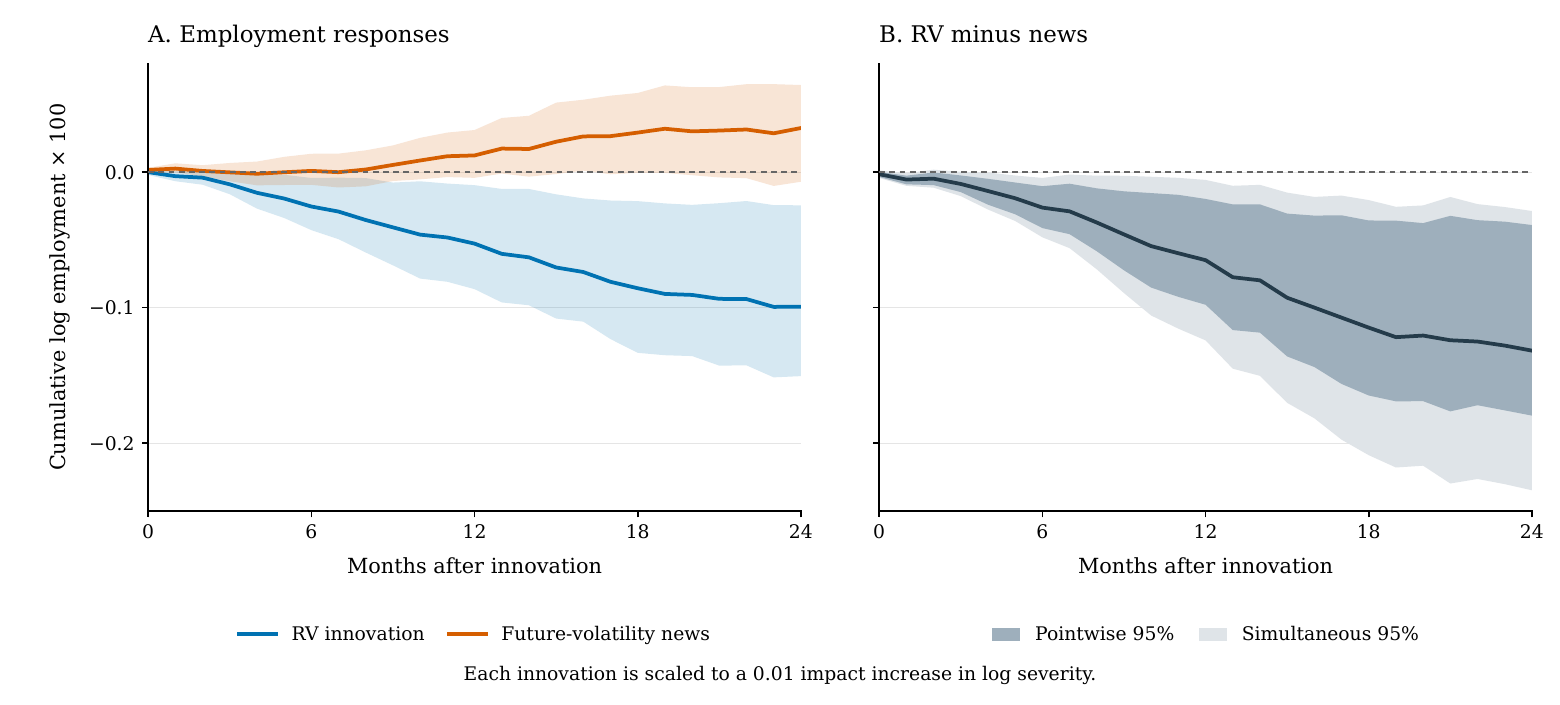}
	\caption{Manufacturing employment responses at equal impact severity.}
	\label{fig:bdg_equal_severity}
	\begin{minipage}{\textwidth}\footnotesize
		\emph{Notes:} Each innovation is scaled to increase $\log m$
		by 0.01 on impact.
		Employment responses are 100 times cumulative log changes.
		Colored shading in the left panel and darker shading in the
		right panel give 95\% pointwise intervals.
		Lighter shading in the right panel is the ratio-based 95\%
		simultaneous band covering both employment and
		industrial-production contrasts over horizons 0--24,
		although only employment is plotted.
		The twelve-month employment contrast is $-0.065$ in units
		of 100 times log changes, with a 95\% simultaneous interval
		of $[-0.124,-0.006]$.
		Each of the 999 joint draws uses its own impact severity
		response as the denominator.
		Under each draw's re-estimated volatility normalization,
		the smallest impact log-severity denominators are 0.01057
		for realized volatility and 0.09057 for news.
		The cross-product check is reported separately in the text.
		Later severity paths and other variables are not held fixed.
		Sample: 1990--2014.
	\end{minipage}
\end{figure}

This normalization compares responses in common impact units; it does
not hold subsequent severity paths or other equilibrium responses fixed.
It consequently evaluates the proportional restriction in
\eqref{eq:bdg_division_free_restriction}, rather than identifying a
causal effect of severity or testing whether the full priced-distribution
path explains employment.
An exploratory equal-impact-premium normalization gives a similar
twelve-month contrast of $-0.068$, with a ratio-based simultaneous
interval of $[-0.127,-0.008]$.
An unrestricted mapping from both impact components is not tested:
with two innovations, two freely chosen component weights can fit the
two employment responses.
The appendix reports that calculation and the direct comparison of
subsequent severity paths.

Call-spread and smooth-curve intersections preserve severity dominance
and the negative twelve-month employment contrast at equal impact
severity under the ratio-based simultaneous bands.
The entropy intersections retain similar employment point estimates
but have substantially less coverage, and their bands include zero.
Both innovations also increase conditional priced gain magnitudes,
so the finding does not describe an exclusively downside movement.
Appendix~\ref{app:bdg_recovery} reports the full recovery comparisons
and upside responses.

The decomposition therefore adds a substantive restriction to the
macroeconomic comparison: both innovations increase the price of
protection primarily through conditional shortfall severity, yet
their estimated employment responses differ even in common impact
severity units.
The evidence against a common proportional response has pointwise
support under both inference formulations, with stronger simultaneous
support under the ratio formulation.
Severity dominance identifies a shared dimension of repricing, but
that dimension alone does not give the two innovations a common
employment response.

\subsection{What a Structural Model Must Explain}
\label{Sec:CompositionRestriction}

The option restrictions concern the priced distribution generated by a
model's equilibrium equity payoffs and stochastic discount factor.
To make this mapping explicit, fix the horizon and write
$X_t=S_{t+\tau}/F_{t,t+\tau}$ and
$L_t=(1-X_t)^+$.
For the stochastic discount factor $M_{t,t+\tau}>0$, define
\begin{equation}
	w_t=\frac{M_{t,t+\tau}}
	{\mathbb E_t^{\mathbb P}[M_{t,t+\tau}]}
	\qquad
	F_{t,t+\tau}=\mathbb E_t^{\mathbb P}[w_tS_{t+\tau}]
	\notag
\end{equation}
Under the maintained deterministic carrying-cost assumptions, $w_t$
maps physical probabilities into the pricing probabilities used
above.\footnote{Normalizing by the zero-coupon bond price generally
	defines the horizon-specific forward pricing measure.
	It coincides with the risk-neutral measure used here under
	the maintained assumptions.
	Payoffs and the discount factor must use the same units.}
For the shortfall event $\mathcal D_t=\{X_t<1\}$, let
$p_t=\mathbb P_t(\mathcal D_t)$,
$\kappa_t=\mathbb E_t^{\mathbb P}[w_t\mid\mathcal D_t]$, and
$m_t^P=\mathbb E_t^{\mathbb P}[L_t\mid\mathcal D_t]$.
For $p_t>0$, the priced components satisfy
\begin{equation}
	q_t=p_t\kappa_t
	\qquad
	m_t=m_t^P+
	\frac{\operatorname{Cov}_t^{\mathbb P}
		(w_t,L_t\mid\mathcal D_t)}
	{\kappa_t}
	\notag
\end{equation}
Probability combines physical frequency with average valuation in
shortfall states; severity combines physical conditional losses with
their relative valuation within those states.
The forward threshold is itself endogenous, so a valuation change can
also change the measured physical components.
Appendix~\ref{App:DPPPhysicalMapping} derives the identities and
examines restricted valuation experiments.

The empirical results restrict this joint mapping without identifying
its physical and valuation terms separately.
Near-zero aggregate probability comovement does not require constant
$p_t$ or $\kappa_t$, or a zero probability response to every shock.
Likewise, a primitive disaster-probability shock need not appear mainly
in $q_t$: by placing more weight on very large losses, it can increase
severity within the broader event of any shortfall.
The empirical components therefore discipline a model's observables,
rather than assigning names to its primitive shocks.

A structural interpretation must account for two linked sets of
observations.
Section~\ref{Sec:DecomposeDPP} establishes the composition of premium
movements, probability levels, and departures from common scale across
strikes.
The present section establishes severity-dominated impact repricing
for two identified volatility innovations with different estimated
employment responses, including after impact-severity normalization.
Matching implied volatility or a near-unit severity loading alone
does not match this joint evidence.

For a model disciplined by implied stock-market volatility, such as
\citet{BasuBundick2017}, the corresponding comparison requires pricing
the model's equity protection claims and applying the empirical
identification to simulated observables at comparable horizons.
\citet{BergerDewBeckerGiglio2020SecondMomentNews} already apply their
identification to simulations of that model.
The additional question is whether its equilibrium option prices
reproduce the composition and quotation restrictions in Section 3
alongside the component and employment responses studied here.
These results give that comparison observable content; they do not
establish that a particular model fails it.
Severity dominance sharpens the interpretation of movements in the
uncertainty proxy, while their macroeconomic consequences depend on
the economic innovations and equilibrium propagation that generate them.
	\section{Conclusion}
\label{sec:conclusion}

The normalized at-the-money forward put tracks the VIX, preserves the
principal macroeconomic responses in the specifications examined, and
separates the priced probability of a shortfall from its conditional
severity.
Conditional severity accounts for almost all monthly premium comovement
and most of the recession premium, and conditional gains rise with it.
Because priced shortfall probabilities average below one-half,
conditional losses remain larger than gains.

These regularities do not reduce to a common change in scale.
For the recovered series, the probability level and its recession
increase rule out exact lognormality, a single shape spanning the
sample, and a constant local probability gap.
Conditional on the maintained forward and premium anchors, quotations
across strikes exclude common scale within 352 of 405 covered months
and require a change in conditional downside shape in at least 112.
Realized-volatility and future-volatility-news innovations both raise
protection costs mainly through severity, yet are followed by different
estimated employment responses at equal impact severity.

A model whose uncertainty process is disciplined by implied volatility
has matched one margin of option prices.
Its equilibrium quotations, constructed for the same payoff and horizon
and passed through the same component recovery and innovation
identification, provide a natural test of the remaining restrictions.

\bigskip
\section*{Declaration of Artificial Intelligence Usage}
The first version of this paper was drafted without any AI assistance. I used ChatGPT-6.1 to help produce the replication package available at XX. A draft of the manuscript was analyzed by ChatGPT-6.0 and suggestions were made to improve the readability of the paper. I implemented a subset of these recommendations and take full responsibility for the contents of this article.

\section*{Data Availability Statement}
A replication package for results in the paper and Appendices A1 and A2 is available at doi: \url{10.5281/zenodo.23106238}.

\pagebreak

\bibliographystyle{aer}
\bibliography{VixAltReferences}

\pagebreak

	\appendix

\counterwithin{equation}{section}
\counterwithin{table}{section}
\counterwithin{figure}{section}
\renewcommand{\theequation}{\thesection.\arabic{equation}}
\renewcommand{\thetable}{\thesection.\arabic{table}}
\renewcommand{\thefigure}{\thesection.\arabic{figure}}

\section{Supplementary Appendix (Online Appendix)}
\label{App:SupplementaryMaterial}

This appendix documents the data and construction in
Appendix~\ref{app:data}, the macroeconomic comparisons and
physical-to-priced mapping in Appendix~\ref{App:MacroReplication},
and the derivations and robustness of the DPP decomposition in
Appendix~\ref{App:DPPSection3}.
The main estimates and their supporting robustness comparisons are
documented together.

\subsection{Data and DPP Construction}
\label{app:data}

\subsubsection{S\&P 500 Index Options}
\label{App:OptionData}

The options data are daily end-of-day S\&P 500 Index (SPX) puts and calls on a
wide range of strikes obtained from DeltaNeutral (HistoricalOptionData.com) from
January 1990 through July 2026.
The data include traditional monthly SPX contracts (morning settlement) and SPXW contracts
(afternoon settlement).
This distinction matters for measuring time to maturity and matching put
and call contracts as an identical strike and calendar expiration date need
not imply identical settlement terms.
For each contract, the data report its full option identifier, put or call
type, strike, expiration date, bid and ask quotations, last transaction
price, trading volume, and open interest, together with the underlying
index level at the recorded snapshot.
Option prices and strikes are measured in index points, and the contract
multiplier is \$100 per index point.
Dividing a put price by the underlying index level therefore expresses its
premium per dollar of index exposure.
Trading volume and open interest are measured in numbers of contracts.

The distinction between quoted and traded prices is relevant for assessing
the quality of these observations.
The bid and ask describe the market quotation at the daily snapshot,
whereas the last transaction price may refer to a trade on an earlier date.
An option can therefore have an end-of-day quotation even when its reported
daily volume is zero.

\subsubsection{Quotation Selection and Premium Coverage}
\label{App:PremiumCoverage}

The baseline uses bid--ask midpoints, requiring a nonnegative bid, a
strictly positive ask, and an ask no lower than the bid. No minimum volume
or maximum relative spread is imposed. Puts and calls used in parity
calculations are matched by snapshot, strike, expiration, contract family,
and settlement terms.
For each selected maturity, the forward is the median parity-implied value
from up to three matched strikes closest to the contemporaneous index.
The daily three-month Treasury bill rate is converted to continuously
compounded units; missing rates are carried forward for at most three
calendar days.
Prices are interpolated first across the nearest strict put strikes
bracketing that maturity's forward, then across the eligible maturities
bracketing 30/365 years. The total strike-bracket width cannot exceed five
percent of the forward. There is no additional limit on the distance
between bracketing maturities. Neither strike nor maturity extrapolation
is used.

The premium is available on 8,825 of 9,212 regular trading sessions from
January 2, 1990 through July 31, 2026. Of the 387 missing sessions, two
have no source SPX quotations, 379 lack a usable maturity bracket, and
six fail the strike-width restriction. Of the maturity failures, 359
occur between the 14th and 21st of the month, when the next standard
monthly expiration is more than 30 days away; the remaining 20 occur
between September 18 and October 15, 1992. No missing date is filled.
Arithmetic averages of available daily premiums yield 439 monthly
observations. Component recovery retains 8,791 dates; its exclusions
and matched-sample effects are documented in
Appendix~\ref{App:DPPRecoveryRobustness}.

\subsubsection{Trading Volume and Coverage}
\label{App:TradingVolume}

Volume measures trading during the day; open interest measures outstanding
positions and is updated separately, before the trading session.
Figure~\ref{fig:volume_coverage} summarizes recorded SPX trading activity
over the available sample.
The volume calculations use observations with positive reported volume,
after removing repeated records for the same contract and date.
Panel~(a) reports 21-weekday moving averages of daily put and call volume,
summed across strikes and expirations.
Trading activity increases substantially during the 2000s, and put volume
generally exceeds call volume.

\begin{figure}[htbp]
	\centering
	\includegraphics[width=\linewidth]{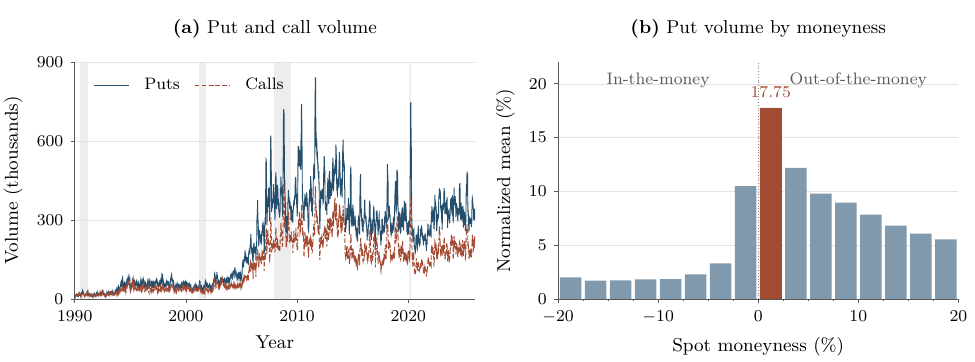}
	\caption{SPX trading activity.
		Panel~(a) shows 21-weekday moving averages of aggregate daily volume,
		in thousands of contracts.
		Weekdays without recorded trading enter the moving average as zeros;
		shaded intervals mark NBER business-cycle peaks to troughs.
		Panel~(b) shows mean volume per positive-volume put contract-date in
		2.5-percentage-point spot-moneyness bins, normalized to sum to 100
		over the displayed range, $[-20,20]$ percent.
		Positive moneyness denotes out-of-the-money puts.
		Both panels pool expirations.
		Source: DeltaNeutral option data; NBER business-cycle chronology.}
	\label{fig:volume_coverage}
\end{figure}

Panel~(b) describes trading activity across strikes using spot moneyness,
$m_{it}=100(S_t-K_i)/S_t$, where $S_t$ is the underlying index level and
$K_i$ is the strike.
This differs from the paper's at-the-money put, which is defined relative
to the forward price.
For each moneyness bin, I compute mean volume per active put contract-date
and normalize these means to sum to 100 across bins.
Average activity is greatest in the 0--2.5 percent out-of-the-money bin,
whose normalized mean is 17.75 percent, and is generally lower farther
from the money.

\subsection{Macroeconomic Replication and Structural Interpretation}
\label{App:MacroReplication}

This appendix examines whether replacing the VIX/VXO with the thirty-day
Downside Protection Premium (DPP) changes
the identification of uncertainty shocks or their estimated
macroeconomic effects.
Within each proxy comparison, I replace only the uncertainty proxy; the sample, macroeconomic controls,
lag order, deterministic terms, and recursive ordering are all held
fixed.
These comparisons assess the preservation of the principal macroeconomic
responses, rather than exact agreement in event dates or inference
classifications.
The remaining material derives the physical-to-priced mapping and documents
the volatility-innovation comparisons in
Section~\ref{Sec:InterpretingUncertainty}.
Derivations and robustness supporting the DPP decomposition, including
the quotation-based scale restrictions, appear in
Appendix~\ref{App:DPPSection3}.

\subsubsection{Bloom-Style Event Identification}
\label{App:BloomReplication}

Following the threshold approach of \citet{Bloom2009}, I apply
the HP filter ($\lambda=129{,}600$) to monthly DPP and VIX
levels over January 1990--July 2026.
Both series are averaged over identical daily dates.
I standardize each cyclical component by its sample standard
deviation and identify the peak month within each contiguous
episode exceeding 1.65 standard deviations consistent with \cite{Bloom2009}.

Figure~\ref{fig:bloom_shock_comparison} shows closely aligned
cycles, with correlations of 0.9926 in levels and 0.9885 in
logarithms.
At the 1.65 cutoff, each level-based series identifies ten
events, eight with identical peak months.
The VIX additionally identifies October 1990 and November 1997,
while the DPP additionally identifies May and October 2022.
October 1990 narrowly misses the DPP cutoff, with a standardized
cycle of 1.6499 compared with 1.6539 for the VIX.
The logarithmic specification identifies fourteen events for
each proxy, eleven with identical peak months.

The largest events coincide.
At a common cutoff of 2.02 standard deviations, both level-based
series identify the same six peak months: September 1998,
September 2002, November 2008, September 2011, March 2020,
and April 2025.
The largest events therefore coincide, although the baseline threshold
selects some different episodes. The October 1990 discrepancy is particularly
sensitive to the cutoff; the other differences are not merely rounding errors.

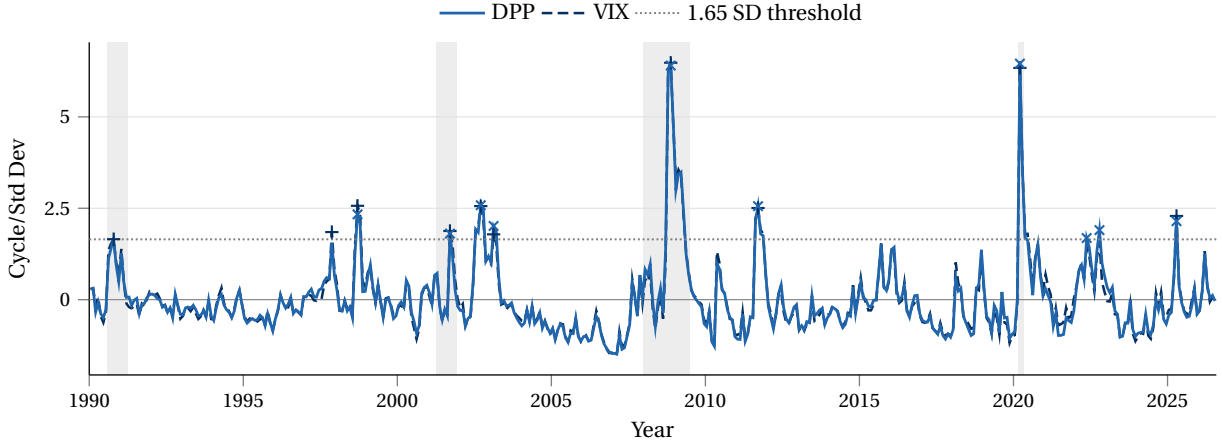
\begin{figure}[h!]
	\centering
\begingroup%
\normalfont\linespread{1}\fontsize{9}{11}\selectfont%
\definecolor{dppred}{HTML}{BE263B}%
\definecolor{vixblack}{HTML}{202020}%
\definecolor{dppblue}{HTML}{2166AC}%
\definecolor{vixblue}{HTML}{053061}%
\begin{tikzpicture}
\begin{groupplot}[
group style={group size=1 by 1,horizontal sep=1.25cm,vertical sep=1.2cm},
width=\dimexpr\linewidth-1.6cm\relax,height=4.4cm,scale only axis,
axis lines*=left,axis line style={line width=0.5pt},tick align=outside,tick style={line width=0.5pt},
label style={font=\fontsize{9}{11}\selectfont},tick label style={font=\fontsize{8}{10}\selectfont},
title style={at={(0,1.02)},anchor=south west,yshift=0pt,font=\fontsize{9}{11}\selectfont},
scaled ticks=false,yticklabel style={/pgf/number format/fixed,/pgf/number format/precision=4},
xticklabel style={/pgf/number format/1000 sep={}},
ymajorgrids=true,grid style={gray!25,line width=0.3pt},unbounded coords=jump,set layers=standard]
\nextgroupplot[xmin=1990,xmax=2026.58333333,ymin=-2.05915140122,ymax=7.03903142487,xtick={1990,1995,2000,2005,2010,2015,2020,2025},ytick={0,2.5,5},title={},ylabel={Cycle/Std Dev},xlabel={Year},legend style={at={(0.5,1.025)},anchor=south,draw=none,fill=none,font=\fontsize{9}{11}\selectfont},legend columns=3]
\begin{pgfonlayer}{axis background}
\fill[black!7] (axis cs:1990.58333333,-2.05915140122) rectangle (axis cs:1991.25,7.03903142487);
\fill[black!7] (axis cs:2001.25,-2.05915140122) rectangle (axis cs:2001.91666667,7.03903142487);
\fill[black!7] (axis cs:2008,-2.05915140122) rectangle (axis cs:2009.5,7.03903142487);
\fill[black!7] (axis cs:2020.16666667,-2.05915140122) rectangle (axis cs:2020.33333333,7.03903142487);
\end{pgfonlayer}
\addplot[black!50,line width=0.4pt,forget plot] coordinates {(1990,0) (2026.5833333333333,0)};
\addplot[vixblue,line width=0.95pt,dash pattern=on 4pt off 2pt,forget plot] table[x=year,y=levels_vix,col sep=comma] {bloom_shock_comparison_plot_data.csv};
\addplot[only marks,mark=+,mark size=2.5pt,vixblue,line width=0.8pt,forget plot] table[x=year,y=levels_vix_peak,col sep=comma] {bloom_shock_comparison_plot_data.csv};
\addplot[dppblue,line width=1.05pt,forget plot] table[x=year,y=levels_dpp,col sep=comma] {bloom_shock_comparison_plot_data.csv};
\addplot[only marks,mark=x,mark size=2.5pt,dppblue,line width=0.8pt,forget plot] table[x=year,y=levels_dpp_peak,col sep=comma] {bloom_shock_comparison_plot_data.csv};
\addplot[black!50,densely dotted,line width=0.75pt,forget plot] coordinates {(1990,1.65) (2026.5833333333333,1.65)};
\addlegendimage{dppblue,line width=1.05pt}\addlegendentry{DPP}
\addlegendimage{vixblue,line width=0.95pt,dash pattern=on 4pt off 2pt}\addlegendentry{VIX}
\addlegendimage{black!50,densely dotted,line width=0.75pt}\addlegendentry{1.65 SD threshold}
\end{groupplot}
\end{tikzpicture}%
\endgroup%
	\caption{Bloom-style uncertainty events, January 1990--July 2026.
		Standardized HP cycles of monthly DPP (solid) and VIX
		(dashed), using levels in Panel~(a) and logarithms in
		Panel~(b), with $\lambda=129{,}600$.
		The dotted line marks 1.65 standard deviations;
		symbols identify peak months within contiguous
		above-threshold episodes.
		Grey shading denotes NBER recessions, from the month
		after the peak through the trough.}
	\label{fig:bloom_shock_comparison}
\end{figure}

\subsubsection{Monthly VAR of Leduc and Liu (2016)}
\label{App:MonthlyVAR}

The monthly comparison follows the recursive identification
of \citet{LeducLiu2016}, using a frequentist log-proxy VAR, as opposed
to their Bayesian approach.
For uncertainty proxy $j\in\{\mathrm{DPP},\mathrm{VIX}\}$, define
\begin{equation}
	X_m^{(j)}
	=
	\begin{pmatrix}
		\log \bar j_m \\
		U_m \\
		\pi_m \\
		i_m
	\end{pmatrix}
	\qquad
	X_m^{(j)}
	=
	a^{(j)}
	+
	\sum_{\ell=1}^{p} A_\ell^{(j)}X_{m-\ell}^{(j)}
	+
	u_m^{(j)}
	\notag
\end{equation}
Here $\bar j_m$ is the matched monthly average of the chosen
uncertainty proxy, $U_m$ is the unemployment rate,
$i_m$ is the three-month Treasury bill rate, and
$
	\pi_m
	=
	100\left[
	\log(\mathrm{CPI}_m)
	-
	\log(\mathrm{CPI}_{m-12})
	\right].
$
Data are the FRED series UNRATE, CPIAUCSL, and TB3MS.
Following \cite{LeducLiu2016}, the baseline specification includes a constant and twelve lags,
and identification uses a lower-triangular Cholesky factor of the
residual covariance matrix, with uncertainty ordered first.

I report two shock normalizations.
The first uses each model's own one-standard-deviation uncertainty
innovation.
The second scales both responses to an initial increase of
0.1 log unit in the uncertainty proxy, corresponding to a
10.52 percent proportional increase.
This comparison separates differences in propagation from
differences in the estimated innovation scale.
Uncertainty bands use 1,000 paired residual-bootstrap replications.
The standardized uncertainty innovations have a correlation
of 0.9836.

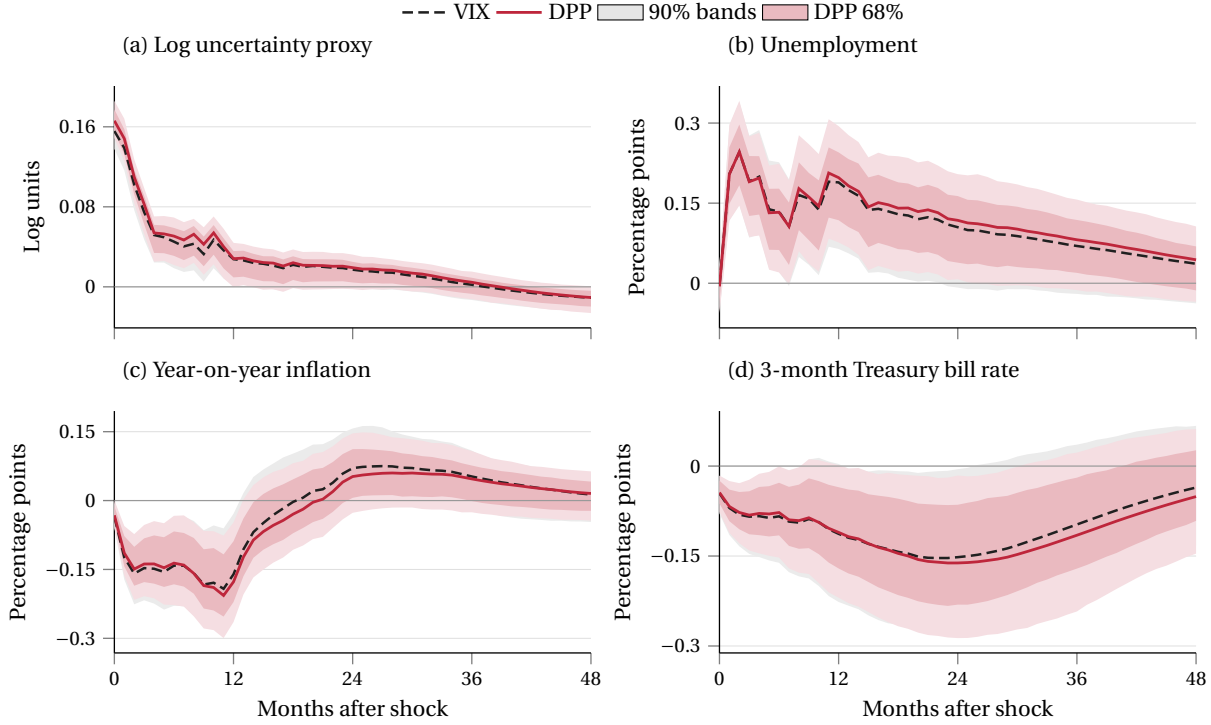
\begin{figure}[h]
	\centering
\begingroup%
\normalfont\linespread{1}\fontsize{9}{11}\selectfont%
\definecolor{dppred}{HTML}{BE263B}%
\definecolor{vixblack}{HTML}{202020}%
\begin{tikzpicture}
\begin{groupplot}[
group style={group size=2 by 2,horizontal sep=1.7cm,vertical sep=1.1cm},
width=\dimexpr.5\linewidth-1.95cm\relax,height=3.2cm,scale only axis,
axis lines*=left,axis line style={line width=0.5pt},tick align=outside,tick style={line width=0.5pt},
label style={font=\fontsize{9}{11}\selectfont},tick label style={font=\fontsize{8}{10}\selectfont},
title style={at={(0,1.02)},anchor=south west,yshift=0pt,font=\fontsize{9}{11}\selectfont},
scaled ticks=false,yticklabel style={/pgf/number format/fixed,/pgf/number format/precision=4},
xticklabel style={/pgf/number format/1000 sep={}},
ymajorgrids=true,grid style={gray!25,line width=0.3pt},unbounded coords=jump,set layers=standard]
\nextgroupplot[xmin=0,xmax=48,ymin=-0.04104467451,ymax=0.200799510812,xtick={0,12,24,36,48},ytick={0,0.08,0.16},title={(a) Log uncertainty proxy},ylabel={Log units},xlabel={},xticklabels={},legend style={at={([xshift=0.85cm]axis description cs:1,1.02)},anchor=south,yshift=18pt,draw=none,fill=none,font=\fontsize{9}{11}\selectfont},legend columns=4]
\addplot[draw=none,name path=band0vix90lo,forget plot] table[x=horizon,y=log_uncertainty_vix_lower90,col sep=comma] {logdpp30_var_irfs_vs_benchmark_dual_plot_data.csv};
\addplot[draw=none,name path=band0vix90hi,forget plot] table[x=horizon,y=log_uncertainty_vix_upper90,col sep=comma] {logdpp30_var_irfs_vs_benchmark_dual_plot_data.csv};
\addplot[draw=none,fill=black!8,forget plot] fill between[of=band0vix90lo and band0vix90hi];
\addplot[draw=none,name path=band0dpp90lo,forget plot] table[x=horizon,y=log_uncertainty_dpp_lower90,col sep=comma] {logdpp30_var_irfs_vs_benchmark_dual_plot_data.csv};
\addplot[draw=none,name path=band0dpp90hi,forget plot] table[x=horizon,y=log_uncertainty_dpp_upper90,col sep=comma] {logdpp30_var_irfs_vs_benchmark_dual_plot_data.csv};
\addplot[draw=none,fill=dppred!15,forget plot] fill between[of=band0dpp90lo and band0dpp90hi];
\addplot[draw=none,name path=band0dpp68lo,forget plot] table[x=horizon,y=log_uncertainty_dpp_lower68,col sep=comma] {logdpp30_var_irfs_vs_benchmark_dual_plot_data.csv};
\addplot[draw=none,name path=band0dpp68hi,forget plot] table[x=horizon,y=log_uncertainty_dpp_upper68,col sep=comma] {logdpp30_var_irfs_vs_benchmark_dual_plot_data.csv};
\addplot[draw=none,fill=dppred!32,forget plot] fill between[of=band0dpp68lo and band0dpp68hi];
\addplot[black!40,line width=0.45pt,forget plot] coordinates {(0,0) (48,0)};
\addplot[vixblack,line width=0.95pt,dash pattern=on 4pt off 2pt] table[x=horizon,y=log_uncertainty_vix_point,col sep=comma] {logdpp30_var_irfs_vs_benchmark_dual_plot_data.csv};
\addlegendentry{VIX}
\addplot[dppred,line width=1.05pt] table[x=horizon,y=log_uncertainty_dpp_point,col sep=comma] {logdpp30_var_irfs_vs_benchmark_dual_plot_data.csv};
\addlegendentry{DPP}
\addlegendimage{area legend,fill=black!10,draw=none}\addlegendentry{90\% bands}
\addlegendimage{area legend,fill=dppred!32,draw=none}\addlegendentry{DPP 68\%}
\nextgroupplot[xmin=0,xmax=48,ymin=-0.0836200175366,ymax=0.369423295386,xtick={0,12,24,36,48},ytick={0,0.15,0.3},title={(b) Unemployment},ylabel={Percentage points},xlabel={},xticklabels={}]
\addplot[draw=none,name path=band1vix90lo,forget plot] table[x=horizon,y=unemployment_vix_lower90,col sep=comma] {logdpp30_var_irfs_vs_benchmark_dual_plot_data.csv};
\addplot[draw=none,name path=band1vix90hi,forget plot] table[x=horizon,y=unemployment_vix_upper90,col sep=comma] {logdpp30_var_irfs_vs_benchmark_dual_plot_data.csv};
\addplot[draw=none,fill=black!8,forget plot] fill between[of=band1vix90lo and band1vix90hi];
\addplot[draw=none,name path=band1dpp90lo,forget plot] table[x=horizon,y=unemployment_dpp_lower90,col sep=comma] {logdpp30_var_irfs_vs_benchmark_dual_plot_data.csv};
\addplot[draw=none,name path=band1dpp90hi,forget plot] table[x=horizon,y=unemployment_dpp_upper90,col sep=comma] {logdpp30_var_irfs_vs_benchmark_dual_plot_data.csv};
\addplot[draw=none,fill=dppred!15,forget plot] fill between[of=band1dpp90lo and band1dpp90hi];
\addplot[draw=none,name path=band1dpp68lo,forget plot] table[x=horizon,y=unemployment_dpp_lower68,col sep=comma] {logdpp30_var_irfs_vs_benchmark_dual_plot_data.csv};
\addplot[draw=none,name path=band1dpp68hi,forget plot] table[x=horizon,y=unemployment_dpp_upper68,col sep=comma] {logdpp30_var_irfs_vs_benchmark_dual_plot_data.csv};
\addplot[draw=none,fill=dppred!32,forget plot] fill between[of=band1dpp68lo and band1dpp68hi];
\addplot[black!40,line width=0.45pt,forget plot] coordinates {(0,0) (48,0)};
\addplot[vixblack,line width=0.95pt,dash pattern=on 4pt off 2pt,forget plot] table[x=horizon,y=unemployment_vix_point,col sep=comma] {logdpp30_var_irfs_vs_benchmark_dual_plot_data.csv};
\addplot[dppred,line width=1.05pt,forget plot] table[x=horizon,y=unemployment_dpp_point,col sep=comma] {logdpp30_var_irfs_vs_benchmark_dual_plot_data.csv};
\nextgroupplot[xmin=0,xmax=48,ymin=-0.331838706402,ymax=0.194961666763,xtick={0,12,24,36,48},ytick={-0.3,-0.15,0,0.15},title={(c) Year-on-year inflation},ylabel={Percentage points},xlabel={Months after shock}]
\addplot[draw=none,name path=band2vix90lo,forget plot] table[x=horizon,y=inflation_yoy_vix_lower90,col sep=comma] {logdpp30_var_irfs_vs_benchmark_dual_plot_data.csv};
\addplot[draw=none,name path=band2vix90hi,forget plot] table[x=horizon,y=inflation_yoy_vix_upper90,col sep=comma] {logdpp30_var_irfs_vs_benchmark_dual_plot_data.csv};
\addplot[draw=none,fill=black!8,forget plot] fill between[of=band2vix90lo and band2vix90hi];
\addplot[draw=none,name path=band2dpp90lo,forget plot] table[x=horizon,y=inflation_yoy_dpp_lower90,col sep=comma] {logdpp30_var_irfs_vs_benchmark_dual_plot_data.csv};
\addplot[draw=none,name path=band2dpp90hi,forget plot] table[x=horizon,y=inflation_yoy_dpp_upper90,col sep=comma] {logdpp30_var_irfs_vs_benchmark_dual_plot_data.csv};
\addplot[draw=none,fill=dppred!15,forget plot] fill between[of=band2dpp90lo and band2dpp90hi];
\addplot[draw=none,name path=band2dpp68lo,forget plot] table[x=horizon,y=inflation_yoy_dpp_lower68,col sep=comma] {logdpp30_var_irfs_vs_benchmark_dual_plot_data.csv};
\addplot[draw=none,name path=band2dpp68hi,forget plot] table[x=horizon,y=inflation_yoy_dpp_upper68,col sep=comma] {logdpp30_var_irfs_vs_benchmark_dual_plot_data.csv};
\addplot[draw=none,fill=dppred!32,forget plot] fill between[of=band2dpp68lo and band2dpp68hi];
\addplot[black!40,line width=0.45pt,forget plot] coordinates {(0,0) (48,0)};
\addplot[vixblack,line width=0.95pt,dash pattern=on 4pt off 2pt,forget plot] table[x=horizon,y=inflation_yoy_vix_point,col sep=comma] {logdpp30_var_irfs_vs_benchmark_dual_plot_data.csv};
\addplot[dppred,line width=1.05pt,forget plot] table[x=horizon,y=inflation_yoy_dpp_point,col sep=comma] {logdpp30_var_irfs_vs_benchmark_dual_plot_data.csv};
\nextgroupplot[xmin=0,xmax=48,ymin=-0.311583126579,ymax=0.0917591503261,xtick={0,12,24,36,48},ytick={-0.3,-0.15,0},title={(d) 3-month Treasury bill rate},ylabel={Percentage points},xlabel={Months after shock}]
\addplot[draw=none,name path=band3vix90lo,forget plot] table[x=horizon,y=tbill3m_vix_lower90,col sep=comma] {logdpp30_var_irfs_vs_benchmark_dual_plot_data.csv};
\addplot[draw=none,name path=band3vix90hi,forget plot] table[x=horizon,y=tbill3m_vix_upper90,col sep=comma] {logdpp30_var_irfs_vs_benchmark_dual_plot_data.csv};
\addplot[draw=none,fill=black!8,forget plot] fill between[of=band3vix90lo and band3vix90hi];
\addplot[draw=none,name path=band3dpp90lo,forget plot] table[x=horizon,y=tbill3m_dpp_lower90,col sep=comma] {logdpp30_var_irfs_vs_benchmark_dual_plot_data.csv};
\addplot[draw=none,name path=band3dpp90hi,forget plot] table[x=horizon,y=tbill3m_dpp_upper90,col sep=comma] {logdpp30_var_irfs_vs_benchmark_dual_plot_data.csv};
\addplot[draw=none,fill=dppred!15,forget plot] fill between[of=band3dpp90lo and band3dpp90hi];
\addplot[draw=none,name path=band3dpp68lo,forget plot] table[x=horizon,y=tbill3m_dpp_lower68,col sep=comma] {logdpp30_var_irfs_vs_benchmark_dual_plot_data.csv};
\addplot[draw=none,name path=band3dpp68hi,forget plot] table[x=horizon,y=tbill3m_dpp_upper68,col sep=comma] {logdpp30_var_irfs_vs_benchmark_dual_plot_data.csv};
\addplot[draw=none,fill=dppred!32,forget plot] fill between[of=band3dpp68lo and band3dpp68hi];
\addplot[black!40,line width=0.45pt,forget plot] coordinates {(0,0) (48,0)};
\addplot[vixblack,line width=0.95pt,dash pattern=on 4pt off 2pt,forget plot] table[x=horizon,y=tbill3m_vix_point,col sep=comma] {logdpp30_var_irfs_vs_benchmark_dual_plot_data.csv};
\addplot[dppred,line width=1.05pt,forget plot] table[x=horizon,y=tbill3m_dpp_point,col sep=comma] {logdpp30_var_irfs_vs_benchmark_dual_plot_data.csv};
\end{groupplot}
\end{tikzpicture}%
\endgroup%
	\caption{Monthly responses to uncertainty innovations measured
		using the DPP and VIX.
		The solid line denotes the DPP specification and the dashed
		line the VIX specification.
		Each shock is normalized to its own one-standard-deviation
		innovation.
		The specifications use January 1990--September 2025,
		twelve lags, a constant, and uncertainty-first recursive
		identification.
		The shaded regions show pointwise 90 percent bootstrap
		intervals for both specifications and a darker 68 percent
		interval for the DPP.
		The uncertainty response is in log units; the remaining
		responses are in percentage points.
		Intervals use 1,000 paired residual-bootstrap replications.}
	\label{fig:logdpp-irfs}
\end{figure}

Figure~\ref{fig:logdpp-irfs} shows closely aligned impulse responses,
with some differences in the inflation response.
For example, under the own-standard-deviation normalization,
the inflation-response difference is negative with its
90 percent interval excluding zero at months sixteen through
twenty.
The corresponding unemployment and interest-rate difference
intervals contain zero throughout the 48 months.

As a robustness, I consider a common AIC-selected lag order and differing
sample endpoints.
Table~\ref{Tab:MacroInnovationCorrelations} shows that the
monthly innovation correlation remains above 0.981 across
these specifications.
The close correspondence therefore survives changing the lag
order, using an endpoint close to the Leduc--Liu sample, and
excluding the pandemic period.

\begin{table}[h!]
	\centering
	\small
	\caption{Correlations between Identified DPP and
		Implied-Volatility Innovations}
	\label{Tab:MacroInnovationCorrelations}
	\begin{tabular}{lrrr}
		\toprule
		Specification and benchmark & Lags & Observations
		& Correlation \\
		\midrule
		Monthly: VIX, baseline             & 12 & 417 & 0.9836 \\
		Monthly: VIX, common AIC order     & 14 & 415 & 0.9834 \\
		Monthly: VIX, through October 2013 & 12 & 274 & 0.9839 \\
		Monthly: VIX, through December 2019& 12 & 348 & 0.9814 \\
		\addlinespace
		Quarterly: matched VIX            &  4 &  96 & 0.9845 \\
		Quarterly: author VIX             &  4 &  96 & 0.9775 \\
		Quarterly: author VXO             &  4 &  96 & 0.9867 \\
		\bottomrule
	\end{tabular}
	\begin{flushleft}
		\footnotesize\textit{Note:}
		Entries are Pearson correlations between standardized
		recursively identified uncertainty innovations.
		Observation counts refer to fitted observations after lags.
		The baseline monthly sample retains the September 2025 endpoint of
        the saved macroeconomic replication; extending the option sample
        through July 2026 does not re-estimate this VAR beyond that endpoint.
		All quarterly comparisons use 1990Q1--2014Q4.
		Matched VIX uses the same daily observations and monthly
		aggregation as the DPP; author VIX and VXO retain the
		published quarterly series.
	\end{flushleft}
\end{table}

\subsubsection{Quarterly VAR of Basu and Bundick (2017)}
\label{App:QuarterlyVAR}

\begin{figure}[h]
	\centering
	\input{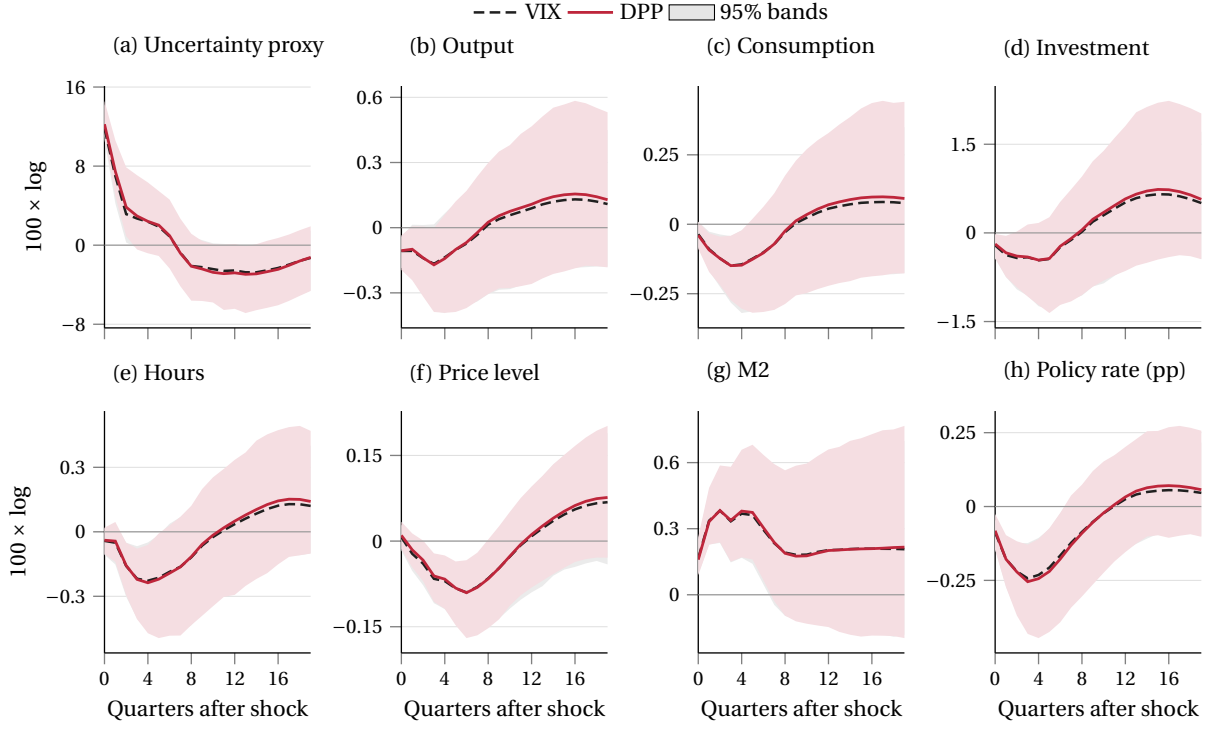}
	\caption{Quarterly DPP and VIX responses in the
		Basu--Bundick specification.
		The solid line denotes the DPP and the dashed line
		the matched VIX specification.
		Both use the same original macroeconomic controls,
		1990Q1--2014Q4, four lags, a constant, and
		uncertainty-first recursive identification.
		Shading denotes modelwise pointwise 95 percent
		posterior intervals based on 1,000 draws.
		Shocks are normalized to each model's own
		one-standard-deviation uncertainty innovation.
		Responses of logarithmic variables are multiplied
		by 100; the policy-rate response is in percentage points.
		The direct VXO comparison is reported separately in
		the text and numerical replication outputs.}
	\label{fig:basu_dpp_irfs}
\end{figure}

The quarterly exercise uses the eight-variable specification
of \citet{BasuBundick2017}.
The variables are ordered as follows: log uncertainty,
log real output per capita, log real consumption per capita,
log real investment per capita, log hours per capita,
log price level, log M2, and the shadow policy rate.
The original macroeconomic observations and transformations
are retained in every comparison.
Each specification includes a constant and four lags, with
uncertainty ordered first.
The common sample, 1990Q1--2014Q4, contains 100 quarters and
96 fitted observations.
I first reproduce the authors' VIX point estimates.
The maximum absolute discrepancy across the stored impulse
responses is below $10^{-10}$.
I then replace the uncertainty series with the DPP while holding the
seven macroeconomic controls and all other specification
choices fixed.
The comparisons include the authors' original VIX, a matched
VIX series, and the authors' original VXO.
For the matched VIX and DPP, quarterly levels are equal-weighted
averages of the three monthly arithmetic means, followed by
the logarithmic transformation.
Uncertainty band construction follows the authors' simulation
conventions; however, I use 1,000 posterior draws per model, compared with 5,000 in
the original simulation.

The DPP innovations correlate 0.9845 with the matched VIX
innovations, 0.9775 with the authors' original VIX innovations,
and 0.9867 with the authors' VXO innovations.
Figure~\ref{fig:basu_dpp_irfs} shows impulse responses between the DPP
and VIX are nearly identical.
At each horizon, I classify a macroeconomic response according
to whether its interval is entirely positive, entirely negative,
or contains zero.
At the 90 percent level, agreement is 140 of 140 for VIX and
136 of 140 for VXO.
Neither comparison switches a response from positive to negative, although
the VXO comparison differs in whether four intervals include zero.
The three uncertainty proxies produce nearly identical real-activity
troughs: output falls by 0.167--0.182 percent and consumption by
0.149--0.151 percent in quarter three, while investment falls by
0.445--0.464 percent and hours by 0.227--0.244 percent in quarter four.
These responses are measured as 100 times the log deviation following
each model's own one-standard-deviation uncertainty innovation, with
the sample, macroeconomic controls, lag order, and recursive
identification held fixed.

Taken together, the exercises show that the principal
macroeconomic conclusions are robust to replacing the
implied-volatility proxy with the DPP in the specifications
examined here.


\subsubsection{Physical Risk, Valuation, and Benchmark Restrictions}
\label{App:DPPPhysicalMapping}

\paragraph{Mapping physical risk into the priced components.}

For the normalized kernel,
$\mathbb E_t^{\mathbb P}[w_t]=1$ and
$\mathbb E_t^{\mathbb Q}[f(X_t)]
=\mathbb E_t^{\mathbb P}[w_tf(X_t)]$
for integrable payoffs.
With $D_t=\{X_t<1\}$ and $L_t=(1-X_t)^+$, define
\begin{equation}
	p_t=\mathbb P_t(D_t)
	\qquad
	\kappa_t=\mathbb E_t^{\mathbb P}[w_t\mid D_t]
	\qquad
	m_t^{\mathbb P}
	=\mathbb E_t^{\mathbb P}[L_t\mid D_t]
	\label{eq:interpretation_physical_components}
\end{equation}
Assuming $p_t>0$, conditioning on the shortfall event gives
\[
q_t
=p_t\mathbb E_t^{\mathbb P}[w_t\mid D_t]
=p_t\kappa_t
\qquad
a_t
=p_t\mathbb E_t^{\mathbb P}[w_tL_t\mid D_t]
\]
Dividing the second expression by the first and expanding the
conditional product expectation yields
\begin{equation}
	m_t
	=
	m_t^{\mathbb P}
	+
	\frac{
		\operatorname{Cov}_t^{\mathbb P}
		(w_t,L_t\mid D_t)
	}{
		\mathbb E_t^{\mathbb P}[w_t\mid D_t]
	}
	\label{eq:app_dpp_severity_mapping}
\end{equation}

The probability identity also gives
$\Delta\log q_t
=\Delta\log p_t+\Delta\log\kappa_t$.
A constant priced probability need not imply a constant physical
probability or a constant average downside valuation.
Any offset between them must follow from the model's equilibrium;
the kernel cannot be chosen afterward to reconcile an otherwise
unrestricted physical distribution with the option evidence.

All quantities refer to the same forward-normalized payoff.
When the forward changes, the event $D_t$ changes with it.
A change in $p_t$ therefore need not have the same sign as a change
in the probability of a fixed loss relative to the current equity
price or of an adverse macroeconomic outcome.

\paragraph{Variation entirely through priced shortfall frequency.}

If conditional priced severity remains fixed at $\bar m$, then
\begin{equation}
	a_t=q_t\bar m
	\quad\Longrightarrow\quad
	\Delta\log a_t=\Delta\log q_t
	\qquad
	\Delta\log m_t=0
	\label{eq:interpretation_frequency_composition}
\end{equation}
This composition is consistent with positive terminal equity and
the forward normalization.
For example, assign probability $\theta_t\in(0,1)$ to
$X_t=1-\ell$, with $0<\ell<1$, and probability $1-\theta_t$ to
$X_t=1+\ell\theta_t/(1-\theta_t)$.
The mean is one, $q_t=\theta_t$, and $m_t=\ell$.
If this is the sole source of premium variation, the probability
loading is one and the severity loading is zero.
This example restricts the equilibrium distribution.
A primitive disaster-frequency shock need not have this composition,
because it can also change equity payoffs, valuation, and the
forward threshold.

\paragraph{Restricted changes in valuation.}

A restricted valuation experiment clarifies the content of the
severity identity.
Hold the physical payoff distribution and the forward fixed, and
perturb the normalized kernel from $w$ to $w+\epsilon h$,
preserving positivity.
Normalization and the fixed forward require
$\mathbb E^{\mathbb P}[h]
=\mathbb E^{\mathbb P}[hX]=0$.
To raise the premium while leaving shortfall probability unchanged,
the perturbation must also satisfy
\begin{equation}
	\mathbb E^{\mathbb P}[h\mathbf 1_D]=0
	\qquad
	\mathbb E^{\mathbb P}[hL]>0
	\label{eq:interpretation_valuation_experiment}
\end{equation}
Valuation must increase the priced loss without increasing total
priced probability on the shortfall event.
A uniform increase in the normalized kernel throughout $D$ cannot
do this in the fixed-payoff experiment.
These are weighted-average conditions.
They do not require pointwise monotone reweighting, and the general
identities must be used when payoffs, physical probabilities, or the
forward also move.

One useful level comparison is available under an additional
restriction.
If the projected kernel
$g_t(x)=\mathbb E_t^{\mathbb P}[w_t\mid X_t=x]$
is nonincreasing below the forward, worse losses receive weakly
greater weight and
$\operatorname{Cov}_t^{\mathbb P}(w_t,L_t\mid D_t)\geq0$.
Then $m_t\geq m_t^{\mathbb P}$.
This condition must be checked in the model.
It neither follows from the decomposition nor implies that the
valuation contribution rises after a shock.

\paragraph{Joint option-price objects.}

The premium, shortfall probability, and idealized log-contract
volatility summarize different features of the same priced
distribution
\begin{equation}
	a_t=\mathbb E_t^{\mathbb Q}[(1-X_t)^+]
	\qquad
	q_t=\mathbb Q_t(X_t<1)
	\qquad
	v_t^2=-\frac{2}{\tau}
	\mathbb E_t^{\mathbb Q}[\log X_t]
	\label{eq:interpretation_joint_option_objects}
\end{equation}
Under common scale, the premium is exactly proportional to the
scale parameter.
For small dispersion and suitable moment and integrability
conditions, the log-contract expansion gives
$v_t\simeq s_t
\sqrt{\mathbb E^{\mathbb Q}[Z^2]/\tau}$
when $X_t=1+s_tZ$.
Neither relationship identifies the common shape.
Different distributions can share the same premium and probability,
and their shapes can change while preserving these moments and an
approximately proportional movement in $v_t$.

In a disaster mixture, changing disaster frequency and the
dispersion of other outcomes can affect all three objects.
The broad shortfall probability need not rise with disaster
frequency, because normalization to the forward also changes the
location of other outcomes relative to the threshold.

\paragraph{The local Gaussian benchmark.}
\label{App:DPPGaussianBenchmark}

Suppose
$X=1+\mu+\sigma\varepsilon$ under $\mathbb P$,
where $\sigma>0$ and $\varepsilon\sim N(0,1)$.
With $\theta=\mu/\sigma$, let
$w=\exp(-\theta\varepsilon-\theta^2/2)$.
The normal moment-generating function gives
$\mathbb E^{\mathbb P}[w]=1$.
Moreover,
\[
\phi(\varepsilon)
\exp(-\theta\varepsilon-\theta^2/2)
=
\phi(\varepsilon+\theta)
\]
Thus $\varepsilon$ has distribution $N(-\theta,1)$ under
$\mathbb Q$, and $X$ has distribution $N(1,\sigma^2)$.
It follows that
\[
q=\frac12
\qquad
a=\frac{\sigma}{\sqrt{2\pi}}
\qquad
m=m^+=\sigma\sqrt{\frac{2}{\pi}}
\]
Under $\mathbb P$, the shortfall event is
$\{\varepsilon<-\theta\}$.
The truncated-normal expectation
$\mathbb E[\varepsilon\mid\varepsilon<-\theta]
=-\phi(\theta)/\Phi(-\theta)$ gives
\[
p=\Phi(-\theta)
\qquad
m^{\mathbb P}
=
\sigma\left[
\frac{\phi(\theta)}{\Phi(-\theta)}-\theta
\right]
\qquad
\kappa=\frac{q}{p}
=\frac{1}{2\Phi(-\theta)}
\]
This example distinguishes physical from priced shortfall components
under an explicit change of measure.

The equality of the priced and physical variances depends on the
exponential-affine change of measure in this benchmark.
Gaussian primitive innovations alone do not establish Gaussian
equity payoffs or an exponential-affine projected pricing kernel.
Nonlinear policies and valuation can change the shape under either
measure.
The arithmetic-Gaussian construction is also only a local
approximation for equity, since it permits negative terminal prices.

\subsubsection{Volatility Innovations: Estimation, Recovery, and Sample Sensitivity}
\label{app:bdg_recovery}

\paragraph{Innovation construction and monthly option measures.}

The first stage uses the original restricted four-lag VAR of
\citet{BergerDewBeckerGiglio2020SecondMomentNews}.
The variable ordering is log realized volatility, log implied
volatility, the federal funds rate, log manufacturing employment,
and log industrial production.
The two volatility equations exclude lags of the three
macroeconomic variables.
The estimation sample is February 1983--December 2014.
I retain the authors' innovation construction and scale realized
volatility to match the cumulative future-volatility effect of
news over months 1--23, as in their replication code.
The component estimates extend these innovations with local
projections; they do not use the authors' VAR response estimator.

The option horizon is 30 days.
For the final three trading sessions $\mathcal T_M$ of month $M$,
let $D_d=e^{-d_d\tau}$ be the daily carry factor and $\chi_d$ the
DPP.
The monthly objects are
\[
A_M
=
\frac{\sum_{d\in\mathcal T_M}\chi_d}
{\sum_{d\in\mathcal T_M}D_d}
\qquad
q_M
=
\frac{\sum_{d\in\mathcal T_M}D_dq_d}
{\sum_{d\in\mathcal T_M}D_d}
\qquad
m_M=\frac{A_M}{q_M}
\]
All three sessions must have usable observations.
An earlier session never replaces a missing required observation.
The baseline retains 296 of the 300 months in 1990--2014.
September 1992, May 1993, August 1994, and October 1994 are
excluded.
Each daily shortfall is defined relative to that date's forward,
so the monthly responses incorporate changes in the forward
threshold.

\paragraph{Local projections and origin counts.}

At each horizon $h=0,\ldots,24$, the dependent variable is
$\log z_{t+h}-\log z_{t-1}$.
Regressors are both innovations, an intercept, and four lags of
log realized volatility, log implied volatility, the funds rate,
employment growth, and industrial-production growth.
Component observations must be valid at $t-1$, $t$, and $t+h$.
Macroeconomic projections use the same origin dates and regressors,
with cumulative log employment or industrial production as the
outcome.
Identical samples and regressors preserve the response identity
$R_A^j(h)=R_q^j(h)+R_m^j(h)$ for each innovation $j$,
including within every bootstrap draw.

The baseline has 291 impact origins, 275 at twelve months, and
264 at 24 months.
Relative to impact, the twelve-month projection loses the twelve
origins in January--December 2014 because their endpoints exceed
the December 2014 sample boundary.
It also loses September 1991, May 1992, August 1993, and October
1993 because the corresponding $t+12$ option components are
missing.
These four endpoints are September 1992, May 1993, August 1994,
and October 1994, respectively.
Reporting horizons were fixed after the component experiment and
before the macroeconomic estimates were computed.

\paragraph{Bootstrap inference.}

Inference uses 999 circular-block resamples with 24-month blocks.
Resampling supplies frequency weights to intact observations at
their original dates, so leads and lags never cross artificial
block joins.
Each draw re-estimates the first-stage VAR, reconstructs the
innovations, and refits all projections using the same weights.
The tables include uncertainty in the realized-volatility
normalization.
The saved fixed-normalization responses are multiplied by the
ratio of that draw's normalization to the original normalization.
This changes response intervals but leaves component response
ratios unchanged.

Pointwise intervals are bootstrap percentiles.
The simultaneous bands use the bootstrap maximum absolute
deviation from the original estimates, standardized by each
estimate's bootstrap standard deviation, over the stated family.
They are not separately studentized bootstrap-$t$ intervals.
Inference conditions on the recovered daily option components.
Comparisons across recovery methods assess sensitivity without
incorporating option-recovery uncertainty into these bands.

\paragraph{Upside component responses.}

Define $m_M^+=A_M/(1-q_M)$ on the same monthly observations.
The local projections add $\log(1-q_M)$ and $\log m_M^+$ as outcomes,
retaining the original innovations, regressors, masks, and joint bootstrap.
The identity
\[
 R_A^j(h)=R_{1-q}^j(h)+R_{m^+}^j(h)
\]
holds at every horizon and in every retained draw. Adding these outcomes
reproduces the original component and macroeconomic point estimates and
all 999 original joint draws to numerical precision.

\begin{table}[htbp]
\centering\small
\caption{Upside impact responses to volatility innovations}
\label{Tab:AppBDGUpside}
\begin{tabular}{lrr}
\toprule
Innovation & Gain probability $1-q$ & Gain severity $m^+$\\
\midrule
Realized volatility & 0.15 & 8.61\\
 & $[-0.45,1.01]$ & $[3.81,16.73]$\\
Future-volatility news & 0.66 & 10.78\\
 & $[-0.20,1.45]$ & $[9.13,11.88]$\\
\bottomrule
\end{tabular}
\par\smallskip
\begin{minipage}{\textwidth}\footnotesize
\emph{Notes:} Entries are 100 times cumulative log responses on impact,
using 291 origins. Brackets are 95 percent pointwise percentile intervals
from the same 999 joint 24-month circular-block draws. The first stage,
innovations, and volatility normalization are re-estimated in every draw.
Inference conditions on the recovered option series.
\end{minipage}
\end{table}

Both innovations increase conditional priced gains as well as losses.
The finding concerns two-sided conditional magnitudes; it does not
identify an exclusively downside shock or establish invariant shape.

\paragraph{Comparison at equal impact severity.}

Let $R_z^j(h)$ denote the estimated cumulative log response of
outcome $z$ to innovation
$j\in\{\mathrm{RV},\mathrm{news}\}$ under the
equal-future-volatility normalization.
Scaling innovation $j$ to produce a 0.01 increase in impact log
severity gives the macroeconomic response
\begin{equation}
	\widetilde R_z^j(h)
	=
	0.01\frac{R_z^j(h)}{R_m^j(0)}
	\qquad
	z\in\{E,Y^{\mathrm{IP}}\}
	\label{eq:bdg_equal_severity_normalization}
\end{equation}
Here $E$ is manufacturing employment and $Y^{\mathrm{IP}}$ is
industrial production.
Each bootstrap draw uses its own impact severity denominator,
which is positive for both innovations in all 999 draws.
With the volatility normalization re-estimated within each draw,
the minimum denominator is 0.01057 for realized volatility and
0.09057 for news, in log units.
Re-estimating the initial volatility normalization does not
change these ratios, because numerator and denominator receive
the same multiplier.

The reported contrast is
$100[\widetilde R_z^{\mathrm{RV}}(h)
-\widetilde R_z^{\mathrm{news}}(h)]$.
Its simultaneous band covers both macroeconomic contrasts over
horizons 0--24, using the same maximum-standardized-deviation
construction as the other macroeconomic bands.

A common proportional mapping from impact severity to employment
would require
\begin{equation}
	R_E^j(h)=\gamma_hR_m^j(0)
	\qquad
	\text{for both }
	j\in\{\mathrm{RV},\mathrm{news}\}
	\label{eq:bdg_common_proportional_mapping}
\end{equation}
It would therefore imply equal normalized employment responses.
At twelve months, their difference is $-0.065$ in units of 100 times log changes, with
a pointwise interval of $[-0.098,-0.020]$ and a simultaneous
interval of $[-0.124,-0.006]$.
This ratio-based result is qualified by the division-free check below.
Both formulations provide pointwise evidence against the proportional
restriction; their simultaneous conclusions differ at twelve months.
The comparison does not hold later severity responses fixed,
estimate a causal mediation effect, or test a general
sufficiency claim for both at-the-money components or their paths.

\paragraph{Division-free inference for the proportional restriction.}

Write $e_j(h)=R_E^j(h)$ and $b_j=R_m^j(0)$ for cumulative log responses,
with $j\in\{R,N\}$. For nonzero impact severities, the common-mapping null is
\[
 e_R(h)/b_R=e_N(h)/b_N
 \quad\Longleftrightarrow\quad e_R(h)b_N-e_N(h)b_R=0.
\]
For comparable reporting units, the point estimate and bootstrap statistic are
\begin{equation}
 \widehat T_h=\frac{\widehat e_R(h)\widehat b_N-
 \widehat e_N(h)\widehat b_R}{\widehat b_R\widehat b_N},\qquad
 T_h^*=\frac{e_R^*(h)b_N^*-e_N^*(h)b_R^*}
 {\widehat b_R\widehat b_N}.
 \label{eq:app_bdg_cross_product}
\end{equation}
Because $100\times0.01=1$, these statistics are in the same units as the
reported employment contrast for a 0.01 impact log-severity increase.
The denominator is fixed at its point-estimate value in every draw;
it sets reporting units and introduces no random division.

\begin{table}[htbp]
\centering\small
\caption{Twelve-month employment restriction: two inference formulations}
\label{Tab:AppBDGCrossProduct}
\begin{tabular}{lrrr}
\toprule
Formulation & Estimate & Pointwise 95\% interval & Simultaneous 95\% interval\\
\midrule
Ratio & $-0.065$ & $[-0.098,-0.020]$ & $[-0.124,-0.006]$\\
Cross-product, fixed scale & $-0.065$ & $[-0.104,-0.023]$ & $[-0.1305,0.0005]$\\
\bottomrule
\end{tabular}
\par\smallskip
\begin{minipage}{\textwidth}\footnotesize
\emph{Notes:} Both procedures use the same 999 joint draws and 275
origins at twelve months. Simultaneous bands cover employment and
industrial-production restrictions over horizons 0--24, using the
maximum standardized bootstrap deviation. Pointwise intervals are
percentile intervals. Inference conditions on recovered components.
\end{minipage}
\end{table}

The point estimates of $b_R$ and $b_N$ are 0.09033 and 0.12379 in log units.
Every retained draw has positive impact severities, so the cross-product
and ratio have identical signs within each draw. Their bootstrap magnitudes
and therefore simultaneous critical values differ. The cross-product
simultaneous interval narrowly includes zero at twelve months. The result
has pointwise support under both formulations, while simultaneous
significance depends on the formulation. The additional check does not
hold subsequent severity paths or other equilibrium responses fixed.

Let $s_i$ be the bootstrap standard deviation at coordinate $i$ and define
$M_b=\max_{i\in\mathcal F}|(T_{b,i}^*-\widehat T_i)/s_i|$ for the
50-coordinate family $\mathcal F$ (employment and industrial production,
each at horizons 0--24). The 95 percent critical value is the empirical
0.95 quantile of $M_b$. For coordinate $i$, the reported adjusted
$p$-value is
\[
 p_i^{\mathrm{adj}}=
 \frac{1+\sum_{b=1}^{999}\mathbf 1\{M_b\geq|\widehat T_i|/s_i\}}{1000}.
\]
The ratio critical value is 2.80964 and its twelve-month employment
statistic is 3.08809, giving $p^{\mathrm{adj}}=0.024$.
The cross-product values are 3.14906, 3.12350, and 0.053.
These are coordinate-wise family-adjusted $p$-values, not global tests
of the entire path. The original reporting protocol listed horizons
0, 6, 12, and 24 and a mean over horizons 0--6; it did not specify a
sole primary twelve-month employment hypothesis. No smaller outcome
family is selected after inspecting these results.

\paragraph{Two-component mapping and equal-premium comparison.}

With two innovations, an unrestricted mapping
$R_E^j(12)=\beta_qR_q^j(0)+\beta_mR_m^j(0)$ exactly fits the two
employment responses whenever the component-response matrix is nonsingular.
The point estimates are $\beta_q=-1.414$ and $\beta_m=-0.0953$, with
ratio 14.83. Across the joint draws, their central 95 percent ranges are
$[-9.77,12.76]$, $[-0.503,0.619]$, and $[-60.10,82.05]$, respectively.
The determinant changes sign in 276 of 999 draws; inversion becomes
unstable as the compositions approach one another. No draw is trimmed.
These descriptive ranges do not identify causal component effects,
and an exactly fitting two-weight mapping supplies no additional
restriction with which to test component sufficiency.

Replacing $R_m^j(0)$ in
\eqref{eq:bdg_equal_severity_normalization} with $R_A^j(0)$
gives the response to an innovation scaled to raise impact log
premium by 0.01.
The twelve-month employment contrast is $-0.06764$ in units of 100 times log changes,
with a 95\% pointwise interval of $[-0.09960,-0.02166]$ and a
simultaneous interval of $[-0.12718,-0.00810]$.
This exploratory band covers both macroeconomic contrasts over
horizons 0--24 under the premium normalization.
It is a separate family from the severity-normalized bands.
Impact premium denominators are positive in all 999 draws.
The same evidence therefore disfavors a common proportional
mapping from the impact premium response alone to twelve-month
employment.

\paragraph{Direct comparison of severity paths.}

Table~\ref{tab:bdg_severity_paths} reports selected horizons of
$100[R_m^{\mathrm{news}}(h)-R_m^{\mathrm{RV}}(h)]$ under the original
equal-future-volatility normalization, re-estimated in every draw.
The difference changes sign and no horizon from 0 through 24 excludes
zero even pointwise. The simultaneous band covers this one contrast
over all 25 horizons, separately from the macroeconomic bands.
This exploratory comparison was added after inspecting the original
results; it establishes neither a larger or more persistent severity
response to news nor equivalence of the paths. The reporting window
remains 0--24 and the initial volatility normalization uses months 1--23.

\begin{table}[!htbp]
	\centering
	\small
	\setlength{\tabcolsep}{3pt}
	\caption{Severity responses and their difference}
	\label{tab:bdg_severity_paths}
	\begin{tabular*}{\textwidth}{@{\extracolsep{\fill}}lrrrcc@{}}
		\toprule
		& & & \multicolumn{3}{c}{News minus RV}\\
		\cmidrule(l){4-6}
		Horizon & RV & News & Difference & Pointwise & Simultaneous\\
		\midrule
		0 & $9.03$ & $12.38$ & $3.35$
		& $[-6.41,7.96]$ & $[-8.80,15.49]$\\
		6 & $3.67$ & $2.86$ & $-0.82$
		& $[-5.14,4.24]$ & $[-8.29,6.66]$\\
		12 & $2.62$ & $1.85$ & $-0.78$
		& $[-6.39,3.95]$ & $[-9.48,7.93]$\\
		24 & $1.17$ & $2.10$ & $0.93$
		& $[-4.22,6.28]$ & $[-7.88,9.73]$\\
		\bottomrule
	\end{tabular*}
	\par\smallskip
	\begin{minipage}{\textwidth}\footnotesize
		\emph{Notes:} Entries are 100 times cumulative log-severity
		responses.
		Both innovations retain the equal-future-volatility
		normalization, with the RV multiplier re-estimated in each
		of the 999 joint draws.
		Intervals refer to the news-minus-RV difference and are 95\%.
		The simultaneous band covers all horizons 0--24, including
		those not displayed.
		Origins number 291, 281, 275, and 264 at the four displayed
		horizons.
		Differences are formed before rounding.
		This is an exploratory comparison; an insignificant
		difference does not establish equivalence.
		Recovery uncertainty is not included in these bands.
	\end{minipage}
\end{table}

\paragraph{Recovery methods and sample sensitivity.}

The recovery exercise uses the put-spread baseline, call spreads,
a smooth curve, and entropy recovery with lognormal and uniform
reference distributions.
Each method retains its own fitted premium.
I estimate responses on each baseline--alternative intersection
and on the common intersection of all five methods.
Across these comparisons, the impact severity intervals remain
strictly positive for both innovations.

Table~\ref{tab:bdg_pairwise} shows how the associated employment
contrast changes with the available dates under the
equal-future-volatility normalization.
The call and smooth samples preserve coverage, and the
simultaneous contrast remains negative.
The entropy and five-method intersections lose substantial
coverage, especially in 1992--1997, and yield wider bands.
Because the macroeconomic projections do not use recovered
components as regressors, this table measures sample sensitivity
of the macroeconomic contrast.

\begin{table}[!htbp]
	\centering
	\small
	\setlength{\tabcolsep}{3pt}
	\caption{Twelve-month employment contrast across recovery samples}
	\label{tab:bdg_pairwise}
	\begin{tabular*}{\textwidth}{@{\extracolsep{\fill}}lrrccc@{}}
		\toprule
		Sample & Months & Origins & Contrast & Pointwise & Simultaneous\\
		\midrule
		Original baseline & 296 & 275 & $-0.628$
		& $[-0.929,-0.224]$ & $[-1.159,-0.097]$\\
		Baseline--call & 292 & 263 & $-0.645$
		& $[-0.953,-0.209]$ & $[-1.201,-0.090]$\\
		Baseline--smooth & 295 & 273 & $-0.627$
		& $[-0.922,-0.208]$ & $[-1.159,-0.094]$\\
		Entropy: lognormal & 180 & 123 & $-0.687$
		& $[-1.226,0.010]$ & $[-1.992,0.618]$\\
		Entropy: uniform & 181 & 125 & $-0.713$
		& $[-1.387,0.025]$ & $[-2.050,0.624]$\\
		All five methods & 178 & 118 & $-0.681$
		& $[-1.255,0.083]$ & $[-2.059,0.698]$\\
		\bottomrule
	\end{tabular*}
	\par\smallskip
	\begin{minipage}{\textwidth}\footnotesize
		\emph{Notes:} RV-minus-news manufacturing-employment
		responses, multiplied by 100, with normalization
		re-estimated in every draw.
		Origins refer to $h=12$.
		Each entropy row uses its intersection with the baseline.
		Intervals are 95\%.
		Each simultaneous band covers both macroeconomic contrast
		paths over horizons 0--24 within that sample; coverage is
		not joint across the six samples.
		Pairings were fixed before estimation and all are reported.
		Baseline, call, and smooth samples retain 999 draws;
		entropy pairings retain 994 and the five-method sample
		retains 993.
		Rank-deficient resamples are recorded in the replication
		output and excluded jointly within each comparison.
	\end{minipage}
\end{table}

The equal-impact-severity employment contrast also remains
negative on the high-coverage intersections.
Using baseline severity on the call-spread intersection gives
$-0.067$ in units of 100 times log changes at twelve months, with a 95\% simultaneous
interval of $[-0.127,-0.006]$.
Using call-spread severity gives $-0.069$ and
$[-0.133,-0.005]$.
On the smooth-curve intersection, baseline severity gives
$-0.064$ and $[-0.123,-0.006]$; smooth-curve severity gives
$-0.063$ and $[-0.120,-0.007]$.
Each band covers both macroeconomic contrasts over horizons
0--24 within its sample and severity normalization.
Coverage is not joint across samples, methods, or normalizations.
All four comparisons use 999 retained joint draws, with positive
impact severity denominators in every draw.
These added comparisons are exploratory.
Equal-severity comparisons on the lower-coverage entropy
intersections are not reported here.

An additional comparison examines the RV-minus-news difference
in impact probability-to-premium ratios on the identical
five-method sample.
The differences are $0.004$ for put spreads, $0.071$ for call
spreads, $0.084$ for the smooth curve, $0.127$ for entropy with a
lognormal reference, and $0.111$ for entropy with a uniform
reference.
Only the lognormal-reference contrast excludes zero under a
simultaneous band covering all five method-specific differences.
Its interval is $[0.016,0.237]$, compared with
$[-0.002,0.223]$ for the uniform reference.
This comparison was added after inspecting earlier results.
Adjustment across methods does not make it a prespecified test,
and the evidence does not establish a recovery-robust distinction
between the two innovation classes' composition ratios.\footnote{
	Differences in fitted premium responses contribute modestly
	to the entropy ratios.
	Dividing the entropy news-probability responses by the baseline
	premium response gives $-0.114$ and $-0.116$, compared with
	their own-premium ratios of $-0.117$ and $-0.120$.
	This arithmetic comparison does not replace each method's
	internally consistent decomposition.}

\subsection{Supplement to the DPP Decomposition}
\label{App:DPPSection3}

This appendix supplies the derivations, component-recovery details,
stability comparisons, and quotation calculations supporting
Section~\ref{Sec:DecomposeDPP}.
It separates sampling inference conditional on recovered components
from the range of decompositions allowed by quotation intervals.

\subsubsection{Decomposition and Conditional Loss--Gain Asymmetry}

Let $X=S_{t+\tau}/F_{t,t+\tau}$, so that
$\mathbb E_t^{\mathbb Q}[X]=1$.
The adjusted premium is
$a=\mathbb E_t^{\mathbb Q}[(1-X)^+]=qm$.
The premium alone does not identify its components, even after imposing
the forward-mean restriction and positive terminal values.
For any $a\in(0,1)$ and $q\in(a,1)$, assign probability $q$ to
$X=1-a/q$ and probability $1-q$ to $X=1+a/(1-q)$.
Both values are positive, the mean is one, and the expected shortfall
is $a$, while the shortfall probability can take any such $q$.
For a strictly positive terminal value, $0<m<1$ whenever a shortfall
has positive probability, so $a<q$.

\paragraph{Probability levels and severity asymmetry.}
\label{App:DPPSeverityAsymmetry}

Since $\mathbb E_t^{\mathbb Q}[X_t-1]=0$,
\[
\mathbb E_t^{\mathbb Q}[(1-X_t)^+]
=
\mathbb E_t^{\mathbb Q}[(X_t-1)^+]
\]
Define
$m_t^+
=\mathbb E_t^{\mathbb Q}[X_t-1\mid X_t\geq1]$.
For $0<q_t<1$, conditioning on the two events gives
\[
q_tm_t=(1-q_t)m_t^+
\qquad
\frac{m_t}{m_t^+}=\frac{1-q_t}{q_t}
\]
These give the two identities in \eqref{eq:dpp_upside_identity}.
The identity remains valid with an atom at the forward when that
atom is included in the conditioning event defining $m_t^+$.
For a continuous distribution, $q_t=1/2$ places the forward at a
median.
It does not imply symmetry.
Conversely, symmetry about the forward implies $q_t=1/2$ when
there is no atom there.

The ratios evaluated at $q=0.423$ and $q=0.444$ are approximately
1.364 and 1.252.
These calculations transform the reported group mean probabilities.
They are not group means of the date-specific severity ratios.
For any common set of averaging weights,
\begin{equation}
	\mathbb E\left[\frac{m_t}{m_t^+}\right]
	=
	\mathbb E\left[\frac{1}{q_t}\right]-1
	\geq
	\frac{1}{\mathbb E[q_t]}-1
	\label{eq:app_dpp_ratio_aggregation}
\end{equation}
provided the expectation exists.
The inequality follows from convexity of $1/q$.
Differences in the dispersion of $q_t$ across groups can affect the
mean ratio, so an ordering of group mean probabilities does not by
itself establish the reverse ordering of group mean severity ratios.

\subsubsection{Scale Benchmarks and Local Option Information}

Proposition~\ref{prop:dpp_scale_severity} supplies an exact additive
common-scale benchmark. The following derivations distinguish it from
lognormal scale variation and a fixed nonlinear payoff map.

\paragraph{The lognormal benchmark.}
\label{App:DPPLognormalBenchmark}

Let
$X=\exp(-s^2/2+s\varepsilon^{\mathbb Q})$,
with $s>0$ and $\varepsilon^{\mathbb Q}\sim N(0,1)$.
Then $\mathbb E^{\mathbb Q}[X]=1$ and
$X<1$ exactly when $\varepsilon^{\mathbb Q}<s/2$.
Therefore
\[
q=\Phi(s/2)
\qquad
\mathbb E^{\mathbb Q}
[X\mathbf 1_{\{X<1\}}]
=
\Phi(-s/2)
\]
Subtracting the second expression from the first gives
$a=2\Phi(s/2)-1$.
Consequently,
\begin{equation}
	q=\frac{1+a}{2}
	\qquad
	m=\frac{2a}{1+a}
	\qquad
	m^+=\frac{2a}{1-a}
	\qquad
	\frac{m}{m^+}=\frac{1-a}{1+a}
	\label{eq:app_dpp_lognormal_proof}
\end{equation}
The associated within-family elasticities satisfy
\begin{equation}
	\frac{d\log q}{d\log a}=\frac{a}{1+a}
	\qquad
	\frac{d\log m}{d\log a}=\frac{1}{1+a}
	\label{eq:interpretation_lognormal_comparison}
\end{equation}
These are analytical elasticities within the lognormal family.
They are not estimates of the aggregate covariance loading in
Section~\ref{Sec:DecomposeDPP}.

For small $a$, severity accounts for nearly all premium variation.
Nevertheless, $q>1/2$ and $m/m^+<1$ for every $s>0$.
A probability below one half cannot be obtained by changing the
volatility parameter within this family.
The lognormal family is not an exact additive common-scale family
around the forward: both $q$ and standardized shape change with
$s$.
\paragraph{The local implied-volatility slope.}
\label{App:DPPLocalSlope}

At a fixed maturity, write normalized strike as $k=K/F$ and total
implied standard deviation as $\omega(\ell)$, where $\ell=\log k$.
The normalized Black--Scholes representation of a put price is
\[
 \mathcal P(k)=k\Phi(-d_2)-\Phi(-d_1),\qquad
 d_1=\frac{-\log k}{\omega}+\frac{\omega}{2},\qquad
 d_2=d_1-\omega.
\]
This represents implied volatility, without assuming that the actual
priced distribution is lognormal. At $k=1$, the price is
$a=2\Phi(\omega(0)/2)-1$. For a differentiable arbitrage-free price curve,
with no atom at the forward, total differentiation gives
\[
 q=\mathcal P'(1)=\Phi(\omega(0)/2)
 +\phi(\omega(0)/2)\omega'(0).
\]
Define the premium-adjusted local probability gap by
\begin{equation}
 \Delta_t^\tau\equiv q_t^\tau-\frac{1+a_t^\tau}{2}.
 \label{eq:dpp_local_gap}
\end{equation}
The preceding derivative yields
\begin{equation}
 \Delta_t^\tau=
 \phi\!\left(\frac{\omega_t(0)}{2}\right)
 \left.\frac{\partial\omega_t(\ell)}{\partial\ell}\right|_{\ell=0}.
 \label{eq:dpp_local_smile_slope}
\end{equation}
This is a single-maturity identity.
The finite-strike estimator approximates the local derivative; maturity
and temporal mixtures preserve the premium decomposition but need not
have the implied-volatility slope of an exactly traded target-horizon claim.
The gap is distinct from global third-moment skewness.
\paragraph{Fixed payoff curvature.}
\label{App:DPPFixedCurvature}

\begin{proposition}[Fixed curvature and the probability response]
\label{prop:app_fixed_curvature}
Let $\varepsilon$ have a symmetric, bounded-support distribution under $\mathbb Q$, with $\mathbb E[\varepsilon^2]=1$ and a twice continuously differentiable density $f$ near zero, where $f(0)>0$. Write $\mu=\mathbb E[(-\varepsilon)^+]$. For a fixed coefficient $c$ and sufficiently small $\sigma>0$, define
\begin{equation}
X_\sigma=1+\sigma\varepsilon+\frac{c}{2}\sigma^2(\varepsilon^2-1).
\label{eq:app_curved_payoff}
\end{equation}
Then $X_\sigma$ is positive and has priced mean one. Its probability, premium, and local probability elasticity satisfy
\begin{align}
q(\sigma)&=\frac12+\frac{c f(0)}2\sigma+O(\sigma^3),
& a(\sigma)&=\mu\sigma+O(\sigma^3),
\label{eq:app_curvature_moments}\\
\eta_q(\sigma)&\equiv\frac{d\log q(\sigma)}{d\log a(\sigma)}
=\frac{q(\sigma)-1/2}{q(\sigma)}+O(\sigma^3).
\label{eq:app_curvature_elasticity}
\end{align}
For $c<0$, sufficiently small increases in scale both deepen conditional losses and lower the shortfall probability.
\end{proposition}
\begin{proof}
Bounded support ensures positivity and monotonicity in $\varepsilon$ for sufficiently small $\sigma$. Symmetry and unit variance give $\mathbb E[X_\sigma]=1$. The threshold $e_\sigma$ satisfying $X_\sigma=1$ solves $e_\sigma+(c\sigma/2)(e_\sigma^2-1)=0$, hence $e_\sigma=c\sigma/2+O(\sigma^3)$. Symmetry gives $f'(0)=0$, so $q(\sigma)=\Pr(\varepsilon<e_\sigma)=1/2+c f(0)\sigma/2+O(\sigma^3)$. In the shortfall integral, the coefficient of $\sigma^2$ over $\varepsilon<0$ vanishes because $\mathbb E[(\varepsilon^2-1)\mathbf 1\{\varepsilon<0\}]=0$. Moving the integration boundary from zero to $e_\sigma$ adds $O(\sigma^3)$, yielding $a(\sigma)=\mu\sigma+O(\sigma^3)$. Differentiation gives $q'(\sigma)=c f(0)/2+O(\sigma^2)$ and $a'(\sigma)=\mu+O(\sigma^2)$, which imply~\eqref{eq:app_curvature_elasticity}. For $c<0$, $q'<0$ and $a'>0$ locally, so $m=a/q$ rises.
\end{proof}

The proposition is local and concerns a fixed distribution under the pricing
measure. It does not assert that the conditional law generated by a general
second-order equilibrium solution belongs to this family. Nor does a local
elasticity equal the regression coefficient over a finite sample unless
additional conditions make the relationship log-linear over the relevant range.

For the exact monthly mappings in Table~\ref{Tab:DPPBenchmarkComparison},
I retain all 439 observed $A_M$ values, define each counterfactual $q_M$
observation by observation, and re-estimate the same loading on 438 adjacent
changes and the same recession gap on 36 recession and 403 expansion months.
The constant-gap calibration uses
$\delta=439^{-1}\sum_M[q_M-(1+A_M)/2]$.
No distribution for quotation errors is imposed. These counterfactual
statistics are diagnostics conditional on observed premiums; the table does
not provide a joint sampling test for calibrated model parameters.
\paragraph{Finite-change benchmark diagnostics.}

The probability level and recession comparison distinguish restrictions
that a near-unit severity loading alone cannot.
A single common shape spanning the full sample implies constant $q$ and
therefore $\mathcal G(q)=0$, conflicting with the positive estimated gap.
An exactly lognormal priced distribution implies $q>1/2$ on every date,
conflicting with the recovered phase means below one-half.
Holding the premium-adjusted local gap $\Delta=q-(1+A)/2$ constant also
imposes a testable restriction. Its estimated monthly recession increase
is 1.391 percentage points, with a 95 percent HAC interval of
$[0.258,2.524]$, inconsistent with a constant gap.
These comparisons condition on component recovery.

For a finite-change comparison, I apply each exact probability mapping to
the observed monthly adjusted premiums and use the same regressions and
recession dates. The resulting probability loadings and log gaps are
reported in Table~\ref{Tab:DPPBenchmarkComparison}.
This calculation avoids replacing a sample regression coefficient by a
local derivative evaluated at a sample mean.

\begin{table}[htbp]
\centering\small
\caption{Probability restrictions and finite-change benchmark diagnostics}
\label{Tab:DPPBenchmarkComparison}
\begin{tabular}{lrr}
\toprule
Mapping evaluated on monthly premiums & $\beta_q$ & $\mathcal G(q)$\\
\midrule
Observed components & $-0.032$ & $0.049$\\
Common scale: constant $q$ & $0.000$ & $0.000$\\
Lognormal: $q=(1+A)/2$ & $0.022$ & $0.014$\\
Constant gap: $q=(1+A)/2+\delta$ & $0.026$ & $0.016$\\
\bottomrule
\end{tabular}
\par\smallskip
\begin{minipage}{\textwidth}\footnotesize
\emph{Notes:} Mappings are evaluated on the original 439 monthly adjusted
premiums, with 438 adjacent changes and 36 recession months.
The constant gap $\delta$ equals the sample mean of $q-(1+A)/2$.
The table reports descriptive model-implied statistics, not confidence
intervals or tests treating calibrated parameters as known. The exact
level and phase restrictions discussed in the text supply the corresponding
benchmark comparisons.
\end{minipage}
\end{table}

Fixed curvature gives a different local restriction. At $q=0.425$,
its leading-order elasticity is approximately $-0.176$; at $q=0.409$ it
is approximately $-0.222$. These are local illustrations, not the implied
monthly or quarterly regression coefficients. Even with one varying
parameter, a nonlinear probability mapping generally makes the covariance
loading depend on the entire path of premiums. A comparison with fixed
curvature therefore requires finite-scale predictions and control of
approximation error. The present calculation does not reject perturbation
methods or all one-parameter priced-distribution families.
The global fixed-shape comparison also differs from the within-month
quotation test in Section~\ref{Sec:CommonScaleInterpretation}, which permits a different candidate shape in every month.

\subsubsection{Component Recovery and Aggregation}
\label{App:DPPIdentification}

\paragraph{Identification and convexity bounds.}
For a fixed maturity, write
$P(K)=e^{-r\tau}\mathbb E^{\mathbb Q}[(K-S_{t+\tau})^+]$.
The put payoff is Lipschitz in $K$, so differentiating under the
expectation at a strike without an atom gives
$P'(K)=e^{-r\tau}\mathbb Q(S_{t+\tau}<K)$.
Evaluating at the forward proves
Proposition~\ref{prop:dpp_identification}.
If there is an atom, the left derivative prices the strict shortfall
event, while the right derivative also includes the atom.

Convexity gives, for $K_-<F<K_+$,
\begin{equation}
 e^{r\tau}\frac{P(F)-P(K_-)}{F-K_-}
 \leq q\leq
 e^{r\tau}\frac{P(K_+)-P(F)}{K_+-F}.
 \label{eq:dpp_probability_secants}
\end{equation}
The two expressions are averages of the cumulative probability over
the left and right intervals, respectively.
For positive probability bounds $\underline q$ and $\overline q$,
the corresponding severity bounds are $a/\overline q$ and
$a/\underline q$.
If the AtM price is linearly interpolated from these same two quoted
prices, the two secants coincide by construction.
That coincidence does not establish that the underlying derivative is
known. The quotation calculations instead permit quoted prices to
vary within their bid--ask intervals, conditional on the constructed
AtM premium and original forward.

\paragraph{Maturity and monthly aggregation.}
\label{App:DPPAggregation}
Let $w$ be the original weight on the longer component maturity, with
carry factors $D_s,D_\ell$ and adjusted premiums $a_s,a_\ell$.
Define
\begin{equation}
 D=(1-w)D_s+wD_\ell,\qquad
 q=\frac{(1-w)D_sq_s+wD_\ell q_\ell}{D},\qquad
 m=\frac{\chi}{Dq}.
 \label{eq:app_dpp_maturity_aggregation}
\end{equation}
Then $\chi=Dqm$ exactly, and
$a=\chi/D$ is the carry-adjusted mixture of the two component premiums.
These weights preserve the synthetic daily premium and its
decomposition; they do not identify a separately traded target-maturity law.

Let $\mathcal D_M$ contain the common daily observations in month $M$,
with $n_M$ dates. Since $\chi_j=D_jq_jm_j$, define
\begin{equation}
 \bar D_M=\frac{1}{n_M}\sum_{j\in\mathcal D_M}D_j,\qquad
 q_M=\frac{\sum_{j\in\mathcal D_M}D_jq_j}{\sum_{j\in\mathcal D_M}D_j},\qquad
 m_M=\frac{\sum_{j\in\mathcal D_M}D_jq_jm_j}
            {\sum_{j\in\mathcal D_M}D_jq_j}.
 \label{eq:app_dpp_monthly_aggregation}
\end{equation}
Thus $\bar\chi_M=\bar D_Mq_Mm_M$ and $A_M=q_Mm_M$.
With constant carry, $q_M$ is an arithmetic daily average, while
$m_M$ remains a probability-weighted severity average.
The same construction preserves $A_M=(1-q_M)m_M^+$ on identical dates.
The monthly components describe a mixture of daily forward-looking
distributions, not the distribution of a realized calendar-month return.

\paragraph{Sample and admissibility.}
The original maturity and quotation selection is retained before
checking whether estimates are admissible.
A selected maturity must have $0<\widehat q<1$, $a>0$, and
$a\leq\widehat q$, with a numerical tolerance of $10^{-12}$ on
the last inequality.
Missing brackets or failures of these conditions exclude 34 of the
8,825 premium dates (0.39 percent), including one recession date.
No truncation is applied and the retained daily premiums are unchanged.
All 439 months remain, with 8,791 daily component observations and
438 adjacent monthly changes.
Restricting monthly averages to component dates gives correlations
with the original DPP of 0.99998 in levels and 0.99984 in log changes;
the carry-adjusted correlations are 0.99998 and 0.99982.

\paragraph{Empirical attribution and inference.}
On identical observations, the log-product identity implies
$\beta_q+\beta_m=1$ for regressions with intercepts, and
$\mathcal G(A)=\mathcal G(q)+\mathcal G(m)$ for phase log gaps.
The two component regressions have opposite residuals and equal
slope standard errors. These covariance contributions need not lie
between zero and one.
Baseline intervals use Bartlett Newey--West covariance with six
monthly lags and the finite-sample adjustment $n/(n-2)$.
Recession shares use the delta method with the joint covariance of
the component and premium log gaps.
The severity-minus-probability recession gap is 0.477, with a
95 percent interval of $[0.276,0.679]$.
All intervals condition on recovered components; quotation uncertainty
is assessed separately below.

\subsubsection{Time Stability, Two-Sided Evidence, and Option Horizons}
\label{App:DPPTimeStability}

This subsection reports the time-stability and recession comparisons for
the baseline put-spread decomposition.
The sample contains 439 months from January 1990 through July 2026 and
438 adjacent monthly changes.
The adjusted premium and its components satisfy $A_M=q_Mm_M$ in every
month, with identical daily observations used to construct all three
monthly series.
Severity is the risk-neutral expected proportional shortfall below the
forward, conditional on that event.

\paragraph{Rolling elasticities}

Table~\ref{Tab:DPPRollingStability} reports rolling estimates of the
severity elasticity from regressions of $\Delta\log m_M$ on an
intercept and $\Delta\log A_M$.
The corresponding probability regression uses the same observations,
so the two slope estimates sum to one in each window.
Severity elasticities range from 0.951 to 1.159 in five-year windows,
from 0.987 to 1.097 in ten-year windows, and from 1.005 to 1.064 in
fifteen-year windows.
Figures~\ref{fig:dpp_rolling_60m} and~\ref{fig:dpp_rolling_180m} show
the five- and fifteen-year estimates; the ten-year estimates appear
in Figure~\ref{fig:dpp_rolling_elasticities} in the main text and
are summarized in Table~\ref{Tab:DPPRollingStability}.

Confidence intervals use a Bartlett Newey--West covariance estimator
with six monthly lags and the finite-sample adjustment $n/(n-2)$,
where $n$ is the number of observations in the regression.
They are conditional on the recovered component series.
The windows overlap, and the intervals are pointwise rather than
simultaneous bands.
The lowest lower confidence limit is therefore a descriptive summary
of these intervals, not a simultaneous confidence bound or a formal
test that the coefficients are constant through time.

\begin{table}[!htbp]
	\centering
	\small
	\caption{Severity Elasticities Across Rolling Windows}
	\label{Tab:DPPRollingStability}
	\begin{tabular}{@{}lrrrr@{}}
		\toprule
		Window & Windows & Min.\ $\widehat\beta_m$ & Max.\ $\widehat\beta_m$
		& \shortstack[r]{Lowest pointwise\\95\% lower limit} \\
		\midrule
		5 years & 379 & 0.951 & 1.159 & 0.869 \\
		10 years & 319 & 0.987 & 1.097 & 0.935 \\
		15 years & 259 & 1.005 & 1.064 & 0.957 \\
		\bottomrule
	\end{tabular}
	\par\smallskip
	\begin{minipage}{\textwidth}
		\footnotesize
		\textit{Notes:} Windows contain 60, 120, or 180 adjacent monthly
		changes. Their first endpoints are January 1995, January 2000,
		and January 2005, respectively; all end by July 2026.
		Each regression includes an intercept.
		The last column reports the minimum of the pointwise 95 percent
		lower confidence limits across windows, using six monthly HAC
		lags. These limits do not provide simultaneous coverage.
	\end{minipage}
\end{table}

\paragraph{Sensitivity to the HAC lag length}

Table~\ref{Tab:DPPHACSensitivity} compares three, six, and twelve
monthly HAC lags in the full sample.
Changing the lag length affects the covariance estimates, leaving
the point estimates unchanged.
The severity elasticity remains close to one, the probability
elasticity remains close to zero, and the confidence interval for
the recession contribution contrast $\mathcal G(m)-\mathcal G(q)$
remains strictly positive.
Here $\mathcal G(z)$ is the mean log level of $z_M$ in recession
months minus its mean log level in expansion months, and the severity
share is $100\mathcal G(m)/\mathcal G(A)$.

\begin{table}[!htbp]
	\centering
	\small
	\caption{Sensitivity of Conditional Inference to the HAC Lag Length}
	\label{Tab:DPPHACSensitivity}
	\setlength{\tabcolsep}{4pt}
	\begin{tabular}{@{}lrrrr@{}}
		\toprule
		& & \multicolumn{3}{c}{95\% confidence interval} \\
		\cmidrule(l){3-5}
		Statistic & Estimate & 3 lags & 6 lags & 12 lags \\
		\midrule
		Severity elasticity $\widehat\beta_m$ & $1.032$ & $[0.995,1.070]$ & $[0.994,1.071]$ & $[0.993,1.071]$ \\
		Probability elasticity $\widehat\beta_q$ & $-0.032$ & $[-0.070,0.005]$ & $[-0.071,0.006]$ & $[-0.071,0.007]$ \\
		Probability gap $\mathcal G(q)$ & $0.049$ & $[0.023,0.075]$ & $[0.021,0.077]$ & $[0.021,0.077]$ \\
		Contribution contrast $\mathcal G(m)-\mathcal G(q)$ & $0.477$ & $[0.292,0.662]$ & $[0.276,0.679]$ & $[0.276,0.679]$ \\
		Severity share (\%) & $91.5$ & $[86.7,96.3]$ & $[86.7,96.3]$ & $[86.8,96.1]$ \\
		\bottomrule
	\end{tabular}
	\par\smallskip
	\begin{minipage}{\textwidth}
		\footnotesize
		\textit{Notes:} Elasticities use 438 adjacent monthly changes;
		recession gaps and shares use 439 monthly levels, including
		36 recession and 403 expansion months.
		All regressions include an intercept and use Bartlett
		Newey--West covariance estimates with the $n/(n-2)$ adjustment.
		Contrast intervals retain the covariance between component
		estimates. Share intervals use the delta method and the joint
		covariance of the severity and premium gaps.
		All intervals use normal critical values and are conditional
		on the recovered component series.
	\end{minipage}
\end{table}

\paragraph{Recession dating and omitted episodes}

Table~\ref{Tab:DPPRecessionStability} compares the baseline recession
definition with peak-inclusive dating and with the omission of each
recession in turn.
The baseline assigns months after each business-cycle peak through
the trough to recession; peak months belong to expansions.
Peak-inclusive dating transfers the four peak months to recessions,
raising the severity share from 91.5 to 91.7 percent.
The omission comparisons remove only the recession months of the
specified episode and retain all 403 baseline expansion months.
Severity shares range from 90.0 to 93.3 percent, and every reported
95 percent confidence interval for $\mathcal G(m)-\mathcal G(q)$
lies above zero.
The pooled comparison is therefore not driven by any one of the four
recessions in this sample.

\begin{table}[!htbp]
	\centering
	\small
	\caption{Recession Definitions and Leave-One-Recession-Out Comparisons}
	\label{Tab:DPPRecessionStability}
	\setlength{\tabcolsep}{4pt}
	\begin{tabular}{@{}lrrrrrr@{}}
		\toprule
		Comparison & $N_R/N_E$ & $\mathcal G(A)$ & $\mathcal G(q)$ & $\mathcal G(m)$
		& \shortstack[r]{Severity\\share (\%)}
		& \shortstack[r]{95\% CI for\\$\mathcal G(m)-\mathcal G(q)$} \\
		\midrule
		Baseline & 36/403 & $0.576$ & $0.049$ & $0.527$ & $91.5$ & $[0.276,0.679]$ \\
		Peak inclusive & 40/399 & $0.545$ & $0.045$ & $0.500$ & $91.7$ & $[0.268,0.641]$ \\
		Omit 1990--91 & 28/403 & $0.640$ & $0.055$ & $0.585$ & $91.3$ & $[0.295,0.763]$ \\
		Omit 2001 & 28/403 & $0.620$ & $0.041$ & $0.578$ & $93.3$ & $[0.310,0.764]$ \\
		Omit 2007--09 & 18/403 & $0.461$ & $0.040$ & $0.420$ & $91.3$ & $[0.155,0.605]$ \\
		Omit 2020 & 34/403 & $0.547$ & $0.055$ & $0.492$ & $90.0$ & $[0.236,0.638]$ \\
		\bottomrule
	\end{tabular}
	\par\smallskip
	\begin{minipage}{\textwidth}
		\footnotesize
		\textit{Notes:} $N_R$ and $N_E$ are the numbers of recession
		and expansion months in each comparison.
		$\mathcal G(z)$ is the recession-minus-expansion mean log gap
		on the indicated sample. Shares are computed from unrounded
		gaps as $100\mathcal G(m)/\mathcal G(A)$.
		The omission rows retain baseline recession dating and remove
		the indicated recession months from the sample, without
		reclassifying them as expansion months.
		Confidence intervals use joint six-lag HAC covariance estimates,
		the $n/(n-2)$ adjustment, and normal critical values.
		Calendar-month lag distances are preserved after omission;
		observations on opposite sides of an omitted episode are not
		treated as adjacent months.
	\end{minipage}
\end{table}

\paragraph{Local expansion benchmarks}

Table~\ref{Tab:DPPLocalRecessions} replaces the full expansion sample
with the twelve or twenty-four months immediately preceding each
recession, including the peak month.
For each component, the local gap is its mean log level during the
recession minus its mean log level in the available expansion
reference months.
The table reports these gaps as well as the severity share, so that
the magnitude of the premium gap remains visible.

Severity accounts for 72.5--98.9 percent of these local premium gaps.
The 2001 comparison is sensitive to the reference period: the severity
share is 98.8 percent using twelve preceding months and 72.5 percent
using twenty-four, with premium log gaps of only 0.058 and 0.082,
respectively.
Only seven preceding months are available for the 1990--91 recession,
so both requested reference windows use those same seven months.
These episode comparisons are descriptive and do not include
episode-specific confidence intervals.

\begin{table}[!htbp]
	\centering
	\small
	\caption{Recession Log Premiums Relative to the Preceding Expansion}
	\label{Tab:DPPLocalRecessions}
	\begin{tabular}{@{}lrrrrr@{}}
		\toprule
		Recession & $N_R/N_E$ & $\mathcal G_{\mathrm{local}}(A)$
		& $\mathcal G_{\mathrm{local}}(q)$ & $\mathcal G_{\mathrm{local}}(m)$
		& \shortstack[r]{Severity\\share (\%)} \\
		\midrule
		\multicolumn{6}{@{}l}{\textit{Panel A: Preceding twelve months}} \\
		\addlinespace[0.35em]
		1990--91$^{*}$ & 8/7 & $0.229$ & $0.008$ & $0.221$ & $96.6$ \\
		2001 & 8/12 & $0.058$ & $0.001$ & $0.058$ & $98.8$ \\
		2007--09 & 18/12 & $0.702$ & $0.097$ & $0.606$ & $86.3$ \\
		2020 & 2/12 & $1.228$ & $0.036$ & $1.192$ & $97.1$ \\
		\addlinespace[0.75em]
		\multicolumn{6}{@{}l}{\textit{Panel B: Preceding twenty-four months}} \\
		\addlinespace[0.35em]
		1990--91$^{*}$ & 8/7 & $0.229$ & $0.008$ & $0.221$ & $96.6$ \\
		2001 & 8/24 & $0.082$ & $0.023$ & $0.060$ & $72.5$ \\
		2007--09 & 18/24 & $0.846$ & $0.070$ & $0.776$ & $91.7$ \\
		2020 & 2/24 & $1.211$ & $0.013$ & $1.198$ & $98.9$ \\
		\bottomrule
	\end{tabular}
	\par\smallskip
	\begin{minipage}{\textwidth}
		\footnotesize
		\textit{Notes:} $N_R$ is the number of recession months and
		$N_E$ is the number of available expansion reference months.
		Recession intervals are August 1990--March 1991,
		April--November 2001, January 2008--June 2009, and
		March--April 2020. Peak months belong to the preceding expansion.
		$^{*}$The January 1990 sample start leaves only seven preceding
		months for the first recession, so neither requested reference
		window is complete.
		Shares equal $100\mathcal G_{\mathrm{local}}(m)/
		\mathcal G_{\mathrm{local}}(A)$, calculated from unrounded mean
		log differences. No episode-specific confidence intervals are
		reported.
	\end{minipage}
\end{table}

\paragraph{Monthly Upside Components and Local Asymmetry}
\label{App:DPPUpsideEvidence}

For each of the original 439 monthly observations, define
$m_M^+=A_M/(1-q_M)$ and $\Delta_M=q_M-(1+A_M)/2$.
All transformations precede averaging and estimation. The upside
identity $A_M=(1-q_M)m_M^+$ therefore holds observation by observation.
The local gap is the mixture counterpart of the single-maturity object;
it need not be the slope of an exactly traded 30-day implied-volatility curve.

\begin{table}[htbp]
\centering\small
\caption{Monthly upside and relative-severity evidence}
\label{Tab:AppMonthlyUpside}
\begin{tabular}{lrr}
\toprule
Object & Estimate & 95\% interval\\
\midrule
Gain-probability loading & 0.023 & $[-0.003,0.050]$\\
Gain-severity loading & 0.977 & $[0.950,1.003]$\\
Gain-probability recession log gap & $-0.037$ & $[-0.058,-0.015]$\\
Gain-severity recession log gap & 0.612 & $[0.392,0.833]$\\
Loss-to-gain ratio recession log gap & $-0.086$ & $[-0.135,-0.037]$\\
\bottomrule
\end{tabular}
\par\smallskip
\begin{minipage}{\textwidth}\footnotesize
\emph{Notes:} Loadings use 438 adjacent monthly changes. Recession gaps
use 439 months, including 36 recession months. Regressions, observations,
six-lag Bartlett HAC inference, and the correction $n/(n-2)$ match the
baseline downside analysis. Intervals condition on recovered components.
\end{minipage}
\end{table}

Gain severity averages 3.07 percent in expansions and 5.72 percent in
recessions. The observation-level loss-to-gain ratios average 1.373 and
1.258, respectively. Conditional gains rise more proportionately, while
conditional losses remain larger in phase averages.
The identity
$\mathcal G(m)-\mathcal G(m^+)=\mathcal G((1-q)/q)$ separates relative
conditional magnitudes from the premium's overall movement.

The 30-day monthly gap averages $-0.08569$ in expansions and $-0.07178$
in recessions. Its increase is 1.391 percentage points, with a six-lag
HAC interval of $[0.258,2.524]$ percentage points. The premium term adds
0.704 percentage points to the 2.095-percentage-point probability increase.
This local asymmetry calculation is distinct from global third-moment skewness.

\paragraph{Quarterly-horizon evidence}
\label{App:DPPQuarterlyEvidence}
The quarterly construction uses 91.25 ACT/365 days, exactly one quarter of a year. It retains the 30-day construction's quotation filters, parity forward, and strike and maturity interpolation rules. Put-spread, call-spread, and smooth-curve estimates use the same selected chains, with maturities bracketing the target.

Each observation combines the final three exchange sessions of a completed calendar quarter. All three must have valid components; invalid sessions are left missing. Common dates are selected across methods before carry-weighted aggregation, preserving $a=qm=(1-q)m^+$. The resulting sample has 136 quarter-end observations: 124 in expansions and 12 in recessions, spanning four recession episodes. Completed-quarter comparisons end in 2026Q2.

The mean put-spread premium is 3.55 percent of the forward, with a 95 percent interval of $[3.15,3.97]$ percent. Expansion and recession means are 3.36 and 5.51 percent. These levels provide a horizon-matched target separately from the monthly 30-day series.

\begin{table}[tbp]
\centering\small
\caption{Quarterly-horizon probability and conditional severity asymmetry}
\label{tab:quarterly}
\begin{tabular}{lcccc}
\toprule
 & \multicolumn{2}{c}{Probability $q$} & \multicolumn{2}{c}{Loss-to-gain severity $m/m^+$}\\
\cmidrule(lr){2-3}\cmidrule(lr){4-5}
Method & Expansion & Recession & Expansion & Recession\\
\midrule
Put spread & 0.408 & 0.422 & 1.471 & 1.385\\
 & $[0.396,0.420]$ & $[0.400,0.435]$ & $[1.401,1.545]$ & $[1.304,1.530]$\\
Call spread & 0.410 & 0.430 & 1.458 & 1.338\\
 & $[0.398,0.421]$ & $[0.399,0.443]$ & $[1.390,1.528]$ & $[1.259,1.518]$\\
Smooth curve & 0.408 & 0.428 & 1.461 & 1.353\\
 & $[0.397,0.420]$ & $[0.398,0.444]$ & $[1.391,1.534]$ & $[1.252,1.537]$\\
\bottomrule
\end{tabular}
\par\smallskip\begin{minipage}{\textwidth}\footnotesize\emph{Notes:} Option horizon: 91.25 ACT/365 days. Entries are phase means on 136 common quarter-end observations (124 expansion and 12 recession). Each observation combines the final three exchange sessions, all valid under all methods. The smooth curve has a free AtM anchor. Severity-ratio entries average observation-level $(1-q)/q$; they are not ratios evaluated at average probabilities. Brackets give 95 percent pointwise percentile intervals from 1,999 circular calendar-block draws with eight-quarter blocks, conditional on recovery. Missing quarters remain in the sampling frame. Quarter-end phase uses the final month's NBER status, with peaks in expansions and troughs in recessions. Sample: 1990Q1--2026Q2.\end{minipage}
\end{table}

Expansion probabilities range from 0.408 to 0.410 across recoveries (Table~\ref{tab:quarterly}). Their observation-level conditional loss-to-gain ratios average 1.458--1.471: conditional priced losses are 46--47 percent larger than gains. Recession probabilities range from 0.422 to 0.430, with mean ratios of 1.338--1.385. The probability and ratio columns report the same underlying local information in probability and conditional-loss units.

The level is more precisely established than the phase change. The put-spread recession-minus-expansion probability difference is 0.014, with an interval of $[-0.009,0.030]$; the ratio difference is $-0.086$, with an interval of $[-0.186,0.060]$. All three recoveries have phase-difference intervals containing zero. Alternative sampling conventions give similar probability levels: means of 0.412--0.413 on 9,080 common daily dates and 0.410--0.413 on each method's own quarter-end sample.

The lognormal gap separates the probability level from the premium contribution. The mean quarterly put-spread gap is $-0.109$, with an interval of $[-0.120,-0.098]$; call-spread and smooth-curve means are $-0.106$ and $-0.108$. A lognormal volatility chosen to match each observed premium leaves this probability discrepancy by construction.

The quarterly local-gap change is 0.318 percentage points, with an eight-quarter calendar-block interval of $[-1.819,1.758]$ percentage points. Table~\ref{tab:local_gap_phases} compares it with the monthly 30-day result, then aligns both horizons on identical quarter-end dates. The aligned construction uses the same final three exchange sessions for each horizon and retains 134 quarters, including 12 recession quarters. Its recession gap changes are 2.393 percentage points at 30 days and 0.250 at 91 days. Their directly estimated difference is 2.144 percentage points, with a joint interval of $[0.598,4.057]$. The covariance retained by the paired comparison makes this difference more precise than either phase change alone.

\begin{table}[tbp]
\centering\small
\caption{Local asymmetry across phases and option horizons}
\label{tab:local_gap_phases}
\begin{tabular}{lrrrr}
\toprule
 & Expansion & Recession & Change & 95\% interval\\
Construction & $100\overline\Delta_E$ & $100\overline\Delta_R$ & R minus E & for change\\
\midrule
\multicolumn{5}{l}{A. Original observation rules}\\
30 days, monthly averages & $-8.569$ & $-7.178$ & 1.391 & $[0.258,2.524]$\\
91 days, quarter-end mixture & $-10.892$ & $-10.574$ & 0.318 & $[-1.819,1.758]$\\
\midrule
\multicolumn{5}{l}{B. Identical quarter-end dates at both horizons}\\
30 days & $-8.969$ & $-6.576$ & 2.393 & $[-0.495,4.832]$\\
91 days & $-10.824$ & $-10.574$ & 0.250 & $[-1.881,1.695]$\\
30 minus 91 days & 1.855 & 3.998 & 2.144 & $[0.598,4.057]$\\
\bottomrule
\end{tabular}
\par\smallskip\begin{minipage}{\textwidth}\footnotesize\emph{Notes:} All entries are percentage points of $\Delta=q-(1+a)/2$. Panel A uses 439 monthly observations at 30 days and 136 quarter-end observations at 91.25 ACT/365 days. Its monthly interval uses six-lag HAC; quarterly intervals use eight-quarter circular calendar blocks. Panel B uses baseline put spreads on 134 common quarters, with both horizons combining exactly the same final three exchange sessions; 12 quarters are in recessions. Differences are formed jointly in 1,999 attempted calendar-block draws. Five draws contain no recession observations and are omitted jointly, leaving 1,994 valid draws. Missing calendar quarters remain in the resampling frame. Intervals are pointwise and conditional on recovery.\end{minipage}
\end{table}

The local asymmetry comparison is therefore horizon dependent: the recession change is larger at the shorter horizon on the aligned sample. This evidence concerns a vega-weighted local implied-volatility slope. It does not equate the local gap to a standardized third moment or interpret an imprecise quarterly change as exact acyclicality.

The same construction permits a matched-frequency comovement comparison. Table~\ref{tab:quarterly_loadings} reports the probability loading using either the exact final exchange session or the final-three-session mixture. Only changes between adjacent calendar quarters enter. The exact-date baseline loading is 0.0064, with a complementary severity loading of 0.9936. The three-session baseline gives 0.0358 and 0.9642. The six probability intervals include zero, and the point estimates imply severity loadings between 0.9514 and 1.0226. These results extend severity dominance to quarterly changes; they do not establish equality of the monthly and quarterly coefficients or an exactly constant probability.

\begin{table}[tbp]
\centering\small
\caption{Quarterly probability loadings on protection-premium changes}
\label{tab:quarterly_loadings}
\begin{tabular}{lrrrr}
\toprule
Recovery & Quarters & Changes & $\beta_q$ & 95\% HAC interval\\
\midrule
\multicolumn{5}{l}{Exact final exchange session}\\
Put spread & 142 & 137 & 0.0064 & $[-0.0798,0.0926]$\\
Call spread & 142 & 137 & 0.0182 & $[-0.0707,0.1072]$\\
Smooth curve & 142 & 137 & 0.0486 & $[-0.0005,0.0976]$\\
\addlinespace
\multicolumn{5}{l}{Final-three-session mixture}\\
Put spread & 136 & 128 & 0.0358 & $[-0.0292,0.1008]$\\
Call spread & 136 & 128 & -0.0226 & $[-0.1139,0.0688]$\\
Smooth curve & 136 & 128 & 0.0438 & $[-0.0027,0.0903]$\\
\addlinespace
\bottomrule
\end{tabular}
\par\smallskip\begin{minipage}{\textwidth}\footnotesize\emph{Notes:} All options have a 91.25 ACT/365-day target. Recovery methods use common valid observations within each sampling convention; exact-date and three-session panels have their own eligible quarters. OLS includes an intercept. Bartlett HAC uses four actual calendar-quarter lags and the finite-sample adjustment $n/(n-2)$. Missing quarters are not bridged. Intervals condition on recovered components. Severity loadings equal $1-\beta_q$. On the identical 136-quarter comparison sample, the exact-date smooth estimate is 0.0517 with interval $[0.0034,0.1000]$; this sensitivity is retained in the accompanying results.\end{minipage}
\end{table}

Exact-date observations are the direct empirical counterpart of a quarterly model's date-specific conditional distribution. Three-session observations retain the paper's averaging convention. The exact-date baseline has mean premium 3.5167 percent, mean probability 0.41123, and mean local gap $-0.10636$, compared with 3.5485 percent, 0.40911, and $-0.10864$ in the mixture. A model comparison should choose its corresponding observation rule before assessing a discrepancy.

\paragraph{Component levels and alternative rolling windows}

Figure~\ref{fig:dpp_component_levels} reports the monthly adjusted
premium, shortfall probability, and conditional severity in their
natural units.
Its vertical scales differ across panels; proportional variation
is compared on a common scale in the main text.
The rolling figures use equal vertical ranges centered on the
severity and probability benchmarks of one and zero, respectively.
They provide a visual account of the rolling estimates summarized in
Table~\ref{Tab:DPPRollingStability}.

\begin{figure}[!htbp]
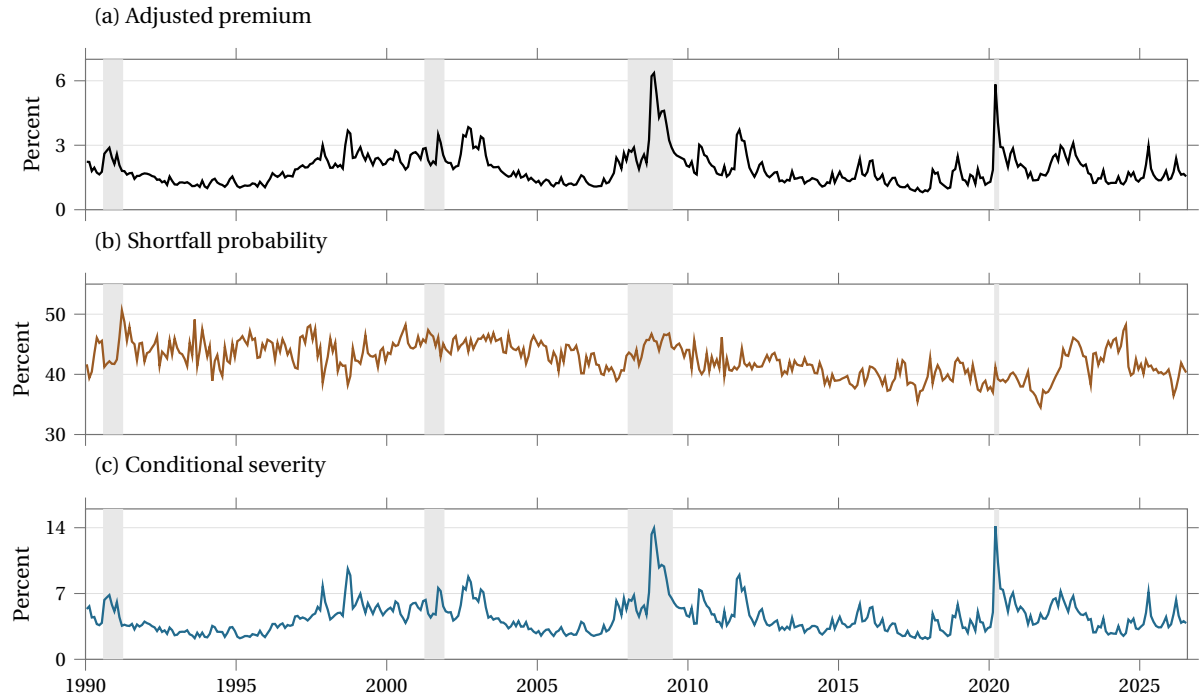

	\centering
	\begingroup%
	\providecommand{\DPPFigureWidth}{\textwidth}%
	\pgfmathsetlengthmacro{\dppW}{\DPPFigureWidth}%
	\pgfmathsetlengthmacro{\dppH}{0.60*\dppW}%
	\definecolor{dppBlue}{HTML}{236B8E}%
	\definecolor{dppBrown}{HTML}{9B5B24}%
	\pgfplotsset{dpp publication/.style={%
			scale only axis, axis lines=box, axis line style={black!55,line width=0.45pt},%
			tick align=outside, tick style={black!55,line width=0.4pt},%
			tick label style={font=\fontsize{8}{9.6}\selectfont},%
			label style={font=\fontsize{9}{10.8}\selectfont},%
			title style={font=\fontsize{9}{10.8}\selectfont,at={(0,1.045)},anchor=south west},%
			legend style={font=\fontsize{8}{9.6}\selectfont,draw=none,fill=none},%
			grid=major, xmajorgrids=false, grid style={black!12,line width=0.3pt},%
			scaled ticks=false, /pgf/number format/1000 sep={},%
			clip=true, unbounded coords=jump%
	}}%
%
	\endgroup%
	\caption{Monthly Premium, Probability, and Severity in Natural Units}
	\label{fig:dpp_component_levels}
	\begin{minipage}{\textwidth}
		\footnotesize
		\textit{Notes:} The panels show $100A_M$, $100q_M$, and $100m_M$,
		respectively, from January 1990 through July 2026.
		Severity is expressed as a percentage of the forward value.
		Vertical scales differ across panels.
		Gray shading denotes recession months, from the month after
		the business-cycle peak through the trough.
		Lines are baseline recovered estimates without component-level
		uncertainty bands.
	\end{minipage}
\end{figure}

\begin{figure}[!htbp]
	\centering
	\begingroup%
	\providecommand{\DPPFigureWidth}{\textwidth}%
	\pgfmathsetlengthmacro{\dppW}{\DPPFigureWidth}%
	\pgfmathsetlengthmacro{\dppH}{0.60*\dppW}%
	\definecolor{dppBlue}{HTML}{236B8E}%
	\definecolor{dppBrown}{HTML}{9B5B24}%
	\pgfplotsset{dpp publication/.style={%
			scale only axis, axis lines=box, axis line style={black!55,line width=0.45pt},%
			tick align=outside, tick style={black!55,line width=0.4pt},%
			tick label style={font=\fontsize{8}{9.6}\selectfont},%
			label style={font=\fontsize{9}{10.8}\selectfont},%
			title style={font=\fontsize{9}{10.8}\selectfont,at={(0,1.045)},anchor=south west},%
			legend style={font=\fontsize{8}{9.6}\selectfont,draw=none,fill=none},%
			grid=major, xmajorgrids=false, grid style={black!12,line width=0.3pt},%
			scaled ticks=false, /pgf/number format/1000 sep={},%
			clip=true, unbounded coords=jump%
	}}%
%
	\endgroup%
	\caption{Probability and Severity Elasticities in Rolling Five-Year Windows}
	\label{fig:dpp_rolling_60m}
	\begin{minipage}{\textwidth}
		\footnotesize
		\textit{Notes:} Each window contains 60 adjacent monthly changes,
		with endpoints from January 1995 through July 2026.
		Lines are slope estimates from regressions with an intercept;
		shaded bands are pointwise 95 percent intervals using six-lag
		Bartlett Newey--West covariance estimates and the $n/(n-2)$
		adjustment, conditional on the recovered series.
		The two panels use equal vertical ranges centered on one and
		zero, with dashed lines marking these benchmarks.
		Gray recession shading refers to the window endpoint, not to
		the composition of the entire estimation window.
		Overlapping windows do not provide simultaneous coverage.
	\end{minipage}
\end{figure}

\begin{figure}[!htbp]
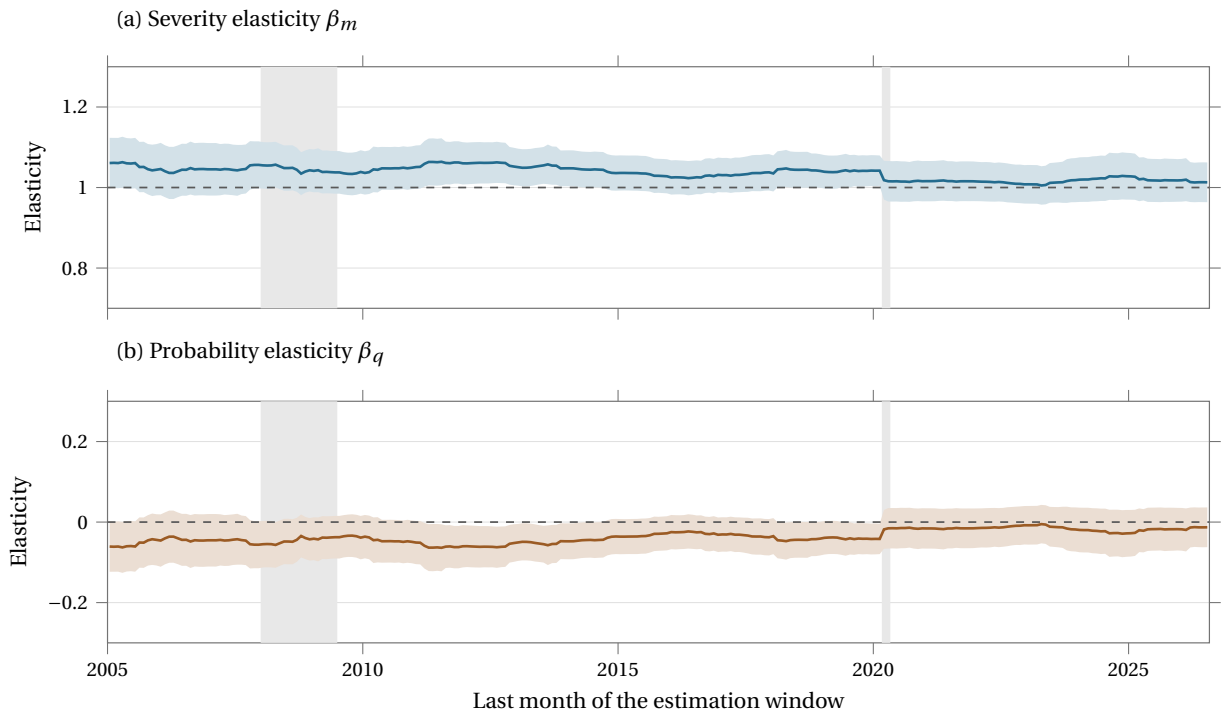

	\centering
	\begingroup%
	\providecommand{\DPPFigureWidth}{\textwidth}%
	\pgfmathsetlengthmacro{\dppW}{\DPPFigureWidth}%
	\pgfmathsetlengthmacro{\dppH}{0.60*\dppW}%
	\definecolor{dppBlue}{HTML}{236B8E}%
	\definecolor{dppBrown}{HTML}{9B5B24}%
	\pgfplotsset{dpp publication/.style={%
			scale only axis, axis lines=box, axis line style={black!55,line width=0.45pt},%
			tick align=outside, tick style={black!55,line width=0.4pt},%
			tick label style={font=\fontsize{8}{9.6}\selectfont},%
			label style={font=\fontsize{9}{10.8}\selectfont},%
			title style={font=\fontsize{9}{10.8}\selectfont,at={(0,1.045)},anchor=south west},%
			legend style={font=\fontsize{8}{9.6}\selectfont,draw=none,fill=none},%
			grid=major, xmajorgrids=false, grid style={black!12,line width=0.3pt},%
			scaled ticks=false, /pgf/number format/1000 sep={},%
			clip=true, unbounded coords=jump%
	}}%
%
	\endgroup%
	\caption{Probability and Severity Elasticities in Rolling Fifteen-Year Windows}
	\label{fig:dpp_rolling_180m}
	\begin{minipage}{\textwidth}
		\footnotesize
		\textit{Notes:} Each window contains 180 adjacent monthly changes,
		with endpoints from January 2005 through July 2026.
		Regression construction, pointwise 95 percent intervals,
		benchmark lines, equal vertical ranges, and endpoint recession
		shading follow Figure~\ref{fig:dpp_rolling_60m}.
	\end{minipage}
\end{figure}

\clearpage
\subsubsection{Robustness of the Probability--Severity Decomposition}
\label{App:DPPRecoveryRobustness}

The analysis distinguishes uncertainty about the recovered components
from sampling uncertainty conditional on those components.
Alternative recovery methods and quotation bounds assess sensitivity to the former;
HAC inference and the joint block bootstrap address the latter.
Every recovery comparison retains the original forward, snapshot,
settlement classification, and selected maturities.
Each method is compared with the baseline on identical usable daily
dates before monthly aggregation, with coverage and failures reported.
Agreement across recoveries does not provide independent protection against
quotation errors shared by the methods. Each recovery retains assumptions
about interpolation, support, or regularization, and all comparisons
condition on the original forward and maturity selection.

The full baseline contains 8,791 daily component observations and
439 months from January 1990 through July 2026.
For each comparison, I first intersect usable daily dates and then
rebuild all monthly series using the carry-adjusted aggregation in
Appendix~\ref{App:DPPAggregation}.
The resulting monthly series satisfy
$A_M^{(j)}=q_M^{(j)}m_M^{(j)}$ for method $j$.
Monthly log changes use adjacent calendar months only; missing months
are not bridged.
The severity elasticity is the OLS coefficient from regressing
$\Delta\log m_M^{(j)}$ on $\Delta\log A_M^{(j)}$ with an intercept.
The recession severity share is
$\mathcal G(m^{(j)})/\mathcal G(A^{(j)})$, where each gap is the
recession mean log level minus the expansion mean log level.
Recessions run from the month after an NBER peak through the trough.

Unless otherwise indicated, confidence intervals use Newey--West
covariance with a Bartlett kernel, six calendar-month lags, and the
finite-sample correction $n/(n-2)$.
Missing calendar months contribute zero regression scores to the
HAC calculation; lag distances are not compressed.
Share intervals use the delta method and the joint covariance matrix
of the component and premium log gaps.
All such intervals are conditional on the recovered component series.

\paragraph{Alternative recovery methods}

I compare the nearest put spread with wider strike brackets,
quote-quality restrictions, and call-side spreads.
The quality checks require positive bids, a relative bid--ask spread
at most 25 percent, or volume of at least five; a separate check
requires a strike bracket no wider than one percent of the forward.
These restrictions operate within the original eligible quotation set.
Call-side spreads provide a quotation-consistency check because
call and put slopes are linked by put--call parity.
For a call bracket $K^-<F<K^+$, the corresponding estimate is
\begin{equation}
    \widehat q^{\mathrm{call}}
    =1+e^{r_tu}
    \frac{C(K^+,u)-C(K^-,u)}{K^+-K^-}.
    \label{eq:app_dpp_call_recovery}
\end{equation}
The nearest call bracket and calls at the baseline put strikes are
examined separately.
Their agreement is a consistency check on quotations linked by
put--call parity; it does not provide an independent source of
pricing information.

For the curve and entropy fits, write normalized strike as $k=K/F$
and normalized put price as $p(k)=e^{r_tu}P(Fk,u)/F$.
The selected quotations are valid, non-crossed puts within
$|k-1|\leq0.20$, retaining the nearest strict strike on each side
of one and requiring at least five strikes.
Duplicate strikes must be resolved before recovery.

The smooth-curve estimator represents the risk-neutral distribution
using nonnegative masses on return intervals, with a uniform density
within each interval.
The implied put-price curve is convex and continuously differentiable.
Mass and mean are constrained to one, and the curve is fitted to
midpoint prices using quotation-spread weights and a curvature penalty.
The default fit uses at most 21 nearby put strikes within 20 percent
of the forward; 11 and 41 strikes provide window sensitivities.
The curvature penalty is also varied by a factor of ten in either
direction.
Because the fit does not require every fitted price to lie inside its
quotation interval, pricing residuals and band violations are reported.

More precisely, let $w_j$ be the mass on interval $[b_{j-1},b_j]$,
let $d_j=b_j-b_{j-1}$, and define its average put payoff by
\begin{equation}
    H_j(k)=\frac{1}{d_j}\int_{b_{j-1}}^{b_j}(k-x)_+\,dx,
    \qquad p_w(k)=\sum_jw_jH_j(k).
    \label{eq:app_dpp_density_prices}
\end{equation}
The restrictions are $w_j\geq0$, $\sum_jw_j=1$, and
$\sum_jw_j(b_{j-1}+b_j)/2=1$.
Given $N$ quotations with midpoint prices $p_i^{\mathrm{mid}}$,
the smooth fit minimizes
\begin{equation}
    \frac{1}{2N}\sum_{i=1}^N
    \left[
       \frac{\widetilde h}{h_i}
       \frac{p_w(k_i)-p_i^{\mathrm{mid}}}{s}
    \right]^2
    +\frac{\lambda}{2}\sum_j\frac{w_j^2}{d_j},
    \label{eq:app_dpp_smooth_objective}
\end{equation}
where $h_i$ is the normalized quotation half-spread with a
$10^{-6}$ numerical floor, $\widetilde h$ is its cross-quote median,
and $s=\max\{a^{\mathrm{baseline}},10^{-5}\}$.
The default penalty is $\lambda=10^{-5}$; the sensitivities use
$10^{-6}$ and $10^{-4}$.
The penalty is the integrated squared density, equivalently the
integrated squared curvature of the implied put-price curve.

The entropy estimator instead minimizes relative entropy over the
interval masses subject to probability normalization, unit mean,
and fitted prices inside every selected bid--ask interval.
Lognormal and uniform reference distributions provide prior sensitivity;
the fitted probabilities remain free to change at every date.
The main support is $[10^{-6},3]$, and a wider, refined specification
uses $[10^{-8},10]$.
Bounds obtained by optimizing over these intervals are conditional
on the specified density family and support.
With positive reference masses $\pi_j$, the entropy objective is
$\sum_jw_j\log(w_j/\pi_j)$, subject to the mass and mean constraints
and $p_i^{\mathrm{bid}}\leq p_w(k_i)\leq p_i^{\mathrm{ask}}$ for
every selected quote.
The reference is normalized over the chosen support, with a
$10^{-14}$ floor on reference masses before normalization; fitted
masses and shortfall probabilities are not floored.
The interval endpoints contain the selected strikes, one, fixed tail
knots, and 33 lognormal-spaced interior knots.
The lognormal scale $v=2\Phi^{-1}((1+a^{\mathrm{baseline}})/2)$
locates the interior knots and defines the lognormal reference.
It does not restrict the fitted distribution to be lognormal.
The wider specification uses 65 interior knots, includes the main
grid and its interval midpoints, and extends the support as stated above.

For both approaches, I distinguish fits that hold the original
adjusted premium fixed from fits that estimate it jointly with
probability.
This distinction matters because the original premium is constructed
by linear interpolation: imposing that premium and both surrounding
midpoints on a convex curve can mechanically reproduce the original
local slope.
The joint fits produce their own coherent premium, probability,
and severity, and the difference between the fitted and original
premium is reported.
The fitted interval masses and shortfall probability are allowed to
change at every date, within the stated support and density family.
All fits are date-specific; no temporal smoothing is imposed.
Successful recovery requires normalized equality residuals at most
$2\times10^{-8}$ and mass negativity no greater than $10^{-10}$;
entropy fits additionally satisfy quotation bands to within
$2\times10^{-8}$.
The smooth solver checks a global convex optimality-gap bound against
$10^{-7}$; entropy fits require a primal--dual gap in
$[-10^{-8},10^{-6}]$.
Infeasible problems and failed numerical certificates are excluded
and recorded rather than assigned component estimates.

\paragraph{Results on common dates}

Table~\ref{Tab:DPPRecoveryCommon2} compares the baseline put spread,
the call spread, the smooth curve, and the two entropy reference distributions
on exactly the same 7,078 daily dates.
This sample contains 408 usable months, 398 adjacent monthly changes,
and all 36 recession months.
The smooth and entropy estimates in this table jointly fit the premium
and probability, so agreement does not require holding the original
premium fixed.

\begin{table}[p!]
\centering
\small
\setlength{\tabcolsep}{3pt}
\renewcommand{\arraystretch}{1.12}
\caption{Component Recovery on Identical Daily Samples}
\label{Tab:DPPRecoveryCommon2}
\begin{tabular*}{\textwidth}{@{\extracolsep{\fill}}lrrrr@{}}
\toprule
Method & \shortstack[r]{Severity\\elasticity}
& \shortstack[r]{Recession severity\\share (\%)}
& \shortstack[r]{$\operatorname{sd}(\Delta\log q_M)$}
& \shortstack[r]{$\operatorname{sd}(\Delta\log m_M)$} \\
\midrule
\multicolumn{5}{@{}l}{\textit{Panel A: All common usable dates}} \\
\addlinespace[0.3em]
Baseline put spread & \shortstack[r]{$1.039$\\$[0.997,\,1.081]$} & \shortstack[r]{$91.2$\\$[86.1,\,96.3]$} & $0.056$ & $0.195$ \\
\addlinespace[0.25em]
Call spread & \shortstack[r]{$1.055$\\$[1.017,\,1.092]$} & \shortstack[r]{$92.0$\\$[87.2,\,96.9]$} & $0.054$ & $0.197$ \\
\addlinespace[0.25em]
Smooth curve, joint premium & \shortstack[r]{$1.042$\\$[1.009,\,1.076]$} & \shortstack[r]{$91.9$\\$[87.0,\,96.9]$} & $0.034$ & $0.191$ \\
\addlinespace[0.25em]
Entropy, lognormal reference & \shortstack[r]{$1.016$\\$[0.972,\,1.060]$} & \shortstack[r]{$90.8$\\$[86.0,\,95.6]$} & $0.032$ & $0.184$ \\
\addlinespace[0.25em]
Entropy, uniform reference & \shortstack[r]{$1.027$\\$[0.985,\,1.069]$} & \shortstack[r]{$93.5$\\$[89.1,\,97.9]$} & $0.033$ & $0.186$ \\
\addlinespace[0.25em]
\midrule
\multicolumn{5}{@{}l}{\textit{Panel B: At least ten common dates per month}} \\
\addlinespace[0.3em]
Baseline put spread & \shortstack[r]{$1.038$\\$[0.997,\,1.080]$} & \shortstack[r]{$90.5$\\$[85.2,\,95.7]$} & $0.041$ & $0.196$ \\
\addlinespace[0.25em]
Call spread & \shortstack[r]{$1.036$\\$[1.006,\,1.065]$} & \shortstack[r]{$91.4$\\$[86.1,\,96.6]$} & $0.039$ & $0.195$ \\
\addlinespace[0.25em]
Smooth curve, joint premium & \shortstack[r]{$1.038$\\$[1.004,\,1.071]$} & \shortstack[r]{$91.9$\\$[86.9,\,96.8]$} & $0.031$ & $0.194$ \\
\addlinespace[0.25em]
Entropy, lognormal reference & \shortstack[r]{$1.011$\\$[0.966,\,1.056]$} & \shortstack[r]{$90.4$\\$[85.4,\,95.4]$} & $0.030$ & $0.187$ \\
\addlinespace[0.25em]
Entropy, uniform reference & \shortstack[r]{$1.022$\\$[0.979,\,1.065]$} & \shortstack[r]{$93.4$\\$[88.9,\,97.8]$} & $0.030$ & $0.189$ \\
\addlinespace[0.25em]
\bottomrule
\end{tabular*}
\par\vspace{0.5em}
\begin{minipage}{\textwidth}\footnotesize\raggedright
\textit{Notes:} Panel A uses 7,078 daily dates, 408 months, 398 adjacent
monthly changes, and 36 recession months. Panel B uses 6,895 daily dates,
361 months, 355 adjacent changes, and 32 recession months.
The dates are identical across all five methods within each panel.
Brackets contain 95 percent HAC confidence intervals, using a Bartlett
kernel, six calendar-month lags, and the correction $n/(n-2)$.
Recession shares are $100\mathcal G(m)/\mathcal G(A)$; their intervals
use the delta method and the joint covariance of the log gaps.
Each smooth or entropy method uses its own jointly fitted premium.
Standard deviations refer to the same adjacent changes used to estimate
elasticities. Inference is conditional on the recovered component series.
\end{minipage}
\end{table}

The severity elasticities range from 1.016 to 1.055, and severity
accounts for 90.8--93.5 percent of the recession log premium.
The smooth estimate is 1.042, with a 95 percent HAC interval of
$[1.009,1.076]$; its recession severity share is 91.9 percent,
with an interval of $[87.0,96.9]$ percent.
Probability's monthly log-change standard deviation ranges from
0.032 to 0.056 across the five methods, compared with 0.184--0.197
for severity.
Thus, the result concerns both comovement with the premium and
comparatively limited probability variation.
Several probability elasticities are slightly negative, so the
findings support near-unit severity elasticity without establishing
that probability is constant or that distributional shape is fixed.

Panel~B of the table requires at least ten common dates in every
retained month, leaving
6,895 dates, 361 months, and 355 adjacent changes, including 32
recession months.
Severity elasticities remain between 1.011 and 1.038, and recession
severity shares remain between 90.4 and 93.4 percent.
The wider-bracket, quotation-quality, smoothing-window, curvature-penalty,
and entropy-support comparisons also retain severity dominance.
Each of these specifications is evaluated against a baseline rebuilt
on its own common daily sample; those samples need not coincide with
the five-method sample in Table~\ref{Tab:DPPRecoveryCommon2}.

Tables~\ref{Tab:DPPRecoverySpreadChecks} and
\ref{Tab:DPPRecoveryDensityChecks} report the additional specifications
and the baseline recalculated on each specification's matched sample.
The positive-bid and one-percent-bracket checks can operate primarily
through sample selection, so their agreement with the matched baseline
should not be read as an independent change in the recovery method.

\begin{table}[p!]\centering\footnotesize
\setlength{\tabcolsep}{3pt}
\renewcommand{\arraystretch}{1.15}
\caption{Spread and Quotation-Quality Comparisons}
\label{Tab:DPPRecoverySpreadChecks}
\begin{tabular*}{\textwidth}{@{\extracolsep{\fill}}p{0.35\textwidth}rrrr@{}}
\toprule
Specification & \shortstack[r]{Daily\\dates} & \shortstack[r]{Severity\\elasticity}
& \shortstack[r]{Severity share\\(\%)}
& \shortstack[r]{Matched baseline\\elasticity;\\share (\%)} \\
\midrule
Second-nearest put bracket & $8,753$ & \shortstack[r]{$1.045$\\$[1.013,\,1.078]$} & \shortstack[r]{$93.0$\\$[88.7,\,97.3]$} & $1.033;\ 91.4$ \\
\addlinespace[0.35em]
Third-nearest put bracket & $8,708$ & \shortstack[r]{$1.060$\\$[1.025,\,1.094]$} & \shortstack[r]{$93.3$\\$[88.6,\,98.1]$} & $1.034;\ 91.5$ \\
\addlinespace[0.35em]
Positive put bids & $8,788$ & \shortstack[r]{$1.033$\\$[0.994,\,1.071]$} & \shortstack[r]{$91.5$\\$[86.7,\,96.3]$} & $1.033;\ 91.5$ \\
\addlinespace[0.35em]
Relative spread $\leq25\%$ & $8,747$ & \shortstack[r]{$1.030$\\$[0.992,\,1.069]$} & \shortstack[r]{$91.5$\\$[86.9,\,96.2]$} & $1.030;\ 91.5$ \\
\addlinespace[0.35em]
Put volume $\geq5$ & $8,418$ & \shortstack[r]{$1.046$\\$[1.009,\,1.082]$} & \shortstack[r]{$92.7$\\$[88.0,\,97.4]$} & $1.035;\ 91.6$ \\
\addlinespace[0.35em]
Bracket width $\leq1\%$ of forward & $6,347$ & \shortstack[r]{$1.023$\\$[0.983,\,1.064]$} & \shortstack[r]{$89.3$\\$[83.8,\,94.8]$} & $1.023;\ 89.3$ \\
\addlinespace[0.35em]
Nearest call bracket & $8,760$ & \shortstack[r]{$1.035$\\$[1.006,\,1.064]$} & \shortstack[r]{$92.5$\\$[88.0,\,97.0]$} & $1.031;\ 91.4$ \\
\addlinespace[0.35em]
Calls on the baseline put bracket & $8,730$ & \shortstack[r]{$1.035$\\$[1.006,\,1.064]$} & \shortstack[r]{$92.6$\\$[88.1,\,97.1]$} & $1.031;\ 91.5$ \\
\addlinespace[0.35em]
\bottomrule
\end{tabular*}
\par\vspace{0.5em}
\begin{minipage}{\textwidth}\footnotesize\raggedright
\textit{Notes:} Each row uses the intersection of that specification's
valid daily dates with the baseline, before monthly aggregation.
The final column gives the baseline severity elasticity and recession
severity share recalculated on those same dates.
Samples can differ across rows and from Table~\ref{Tab:DPPRecoveryCommon2}.
Brackets contain conditional 95 percent HAC intervals with six
calendar-month lags; share intervals use the joint delta method.
The spread methods retain the original premium. Joint smooth and entropy
fits estimate their own premium; fixed-premium fits impose the original
value. No missing calendar month is bridged when forming log changes.
\end{minipage}
\end{table}

\begin{table}[p!]\centering\footnotesize
\setlength{\tabcolsep}{3pt}
\renewcommand{\arraystretch}{1.15}
\caption{Smooth-Curve and Entropy Specifications}
\label{Tab:DPPRecoveryDensityChecks}
\begin{tabular*}{\textwidth}{@{\extracolsep{\fill}}p{0.35\textwidth}rrrr@{}}
\toprule
Specification & \shortstack[r]{Daily\\dates} & \shortstack[r]{Severity\\elasticity}
& \shortstack[r]{Severity share\\(\%)}
& \shortstack[r]{Matched baseline\\elasticity;\\share (\%)} \\
\midrule
Smooth, joint premium & $8,779$ & \shortstack[r]{$1.039$\\$[1.008,\,1.070]$} & \shortstack[r]{$92.7$\\$[88.0,\,97.3]$} & $1.033;\ 91.4$ \\
\addlinespace[0.35em]
Smooth, fixed premium & $8,779$ & \shortstack[r]{$1.041$\\$[1.010,\,1.072]$} & \shortstack[r]{$92.7$\\$[87.9,\,97.5]$} & $1.033;\ 91.5$ \\
\addlinespace[0.35em]
Smooth, 11 quotes, joint premium & $8,780$ & \shortstack[r]{$1.037$\\$[1.006,\,1.069]$} & \shortstack[r]{$92.6$\\$[88.0,\,97.2]$} & $1.033;\ 91.5$ \\
\addlinespace[0.35em]
Smooth, 41 quotes, joint premium & $8,779$ & \shortstack[r]{$1.038$\\$[1.008,\,1.069]$} & \shortstack[r]{$92.9$\\$[88.4,\,97.5]$} & $1.033;\ 91.4$ \\
\addlinespace[0.35em]
Smooth, penalty $10^{-6}$, joint & $8,780$ & \shortstack[r]{$1.035$\\$[1.004,\,1.066]$} & \shortstack[r]{$92.2$\\$[87.5,\,96.9]$} & $1.033;\ 91.4$ \\
\addlinespace[0.35em]
Smooth, penalty $10^{-4}$, joint & $8,781$ & \shortstack[r]{$1.046$\\$[1.013,\,1.078]$} & \shortstack[r]{$93.9$\\$[89.6,\,98.1]$} & $1.033;\ 91.4$ \\
\addlinespace[0.35em]
Entropy, lognormal, joint premium & $7,094$ & \shortstack[r]{$1.017$\\$[0.976,\,1.058]$} & \shortstack[r]{$90.8$\\$[86.0,\,95.6]$} & $1.043;\ 91.1$ \\
\addlinespace[0.35em]
Entropy, uniform, joint premium & $7,096$ & \shortstack[r]{$1.030$\\$[0.989,\,1.070]$} & \shortstack[r]{$93.5$\\$[89.1,\,97.9]$} & $1.043;\ 91.2$ \\
\addlinespace[0.35em]
Entropy, lognormal, fixed premium & $7,026$ & \shortstack[r]{$1.024$\\$[0.984,\,1.064]$} & \shortstack[r]{$91.3$\\$[86.4,\,96.3]$} & $1.046;\ 91.1$ \\
\addlinespace[0.35em]
Entropy, uniform, fixed premium & $7,033$ & \shortstack[r]{$1.036$\\$[0.994,\,1.077]$} & \shortstack[r]{$93.6$\\$[89.1,\,98.1]$} & $1.045;\ 91.1$ \\
\addlinespace[0.35em]
Entropy, wide support, fixed premium & $7,038$ & \shortstack[r]{$1.026$\\$[0.983,\,1.068]$} & \shortstack[r]{$91.4$\\$[86.5,\,96.2]$} & $1.049;\ 91.1$ \\
\addlinespace[0.35em]
\bottomrule
\end{tabular*}
\par\vspace{0.5em}
\begin{minipage}{\textwidth}\footnotesize\raggedright
\textit{Notes:} Each row uses the intersection of that specification's
valid daily dates with the baseline, before monthly aggregation.
The final column gives the baseline severity elasticity and recession
severity share recalculated on those same dates.
Samples can differ across rows and from Table~\ref{Tab:DPPRecoveryCommon2}.
Brackets contain conditional 95 percent HAC intervals with six
calendar-month lags; share intervals use the joint delta method.
The spread methods retain the original premium. Joint smooth and entropy
fits estimate their own premium; fixed-premium fits impose the original
value. No missing calendar month is bridged when forming log changes.
\end{minipage}
\end{table}

\paragraph{Coverage and pricing diagnostics}

Table~\ref{Tab:DPPRecoveryCoverageFit} reports coverage and fitted-premium
differences.
These diagnostics qualify the method comparison.
The jointly fitted smooth estimator recovers components on 8,812
of the 8,825 premium dates, including 8,779 with baseline components.
On the latter sample its premium differs from the baseline by
$-0.28$ percent on average, and the 95th percentile of the absolute
relative difference is 1.75 percent.
Thirty-seven dates have absolute relative differences above five
percent, with a maximum of 40.1 percent.
As an ex-post influence diagnostic, excluding these dates and
rebuilding both series on the remaining dates leaves the smooth
severity elasticity at 1.040 and its recession severity share at
92.7 percent, compared with 1.039 and 92.7 percent before exclusion.
This diagnostic does not replace the main estimates or establish
that all fitted prices lie within their quotation bands.

The median smooth-fit root mean squared pricing residual is 0.116
quotation half-spreads, and its 95th percentile is 1.439 half-spreads,
calculated across valid maturity-level fits.
The largest recorded smooth optimality-gap bound is
$8.3\times10^{-9}$, below the acceptance threshold, and the largest
equality residual is $2.9\times10^{-15}$.
These numerical certificates concern optimization and the moment
restrictions; the smooth objective does not impose hard quotation bands.

The jointly fitted entropy methods recover 7,099 and 7,101 dates
under the lognormal and uniform references, respectively.
Coverage is particularly sparse in the early sample: the lognormal
reference retains only five dates in 1993, one in 1994, and four
in 1995.
Most failures reflect infeasible quotation bands within the specified
density family.
Common-sample comparisons address differences in observation dates;
they do not recover information on dates where an entropy fit is
unavailable.
On their respective baseline-matched samples, the fitted premiums
differ from the baseline by 0.98 and 1.26 percent on average;
the corresponding 95th percentiles of absolute relative differences
are 3.07 and 3.76 percent.

\begin{table}[t!]
\centering\small
\setlength{\tabcolsep}{3pt}
\caption{Coverage and Fitted-Premium Differences}
\label{Tab:DPPRecoveryCoverageFit}
\begin{tabular*}{\textwidth}{@{\extracolsep{\fill}}p{0.33\textwidth}rrrr@{}}
\toprule
Method & \shortstack[r]{Valid\\dates}
& \shortstack[r]{Baseline-matched\\dates}
& \shortstack[r]{Mean relative\\difference\\(\%)}
& \shortstack[r]{Absolute relative\\difference (\%)\\95th percentile} \\
\midrule
Baseline put spread & $8,791$ & $8,791$ & --- & --- \\
Call spread & $8,787$ & $8,760$ & --- & --- \\
Smooth curve, joint premium & $8,812$ & $8,779$ & $-0.28$ & $1.75$ \\
Entropy, lognormal reference & $7,099$ & $7,094$ & $0.98$ & $3.07$ \\
Entropy, uniform reference & $7,101$ & $7,096$ & $1.26$ & $3.76$ \\
\bottomrule
\end{tabular*}
\par\vspace{0.5em}
\begin{minipage}{\textwidth}\footnotesize\raggedright
\textit{Notes:} Each method is attempted on the 8,825 original premium
dates. A daily estimate requires recovery at every selected maturity
needed for that date. Premium differences are
$100(\chi_t^{(j)}/\chi_t^{\mathrm{baseline}}-1)$, calculated on
each method's baseline-matched daily sample.
The last column is the 95th percentile of their absolute values.
These are relative price differences, not percentage-point changes in
the DPP. Dashes indicate that the method retains the original premium.
Counts in this table are method-specific; the five-way intersection
contains 7,078 dates.
\end{minipage}
\end{table}

\paragraph{Outer bounds from option quotations}

For each selected maturity, I bound $q=\mathbb Q(X<1)$ over a linear
relaxation of positive-return distributions whose mean is one,
whose at-the-money put price equals the original constructed premium,
and whose put prices lie inside all retained bid--ask quotations.
The relaxation retains both probability mass and the first moment
above the largest quoted strike, so it imposes no finite upper
return support.
It also permits limiting endpoint configurations; consequently,
the intervals are conservative outer bounds rather than a claim
that every endpoint is attained by a strictly positive distribution.
An atom at the forward is treated separately from the strict shortfall
event.
The relaxation retains every eligible put quotation at the selected
maturity, rather than only the nearby fitting window.
It partitions at quoted strikes and at zero and one, retaining both
tail mass and its first moment above the largest knot.
A separate left-limit copy of one distinguishes mass just below
the forward from an atom exactly at it.
Infeasible quotation sets are reported without relaxing their bands.

The secant inequalities in \eqref{eq:dpp_probability_secants} apply
to the true price curve.
If the AtM price is linearly interpolated between the same two
neighboring strikes, both secant slopes equal the baseline put-spread
slope.
That mechanical equality provides no independent robustness check.
The quotation-bound calculation instead lets each quoted put price
vary within its observed bid--ask interval while keeping the original
AtM premium fixed.

The original positive carry weights map maturity-level intervals
into daily and monthly probability intervals $[L_M,U_M]$.
Given the premium series, both empirical attributions are linear
functions of $\log q_M$.
For monthly comovement, writing $x_M=\Delta\log A_M$ and
$h_M=(x_M-\bar x)/\sum_{J\in\mathcal C}(x_J-\bar x)^2$ on the set
$\mathcal C$ of valid adjacent changes gives
\begin{equation*}
    \beta_q=\sum_{M\in\mathcal C}h_M
    (\log q_M-\log q_{M-1})
    =\sum_M c_M\log q_M.
\end{equation*}
Each transition adds $h_M$ to $c_M$ and subtracts it from $c_{M-1}$.
The minimum is obtained by choosing $\log L_M$ when $c_M\geq0$
and $\log U_M$ when $c_M<0$; the maximum reverses these choices.
The severity interval follows from $\beta_m=1-\beta_q$.
For the recession comparison, the same calculation uses weights
$R_M/n_R-(1-R_M)/n_E$ and
$\mathcal G(m)/\mathcal G(A)=1-\mathcal G(q)/\mathcal G(A)$ when
$\mathcal G(A)>0$.
These are exact extrema over the box of monthly intervals and
conservative attribution envelopes for the underlying distributions.
They describe quotation uncertainty conditional on the original
forward and premium, not sampling confidence intervals.

The quotation envelopes are not sufficiently tight to establish
severity dominance on their own.
On the 5,978 dates common to the distributional outer-bound estimator
and the baseline, the severity-elasticity envelope is
$[-0.064,2.165]$, and the recession severity-share envelope is
$[20.2,163.3]$ percent.
The former uses 374 adjacent monthly changes and the latter 387 months.
The outer envelopes include values below equal probability and
severity contributions; this does not establish that such values
are attainable by admissible underlying distributions.
The empirical conclusion is therefore stability across the reported
component-recovery methods, rather than identification from quotation
bounds alone.

Table~\ref{Tab:DPPRecoveryQuotationBounds} also reports the
convexity-only relaxation and the two finite-density-family calculations.
The latter remain conditional on their specified support and interval
density representation.

\begin{table}[t!]
\centering\small
\setlength{\tabcolsep}{3pt}
\caption{Quotation-Based Attribution Envelopes}
\label{Tab:DPPRecoveryQuotationBounds}
\begin{tabular*}{\textwidth}{@{\extracolsep{\fill}}lrrrr@{}}
\toprule
Restrictions & Daily dates & \shortstack[r]{Months /\\changes}
& \shortstack[r]{Severity-elasticity\\envelope}
& \shortstack[r]{Recession severity-share\\envelope (\%)} \\
\midrule
Convexity inequalities & $8,553$ & $439 / 438$ & $[-0.085,\,2.181]$ & $[26.9,\,160.9]$ \\
Finite density family, main support & $7,041$ & $408 / 398$ & $[-0.060,\,2.184]$ & $[22.0,\,163.8]$ \\
Finite density family, wider support & $7,056$ & $409 / 399$ & $[-0.060,\,2.185]$ & $[22.2,\,163.6]$ \\
Distributional outer relaxation & $5,978$ & $387 / 374$ & $[-0.064,\,2.165]$ & $[20.2,\,163.3]$ \\
\bottomrule
\end{tabular*}
\par\vspace{0.5em}
\begin{minipage}{\textwidth}\footnotesize\raggedright
\textit{Notes:} Each row uses the dates common to the stated bound
calculation and the baseline, holding the original forward and premium
fixed. Months apply to recession comparisons; adjacent changes apply
to elasticities. Samples differ across rows, so widths should not be
read as comparisons on identical observations.
These are measurement envelopes, not sampling confidence intervals.
Finite-family results depend on the stated support and interval-density
representation; the distributional outer relaxation has no finite
upper return bound. Attribution extrema over monthly interval boxes
need not be jointly attainable by admissible price curves.
\end{minipage}
\end{table}

\paragraph{Recovery Error and Regression Attribution}
\label{App:DPPRecoveryError}

Conditional on the constructed premium, write
$\log\widehat q_M=\log q_M+e_M$ and $x_M=\Delta\log A_M$. Then
\begin{equation}
 \widehat\beta_q-\beta_q
 =\frac{\operatorname{Cov}(\Delta e_M,x_M)}{\operatorname{Var}(x_M)},
 \qquad
 \mathcal G(\widehat q)-\mathcal G(q)
 =\mathbb E[e_M\mid R]-\mathbb E[e_M\mid E].
 \label{eq:dpp_recovery_error}
\end{equation}
The probability loading shifts with
$\operatorname{Cov}(\Delta e_M,\Delta\log A_M)$, while the recession gap
shifts with the difference in mean log-recovery error across phases.
Orthogonality of these log errors preserves the population estimands;
unbiased errors in probability levels alone need not do so after taking logs.
If the premium is also measured with error, both regressor error and
its covariance with probability error enter the loading.

The interval-box extrema choose each $\log q_M$ according to the sign
of the attribution weight. They are worst-case choices, not draws from
a distribution of quotation noise. Agreement across recovery methods
is informative about the implemented estimators but cannot establish
orthogonality when the same quotations enter all methods. The HAC and
block-bootstrap intervals condition on recovered components and should
not be read as incorporating a separate quotation-error model.

\paragraph{The 34 excluded component dates}

To assess whether the missing component dates could explain the
baseline result, I retain the estimated probabilities on the 8,791
usable dates and allow probability on each of the other 34 dates
to range over $[a_t,1]$.
This is a conservative closed admissibility range under positive
returns and the forward-mean restriction; the limiting endpoints
need not be attained.
I then rebuild the monthly premium and probability intervals using
all 8,825 premium dates and apply the preceding attribution formulas.
The resulting severity elasticity lies between 1.003 and 1.058,
and severity's share of the recession log premium lies between
90.5 and 92.6 percent.
The severity-minus-probability recession log gap remains between
0.466 and 0.490.
Thus, the excluded dates cannot overturn severity dominance when
the retained component estimates are held fixed.
These ranges isolate exclusion sensitivity; they do not incorporate
measurement error on the retained dates.

\paragraph{Joint block-bootstrap sensitivity}
For the original 439-month baseline sample, I assess sampling uncertainty
by jointly resampling monthly log components,
their original adjacent log changes, and the recession indicator in circular
blocks of 6, 12, and 24 months.
Each block length uses 4,999 draws, and the table reports marginal
95 percent percentile intervals.
Changes are computed before resampling, so concatenating blocks creates no
artificial monthly changes; the probability and severity adding-up restrictions
continue to hold in every usable draw.

\begin{table}[htbp]
\centering
\small
\caption{Joint Block-Bootstrap Sensitivity}
\label{Tab:DPPJointBlockBootstrap}
\begin{tabular}{rcccc}
\toprule
Block & Severity & Recession severity & Probability & Relative probability \\
months & elasticity & share (percent) & log gap & volatility \\
\midrule
6 & [0.998, 1.072] & [85.2, 96.8] & [0.018, 0.077] & [0.209, 0.277] \\
12 & [0.996, 1.074] & [84.8, 97.9] & [0.010, 0.077] & [0.203, 0.284] \\
24 & [0.996, 1.076] & [84.5, 98.5] & [0.009, 0.077] & [0.197, 0.297] \\
\bottomrule
\end{tabular}
\begin{minipage}{0.98\textwidth}
\footnotesize\raggedright
\emph{Notes:} Entries are marginal 95 percent percentile intervals.
Relative probability volatility is
$\operatorname{sd}(\Delta\log q_M)/\operatorname{sd}(\Delta\log m_M)$.
Identical monthly index draws are used for the monthly comovement and recession
comparisons. The first original month's unavailable change stays missing when
that month is drawn. Draws without both recession and expansion observations
are excluded only from the recession calculations; counts are reported in the
replication diagnostics. Negative recession-premium draws are retained and
shares are not clipped.
\end{minipage}
\end{table}

Across the three block lengths, the smallest lower endpoint for the recession
severity share is 84.5 percent, and the smallest lower endpoint for the
severity-minus-probability recession log gap is 0.255.
The largest upper endpoint for the volatility of probability relative to
severity is 0.297.
These comparisons describe the marginal intervals and are not a simultaneous
test across block lengths.

These intervals measure sampling sensitivity conditional on the recovered
component series.
They do not incorporate uncertainty in the recovery of probabilities from
option quotations.
Only four recession episodes are available; resampling blocks reuses and can
split these episodes, and does not produce additional independent cycles.
The bootstrap therefore complements the comparisons across time windows and
individual recessions, while its finite-sample coverage and sensitivity to
structural breaks remain limitations.

\subsubsection{Common Scale and Conditional Downside Shape}
\label{App:DPPScaleRestrictions}

This subsection reports the recovered-probability comparisons and
the implementation of the joint quotation restrictions.
All calculations use the original 30-day construction,
carrying-cost inputs, component at-the-money premiums, and
selected quotation bands.
The main sample contains 6,966 daily dates in 405 covered calendar
months between January 1990 and July 2026.
These are not 405 consecutive calendar months.
The archived outputs retain daily coverage, corrected scalar
dual bounds, numerical residuals, and fitted distributions.

\paragraph{Scale predictions away from the forward.}
Let $G_t(k)=\mathbb Q_t(X_t<k)$. Under common scale, relative to date $b$,
\begin{equation}
 \lambda_{t,b}=\frac{a_t}{a_b}=\frac{s_t}{s_b}.
 \label{eq:dpp_scale_ratio}
\end{equation}
\begin{equation}
 G_{t,b}^{\mathrm{scale}}(k)
 =G_b\!\left(1+\frac{k-1}{\lambda_{t,b}}\right).
 \label{eq:dpp_scale_probability_prediction}
\end{equation}
The composition is
$q_t=\mathbb Q(Z<0)$, $a_t=s_t\mathbb E[(-Z)^+]$, and
$m_t=s_t\mathbb E[-Z\mid Z<0]$.
\begin{equation}
 G_{t,b}^{\mathrm{scale}}(0.95)=G_b(0.975)
 \quad\text{when }a_t=2a_b.
 \label{eq:dpp_scale_doubling_example}
\end{equation}
Thus downside probabilities away from the forward can rise even with
constant probability at the forward.
The scale transformation must preserve positivity:
$\mathbb Q_b[1+\lambda_{t,b}(X_b-1)>0]=1$.
Replacing negative transformed values by zero changes the mean and
premium and does not preserve this benchmark.

\paragraph{Recovered-probability comparisons.}

For a fixed reference date $b$, define the discrepancy between
the observed probability and the common-scale prediction as
\begin{equation}
	\varepsilon_{t,b}(k)
	=
	G_t(k)-G_{t,b}^{\mathrm{scale}}(k)
	\label{eq:dpp_scale_probability_discrepancy}
\end{equation}
This is also the difference between the observed probability
change from $b$ and the change predicted by the scale benchmark.
Predictions are formed at the daily frequency before aggregation,
using identical dates and weights for observed and predicted
values.
An average of daily distributions with different scales need not
itself be a rescaled version of the reference distribution.

The smooth-curve and two entropy reconstructions have 7,018
common valid-fit dates from the 8,825 premium dates.
Requiring coverage of the fixed probability cutoffs leaves
6,966 dates, including 6,351 expansion dates and 615 recession
dates.
There are 394 usable adjacent-month comparisons from 438
potential comparisons and 129,733 eligible daily pairs from
176,432 potential pairs.

For each daily pair, label the date with the larger at-the-money
premium $h$ and the other date $l$.
With $r=a_l/a_h$, common scale implies
$G_h(k)=G_l(1+r(k-1))$.
Daily pair comparisons are combined using product carry weights
$D_tD_b$, normalized within each adjacent-month comparison.
Summary statistics weight the monthly comparisons equally.
Methods and cutoffs use the same eligible pairs.
Calendar gaps are not bridged, and overlapping pairs are not
independent observations.

\begin{table}[!htbp]
	\centering
	\small
	\caption{Recovered-probability discrepancies from common scale}
	\label{Tab:DPPScaleDescriptive}
	\begin{tabular}{llrrr}
		\toprule
		Method & $K/F$ & RMS observed & RMS predicted & RMS discrepancy\\
		\midrule
		Smooth curve & 0.975 & 2.738 & 2.252 & 1.081\\
		& 0.950 & 2.802 & 2.310 & 0.775\\
		& 1.000 & 1.506 & 0.000 & 1.506\\
		Entropy: lognormal & 0.975 & 2.755 & 2.254 & 1.012\\
		& 0.950 & 3.158 & 2.415 & 0.963\\
		& 1.000 & 1.380 & 0.000 & 1.380\\
		Entropy: uniform & 0.975 & 2.663 & 2.242 & 0.908\\
		& 0.950 & 2.669 & 2.350 & 0.612\\
		& 1.000 & 1.409 & 0.000 & 1.409\\
		\bottomrule
	\end{tabular}
	\par\smallskip
	\begin{minipage}{\textwidth}\footnotesize
		\emph{Notes:} Entries are percentage points.
		RMS denotes the root mean square across 394 adjacent-month
		comparisons.
		The first two columns of statistics refer to observed and
		scale-predicted probability changes; the final column
		refers to their difference.
		These are descriptive recovered-distribution comparisons,
		not quotation-based exclusions or sampling confidence
		intervals.
	\end{minipage}
\end{table}

The separate quotation-bound calculation retains 6,999 usable
dates in its expanded sample, adding 33 dates and one
adjacent-month comparison to the main sample.
All 2,367 monthly objectives in that calculation are numerically
resolved.
No signed-discrepancy outer interval excludes zero at
$k=0.95$, $0.975$, or $1$, in either sample.
On the matched sample, median interval widths are 28.18, 32.57,
and 35.13 percentage points, respectively.
These intervals are too wide to establish a departure using the
averaged discrepancy alone.
Separate extrema may be attained by different distributions,
so inclusion of zero does not establish that one distribution
satisfies all restrictions jointly.

\paragraph{Fixed anchors and maturity aggregation.}

The selected put quotations include the 21 nearest at-the-money
strikes and strict brackets around $K/F=0.975$ and $0.95$,
within the maintained strike window.
Both sides of the forward enter the joint calculation.

At each component maturity, the premium anchor is fixed at the
value obtained by strike interpolation of bid--ask midpoints.
The original forward and carrying-cost inputs are also fixed.
The component anchors are not optimized within the quotation
bands.
For the original maturity-interpolation weight $w$ and carrying
factors $D_s,D_l$, the daily mixture weights are
\[
\alpha_s
=
\frac{(1-w)D_s}{(1-w)D_s+wD_l}
\qquad
\alpha_l
=
\frac{wD_l}{(1-w)D_s+wD_l}
\]
The daily adjusted premium is
$a_t=\alpha_sa_t^s+\alpha_la_t^l$.
Every component law retains unit mass, unit forward mean, its
original premium anchor, and its selected quotation bands.
Common scale applies to the daily mixture; it does not require
identical laws at the two component maturities.

All exclusion statements condition on these fixed anchors,
forwards, carrying costs, and quotation restrictions.
The calculations do not assess sensitivity to allowing the
anchors themselves to vary within their quotation bands.

\paragraph{Joint full-distribution scale calculation.}

Under common scale,
$Z_t=(X_t-1)/a_t$ has a common distribution across dates.
Writing
$C(x)=\mathbb E^{\mathbb Q}[(x-Z)^+]$ gives
\[
\frac{\mathcal P_t(k)}{a_t}
=
C\left(\frac{k-1}{a_t}\right)
\qquad
C(0)=1
\]
For a continuous distribution,
$C'(0)=q$.
More generally, the strict shortfall probability is the left
derivative $C'_-(0)$; the right derivative also includes any
atom at the forward.
The common function is chosen once for the comparison window,
rather than separately for each date.

For implementation, write $A=\max_t a_t$ within a window and
transform returns to
$R=1+(A/a_t)(X_t-1)$.
Quotation strikes transform by the same affine map.
The linear program represents each component law using
nonnegative masses at the transformed payoff knots and an
upper tail described by its mass and first moment.
The relevant put payoffs are affine between knots, so this
representation preserves their expectations.
Equality of the weighted component masses and tail quantities
across dates imposes the common daily law.

For clarity, let $h$ index the positive-weight component maturities
at date $t$, with weights $\alpha_{th}$ and anchors $a_{th}$.
All quotations below are undiscounted put prices divided by the
component forward. Set
$\gamma_t=A/a_t$ and $\kappa_{thi}=1+\gamma_t(k_{thi}-1)$ for
quoted normalized strike $k_{thi}$.
The common knot set contains zero, one, and every positive
$\kappa_{thi}$; write its largest knot as $\bar r$.
For component $(t,h)$, let $\pi_{thj}$ be mass at knot $r_j$,
$v_{th}$ upper-tail mass, and $\mu_{th}$ its first moment.
The component moment restrictions are
\begin{equation}
\begin{aligned}
 \sum_j\pi_{thj}+v_{th}&=1,\qquad
 \sum_j r_j\pi_{thj}+\mu_{th}=1,\\
 \sum_j(1-r_j)^+\pi_{thj}&=\gamma_ta_{th},\qquad
 \mu_{th}\geq\bar r v_{th}.
\end{aligned}
\label{eq:app_scale_component_lp}
\end{equation}
All masses and moments are nonnegative. One common daily law is
imposed by requiring
$\sum_h\alpha_{th}\pi_{thj}=\nu_j$,
$\sum_h\alpha_{th}v_{th}=v_*$, and
$\sum_h\alpha_{th}\mu_{th}=\mu_*$ at every date, with
$\nu_j,v_*,\mu_*$ independent of $t$.
The fitted original-unit quotation is
$\widehat p_{thi}=\gamma_t^{-1}
\sum_j(\kappa_{thi}-r_j)^+\pi_{thj}$.
The objective is to minimize $\delta\geq0$ subject to these
restrictions and
\begin{equation}
 p_{thi}^{\mathrm{bid}}-\delta a_t
 \leq\widehat p_{thi}
 \leq p_{thi}^{\mathrm{ask}}+\delta a_t.
\label{eq:app_scale_quote_lp}
\end{equation}
Thus $\delta$ is measured in fractions of the daily adjusted premium.
The upper tail lies beyond every transformed quotation strike, so it
has no put payoff at the retained strikes; its mass and moment preserve
the remaining moment restrictions without a finite upper-price cap.
For numerical implementation, one component at each date is eliminated
using the mixture equalities and its nonnegativity is retained explicitly.

Allowing zero terminal prices and limiting tail configurations
gives an outer closure of the class with strictly positive terminal equity values.
Excluding this larger class excludes the maintained strict
model.
A fit in the outer class need not establish an attainable
law with strictly positive terminal equity values.

When the original bands cannot be matched, the calculation
permits a common allowance $\delta a_t$ on each side of every
selected band at date $t$, while keeping the premium anchors
exact.
A corrected dual lower bound above $10^{-8}$ certifies
exclusion.
The numerical primal solution supplies a fit and residual
diagnostics; it is not a mathematically repaired upper
certificate.
Every full-scale window in the reported samples is numerically
resolved.
All exclusions in this run arise from the quotation restrictions
rather than a standalone premium-anchor or support contradiction.

\paragraph{The magnitude of the pricing departures.}

Across the 405 main-sample monthly windows, the median corrected
dual lower bound on the extra quotation allowance is
1.238 percent of the daily adjusted at-the-money premium.
Among the 352 excluded months, the median is 1.393 percent.
For adjacent-month windows, the corresponding all-window median
is 2.126 percent.

To provide a spread benchmark, I interpolate each component
maturity's normalized quoted half-spread to the forward using
the same strike bracket as the midpoint anchor.
Component half-spreads are expressed in the same units as their
adjusted premiums and combined using the original
carry-adjusted maturity weights.
The median daily interpolated half-spread is 2.741 percent of
the premium.
The median of the within-month median half-spreads is
2.801 percent across all months and 2.717 percent among excluded
months.
Within the excluded months, the median ratio of the certified
extra allowance to that month's median interpolated half-spread
is 0.560.

The benchmark is an interpolated half-spread at the forward,
not a directly quoted contract spread on every date or the
spread of each binding wing quotation.
The extra allowance is additional to the original bid--ask bands
already admitted by the calculation.
Its being smaller than a typical at-the-money half-spread
therefore does not make the original restrictions compatible,
or establish that every wing discrepancy is economically small.
A certified lower bound is a minimum required allowance, not an allowance
shown sufficient to restore compatibility. The reported bounds therefore
do not supply an upper bound on the economic size of the discrepancy.

\paragraph{Conditional downside shape.}

The constant-$q$ and varying-$q$ models in
\eqref{eq:dpp_conditional_put_curve} leave the upside law
unrestricted across dates.
Their comparison isolates the restriction that the broad
shortfall probability remains constant.
The full-scale model additionally restricts the upside shape.

Under the conditional restriction,
$(1-X_t)\mid X_t<1=m_tU$, where
$\mathbb E[U]=1$ and the distribution of $U$ is common across
dates.
Since $m_t=a_t/q_t$, for $k\leq1$,
\[
\mathcal P_t(k)
=
q_t\mathbb E[(m_tU-(1-k))^+]
=
a_tH\left(\frac{(1-k)q_t}{a_t}\right)
\]
Here $H(z)=\mathbb E[(U-z)^+]$, so $H(0)=1$.
The constant-$q$ model chooses one unknown probability for the
window; the varying-$q$ model permits a separate probability
at each date.

The first computational stage projects the unchanged component
restrictions into marginal bounds on put prices at all selected
downside strikes and at $0.95$, $0.975$, and $0.99$, together
with bounds on shortfall probability.
The component grids are combined within each date, and the
original carry-adjusted maturity weights give daily intervals.
These separate marginal intervals enlarge the feasible set.

An inverse conditional put-curve formulation imposes necessary
common-shape restrictions while solving for the unknown
probabilities.
A positive corrected certificate in this outer problem excludes
every admissible probability path.
Its objective is a numerical feasibility diagnostic, not an
economically measured price deviation.

More explicitly, let $[L_t(k),U_t(k)]$ and
$[\underline q_t,\overline q_t]$ be the daily marginal projections.
For $0<k<1$, intersect $[L_t(k)/a_t,U_t(k)/a_t]$ with $[0,k]$,
using $0\leq\mathcal P_t(k)\leq k a_t$, and write its endpoints as
$\rho_t^-(k)$ and $\rho_t^+(k)$.
Empty intersections are retained as infeasible constraints.
Let $J=H^{-1}$, with $J(0)$ interpreted as the upper support endpoint.
Its grid contains zero, one, and all distinct price-interval endpoints.
Necessary restrictions are
\begin{equation}
\begin{aligned}
 J(1)&=0,\qquad J\text{ is convex},\qquad
 \frac{J(\rho_{j+1})-J(\rho_j)}{\rho_{j+1}-\rho_j}\leq-1,\\
 \max\{a_t,\underline q_t\}&\leq q_t
 \leq\min\{1,\overline q_t\},\qquad J(0)\leq\frac{q_t}{a_t},\\
 J(\rho_t^+(k))&\leq\frac{(1-k)q_t}{a_t},\\
 \frac{(1-k)q_t}{a_t}&\leq J(\rho_t^-(k))
 \quad\text{when }\rho_t^-(k)>0.
\end{aligned}
\label{eq:app_conditional_inverse_lp}
\end{equation}
The last inequality is omitted at a zero lower price bound because
$H\geq0$ already supplies that bound.
The slope restriction follows from $-1\leq H'\leq0$ on its decreasing
part; convexity is preserved by inversion.
Common conditional loss support and nonnegative terminal equity imply
$J(0)\leq q_t/a_t$ at every date.
These constraints are linear in the grid values of $J$ and the unknown
probabilities. The constant-$q$ calculation uses one probability variable
for the window, while the varying-$q$ calculation uses one per date.

The implementation represents $J$ in units of $1/\max_t a_t$, so
its grid variables and probability variables lie in $[0,1]$.
Convexity is imposed without dividing by nearly coincident endpoint
spacings; distinct floating-point endpoints are never merged.
Each nonconstant constraint row is divided by the largest absolute
coefficient or right-hand side; a constant-only row uses a divisor
of at least one. A common nonnegative elastic slack, bounded by one,
relaxes these normalized rows. Minimizing it gives an exclusion
diagnostic for the marginal necessary system, rather than a price distance.

If the necessary problem is feasible, the candidate
probabilities are passed to a second linear program restoring
the complete original component quotation and moment
restrictions.
Only a numerically verified joint fit receives the outer-fit
classification.
A failed candidate remains unresolved because another
probability vector may fit.
The outer closure can include zero-loss endpoint mass, so
candidate probabilities are not identified estimates of strict
shortfall-event probabilities.

For a supplied probability path, this verification is a linear
program. Set $m_t=a_t/q_t$ and
$B=\min_t q_t/a_t$. Let $\nu_j$ be a common unit-mass,
unit-mean law on conditional-loss knots $u_j\in[0,B]$.
For each component, its nonnegative labeled downside masses
$d_{thj}$ satisfy
\begin{equation}
 \sum_h\alpha_{th}d_{thj}=q_t\nu_j,\qquad
 m_t\sum_j u_jd_{thj}=a_{th}.
\label{eq:app_conditional_joint_lp}
\end{equation}
An otherwise unrestricted component upside law supplies the remaining
mass and expected gain $a_{th}$. Its excess-return tail is represented
by mass and first moment, with no finite upper cap.
The knot grids retain every selected quotation payoff kink and the
support endpoints. All original put bands, including strikes above
the forward, are imposed jointly.
The sparse stop-loss-curve implementation is equivalent to these
mass restrictions; accepted solutions are converted back to masses
and checked against the original component payoffs, moments, and mixture.
A failed verification concerns only the supplied probability path
and therefore cannot exclude the class with unknown $q_t$.

\paragraph{Numerical bounds and classification rules.}

The full-scale distance and inverse elastic programs have bounded
variables. Write either in the form
\begin{equation}
 \min_x c^{\mathsf T}x\quad\text{subject to}\quad
 \mathsf Bx\leq b,\quad \mathsf Ex=d,\quad \ell\leq x\leq u.
\label{eq:app_scale_bounded_program}
\end{equation}
For numerical inequality multipliers $\lambda$, replace positive
entries by zero so that $\lambda\leq0$, and let $y$ denote the
equality multipliers. Define
$r=c-\mathsf B^{\mathsf T}\lambda-\mathsf E^{\mathsf T}y$.
Weak duality gives the bounded-variable lower bound
\begin{equation}
 L=b^{\mathsf T}\lambda+d^{\mathsf T}y+
 \sum_j\begin{cases}
 r_j\ell_j,&r_j\geq0,\\
 r_ju_j,&r_j<0.
 \end{cases}
\label{eq:app_scale_dual_correction}
\end{equation}
The bound includes the adverse contribution of any residual reduced
cost; it does not assume that numerical multipliers satisfy exact
dual stationarity. In the inverse problem, the equality term is absent
because $J(1)=0$ is imposed as a variable bound.
The implementation subtracts $10^{-10}(1+S)$, where $S$ is the sum
of the absolute terms in $L$, rounds downward to the next representable
floating-point value, and clips the lower bound at zero.
Exclusion requires this corrected numerical lower bound to exceed
$10^{-8}$.

The solver's primal and dual feasibility tolerances are $10^{-9}$.
Accepted solutions have maximum inequality, equality, and variable-bound
residual at most $5\times10^{-8}$ in the implemented units.
The full-scale calculation also checks the original component mass,
mean, and anchor restrictions, and records remaining quotation slack.
A joint conditional fit additionally requires maximum quotation slack
at most $10^{-8}$ of the daily adjusted premium, with its mass,
moment, and common-loss restrictions independently rechecked.
The sparse interior-point verification uses an optimality tolerance
of $10^{-10}$.
Nonfinite outputs, conflicting primal and dual diagnostics, failed
solves, and time limits remain unresolved.
A numerical witness is not an exact-arithmetic proof of attainability;
exclusion certificates and fitted witnesses retain their distinct roles.

\paragraph{Nesting and propagation of existing fits.}

The reported classifications respect
\[
\begin{aligned}
\{\text{full common scale}\}
&\subseteq\{\text{common conditional shape, constant }q\}\\
&\subseteq\{\text{common conditional shape, varying }q\}.
\end{aligned}
\]
This nesting also holds for the outer formulations used here.

To see the first inclusion, let a full-scale fit supply the
common transformed law of
$R=1+(A/a_t)(X_t-1)$.
Its premium satisfies
$\mathbb E[(1-R)^+]=A$.
Let $q=\mathbb Q(R<1)$ and, conditional on $R<1$, define
$U=q(1-R)/A$.
Then $\mathbb E[U]=1$ and
\[
1-X_t
=
\frac{a_t}{A}(1-R)
=
\frac{a_t}{q}U
\qquad
\text{on }X_t<1
\]
The same conditional shape and the same $q$ therefore apply at
every date.
The original component laws continue to satisfy their
quotation, anchor, mass, and mean restrictions.
Allowing the upside law to vary does not invalidate this fit.
The second inclusion follows by allowing an already constant
probability path in the varying-$q$ class.

Verified fits from each smaller class are carried into the nested larger
classes, retaining the accepted numerical residual standards. This changes
no exclusion, quotation, anchor, or sample and does not turn numerical
witnesses into exact-arithmetic certificates.

\paragraph{Coverage and compatibility across windows.}

Table~\ref{Tab:DPPScaleWindowRobustness} reports the final window and
sample comparisons, including inherited fits from nested classes.
The expanded sample is selected using outer quotation
feasibility and strike coverage.
This selection is not itself a certificate of an attainable
law with strictly positive terminal equity values.

\begin{table}[!htbp]
	\centering
	\small
	\setlength{\tabcolsep}{4pt}
	\caption{Coverage and joint compatibility across windows}
	\label{Tab:DPPScaleWindowRobustness}
	\begin{tabular}{llrrrr}
		\toprule
		Window/sample & Restriction
		& $N$ & Excluded & Outer fit & Unresolved\\
		\midrule
		\multicolumn{6}{l}{
			\textit{Main sample: 6,966 daily observations}}\\
		Monthly & Full scale & 405 & 352 & 53 & 0\\
		& Conditional, constant $q$ & 405 & 345 & 55 & 5\\
		& Conditional, varying $q$ & 405 & 112 & 138 & 155\\
		Adjacent months & Full scale & 394 & 389 & 5 & 0\\
		& Conditional, constant $q$ & 394 & 388 & 5 & 1\\
		& Conditional, varying $q$ & 394 & 174 & 6 & 214\\
		\midrule
		\multicolumn{6}{l}{
			\textit{Main sample: at least ten dates per constituent month}}\\
		Monthly & Full scale & 356 & 319 & 37 & 0\\
		& Conditional, constant $q$ & 356 & 312 & 39 & 5\\
		& Conditional, varying $q$ & 356 & 90 & 116 & 150\\
		Adjacent months & Full scale & 348 & 343 & 5 & 0\\
		& Conditional, constant $q$ & 348 & 342 & 5 & 1\\
		& Conditional, varying $q$ & 348 & 133 & 5 & 210\\
		\midrule
		\multicolumn{6}{l}{
			\textit{Expanded sample: 6,999 daily observations}}\\
		Monthly & Full scale & 406 & 354 & 52 & 0\\
		& Conditional, constant $q$ & 406 & 347 & 54 & 5\\
		& Conditional, varying $q$ & 406 & 116 & 133 & 157\\
		Adjacent months & Full scale & 395 & 390 & 5 & 0\\
		& Conditional, constant $q$ & 395 & 389 & 5 & 1\\
		& Conditional, varying $q$ & 395 & 179 & 6 & 210\\
		\bottomrule
	\end{tabular}
	\par\smallskip
	\begin{minipage}{\textwidth}\footnotesize
		\emph{Notes:} Verified fits are inherited by nested larger
		model classes.
		No exclusion count changes through witness propagation.
		All calculations retain the original component quotations,
		fixed midpoint anchors, forwards, carrying costs, and
		maturity weights.
		An outer fit need not correspond to an attainable
		law with strictly positive terminal equity values.
		Unresolved outcomes remain in the denominators.
		Overlapping windows are separate compatibility exercises;
		the counts are not independent tests or sampling
		rejection frequencies.
	\end{minipage}
\end{table}

Candidate failures and time limits are never counted as exclusions.
The final main-sample classifications are 112 exclusions, 138 outer fits,
and 155 unresolved months.

\paragraph{What the nested comparisons establish.}

Of the 352 monthly full-scale exclusions, 345 remain excluded
when the upside shape is unrestricted, two have a constant-$q$
outer fit, and five remain unresolved.
Thus most exclusions do not require imposing a common upside
shape.
This is not a downside-quotations-only test: the projections and
joint verification retain all original component quotation and
moment restrictions, including quotations above the forward.

Among the 345 constant-$q$ exclusions, 80 have a verified joint
outer fit when $q_t$ varies, 112 remain excluded, and 153 remain
unresolved.
The probability paths in the 80 fits are witnesses of
compatibility in the enlarged class, not estimates of the daily
physical or strict priced probabilities.

The 112 exclusions establish that at least 27.7 percent of all
405 months require a departure from common conditional downside
shape under the maintained measurement restrictions.
For the outer class alone, assigning every unresolved month
to incompatibility would raise the share to
$(112+155)/405=65.9$ percent.
This is not an upper bound for the class with strictly positive terminal equity values, because outer fits need not be attainable there.
Neither endpoint is a sampling confidence bound.

\paragraph{Direction of departures and unresolved windows.}

The saved outputs contain corrected scalar dual bounds and
primal quotation fits, but not the full vectors of dual
multipliers.
Quotations attaining the maximum allowance in a saved minimax
fit need not form a unique conflicting set, and an active
quotation can have zero dual weight.
They therefore do not identify a unique directional tail
response.

Across the 352 excluded main-sample months, the saved fits
contain binding downside quotations both above asks and below
bids.
They also contain binding quotations on dates both above and
below the within-window median premium.
These diagnostics do not support a uniform direction of
departure in the left tail.

The unresolved months are not confined to the earliest or
sparsest observations.
Their median coverage is 19 dates and 829 selected component
quotations, compared with 20 dates and 878 quotations in the
fitted months.
Of the 155 unresolved months, 150 contain at least ten retained
dates.
The years 2000--2004 account for 45 unresolved months out of
60 covered months.
These classifications describe limits of the numerical
feasibility search, not low statistical test power.

\end{document}